\documentclass[preprint]{elsarticle}

\usepackage{amsmath,amssymb}
\usepackage{amsthm}
\usepackage{latexsym}
\usepackage{times}
\usepackage{graphicx}
\usepackage{url}
\usepackage{setspace}
\usepackage{color}
\usepackage{acronym}

\usepackage{algorithm}
\usepackage{subfigure}

\newcommand{\Ric}{\textup{Ric}}

\newcommand{\eK}{K^{(e)}}

\newcommand{\ud}{\textup{d}}

\newcommand{\NN}{\mathbb{N}}

\newcommand{\MM}{\mathcal{M}}
\newcommand{\RR}{\mathbb{R}}

\newtheorem{theorem}{Theorem}[section]
\newtheorem{defn}[theorem]{Definition}

\newtheorem{cor}[theorem]{Corollary}
\newtheorem{thm}{Theorem}[section]
\newtheorem{rem}{Remark}[section]

\newtheorem{lemma}{Lemma}

\theoremstyle{remark}

\newtheorem{assumption}[thm]{Assumption}

\acrodef{AD}{alternating diffusion}
\acrodef{SVD}{singular value decomposition}
\acrodef{EVD}{eigenvalue decomposition}

\begin{document}

\begin{frontmatter}
\title{Latent common manifold learning with alternating diffusion: analysis and applications}

\author[RT]{Ronen~Talmon}
\ead{ronen@ef.technion.ac.il}
\author[HTW,HTW2]{Hau-Tieng~Wu}
\ead{hauwu@math.duke.edu}

\address[RT]{Viterbi Faculty of Electrical Engineering, Technion - Israel Institute of Technology, Haifa, Israel}
\address[HTW]{Department of Mathematics and Department of Statistical Science, Duke University, Durham, NC, USA}
\address[HTW2]{Mathematics Division, National Center for Theoretical Sciences, Taipei, Taiwan}


\begin{abstract}

The analysis of data sets arising from multiple sensors has drawn significant research attention over the years. Traditional methods, including kernel-based methods, are typically incapable of capturing nonlinear geometric structures. 
We introduce a latent common manifold model underlying multiple sensor observations for the purpose of multimodal data fusion. A method based on alternating diffusion is presented and analyzed; we provide theoretical analysis of the method under the latent common manifold model. To exemplify the power of the proposed framework, experimental results in several applications are reported.

\end{abstract}


\end{frontmatter}

\section{Introduction}
One of the long-standing challenges in signal processing is the fusion of information acquired by multiple, multimodal sensors. The problem of information fusion has become particularly central in the wake of recent technological advances, which have led to extensive collection and storage of multimodal data. Nowadays, many devices and systems, e.g., cell-phones, laptops, and wearable-devices, incorporate more than one sensor, often of different types. Of particular interest in the context of this paper are the massive data sets of medical recordings and healthcare-related information, acquired routinely, for example, in operation rooms, intensive care units, and clinics. The availability of such distinct and complementary information calls for the development of new theories and methods, leveraging it toward achieving concrete data analysis objectives, such as filtering and prediction, in a broad range of fields.

Problems in multimodal signal processing has been studied for many years and has been approached from various
research directions \cite{lahat2015multimodal}.  A classic approach for such problems is Canonical Correlation Analysis (CCA) \cite{hotelling_relations_1936}, which recovers highly correlated linear projections from two data sets. To extend the linear setting and to address aspects of nonlinearities, CCA was applied in a kernel space (e.g., \cite{lai2000kernel,bach2003kernel}). Recently, ample work based on the optimization criterion of CCA and kernels has been presented, addressing multi-view problems, and in particular, using multi-kernel learning (e.g. \cite{bach2004multiple,lanckriet2004learning}) and a variety of manipulations and combinations of kernels, e.g. \cite{boots_two_manifold_2012,de_sa_multi_view_2010,de_sa_spectral_2005,luo_mixed_2012,luo_shape_2013,zhou2012fusion,kumar_cotraining_2011,xiang_bai_cotransduction_2012,hsinchien_huang_affinity_2012,wu_unsupervised_2013,bo_wang_unsupervised_2012,lindenbaum2015learning,lindenbaum2015multiview,eynard2012multimodal,michaeli2015nonparametric}. 

Our exposition begins by addressing a particular baseline problem. Consider multiple sensors measuring the same physical phenomenon, such that the properties of the physical phenomenon are manifested as a hidden manifold (which we would like to extract), while each sensor presents its own deformation and has its own sensor-specific effects (hidden {\em nuisance} variables, which we would like to suppress). We assume that the relations between the measurements and the {\em nuisance} variables are unknown. The goal is to uncover the common latent manifold and to suppress the sensor-specific variables, thereby extracting the essence of the data and separating the relevant information from the irrelevant information.

This baseline problem highlights an important aspect in the analysis of multimodal data sets; that is, the sensor-related variables may not be strictly related to noise and interferences. Often, such variables exhibit ``structures'', such as the position and orientation of the acquiring sensor, environmental effects, and channel characteristics. To address this, we propose an approach based on manifold learning. The power of manifold learning can be exploited in this setting, since it is designed to capture nonlinear topological and geometric structures underlying data, and it does not require prior model knowledge, which can be particularly hard to obtain in the case of multiple modalities. Manifold learning is a class of nonlinear data-driven methods, e.g. \cite{Tenenbaum2000,Roweis2000,Donoho2003,Belkin_Niyogi:2003}, often used to extract the underlying structures in a given data set. Of particular interest in the context of this paper is diffusion map (DM), \cite{Coifman_Lafon:2006}, in which discrete diffusion processes are constructed on the given data points; these diffusion processes are designed to capture the geometry of the underlying variability in a single data set. 
Multimodal data present a challenge to such a geometric analysis approach, since multiple sensors often lead to undesired geometric structures stemming from the diversity of the different acquisition techniques used in the sensors, making it more difficult to identify and extract only the ``important" variables. Nevertheless, multiple data sets from various sensors encompass more information, and therefore, enable us to recover a more reliable description of the measured (physical) phenomenon. Based on manifold learning, several methods have been proposed to analyze simultaneously multiple data sets. One approach is to concatenate the vectors representing each data set into one vector \cite{davenport2010joint}; however, the question of how each data set should be scaled and concatenated naturally arises, especially if the data sets are acquired by very different modalities. To address such scaling aspects, it has been proposed in \cite{keller_audio_visual_2010} to use DM to obtain a low-dimensional ``standardized" representation of each data set, and then to concatenate these representations. However, such methods aggregate all the variables from all the given data sets, and they neither distinguish the important information nor discard the sensor-specific variables.

Our research direction involving geometric analysis encompasses several significant advantages. First, the method we present is data-driven and ``model-free'' in the sense that in addition to the manifold assumption, it does not rely on prior knowledge. In multimodal problems, this is an important advantage, since it circumvents the need to design an appropriate model for each modality, as well as the ``hard wiring" required for the fusion of different data sets. Second, manifold learning methods are typically formulated in general settings, and therefore, do not require strong assumptions on the nature of the data or on the nature of the sensors. As a result, our method is restricted to neither certain applications nor to multi-view problems, consisting of data acquired only by a single type of sensors. Third, the combination of geometric analysis, enabling to integrate subtle patterns and structures underlying data, and the availability of multiple data sets providing complementary information, gives rise to the discovery of intrinsic structures. Furthermore, the diffusion approach has been shown to be robust to noise \cite{ElKaroui:2010a,ElKaroui_Wu:2014}, and hence allows us to devise a reliable approach for extracting the interesting information from highly noisy data.

Recently, a data-driven method to recover the common latent variable underlying multiple, multimodal sensor data based on alternating products of diffusion operators was presented \cite{lederman2015alternating,lederman2015icassp}; we refer to this method as \ac{AD}. 
It was shown both theoretically and in illustrative examples (with real recordings) that this method extracts useful information about the common source of variability in a multi-sensor experiment as if it were a typically diffusion operator applied directly to data sampled from only the interesting/common source of variability.
The formulation and analysis in \cite{lederman2015alternating} are based on a setting including only metric spaces and the common source of variability is identical in the different sensors. 

In the current work, we extend \cite{lederman2015alternating} and enhance the theoretical results by explicitly introducing a setting with \emph{a common manifold}, which can be well accessed by different sensors, while different types of deformations are introduced due to various effects of the specific sensors.
We focus on the case in which there is a common manifold underlying all the given data sets, and propose a data-driven method for recovering the common manifold and constructing its representation. More specifically, in an analogous way to the classic diffusion geometry approach \cite{Coifman_Lafon:2006}, we show that by using a product of diffusion operators in DM, we are able to approximate a modified/deformed Laplace operator on the underlying common manifold.


Our main contribution in this paper can be summarized as follows. (i) While most current manifold learning methods only address a single data set arising from a single manifold, we extend this basic setting and present a technique addressing several data sets arising from multiple manifolds. (ii) We introduce a new concept of ``nonlinear manifold filtering''{, that is, removing the influence of the nuisance variables,} with a rigorous theoretical foundations and analysis. (iii) The ability to extract the common manifold underlying several data sets enables us to propose new methods for the multimodal sensor fusion problem. These methods consist of a scheme to incorporate nonlinear geometric priors into a manifold filtering procedure and are demonstrated on illustrative (real) examples. Specifically, we show application to the analysis of sleep dynamics as well as to seasonal pattern detection (a problem which is commonly encountered in time series analysis and will be introduced in the sequel). 

The remainder of this paper is organized as follows. In Section \ref{Section:CommonManifoldAlternationDiffusion}, the common manifold model for multiple data sets collected from multimodal sensors is introduced. With the common manifold model, the \ac{AD} algorithm is formulated and its analysis is provided. To further study the behavior of \ac{AD} in the common manifold model, in Section \ref{Section:ADAnalysis}, asymptotical analysis results are provided. More details about the \ac{AD} algorithm for practical purposes are discussed in {Section \ref{Section:ADalgorithm}}. In Section \ref{Section:Seasonality}, we exploit the common manifold model and design an \ac{AD} algorithm to detect seasonal patterns in time series. We provide a seasonality index to quantify the seasonal effects. In Section \ref{Section:Sleep}, a sleep data set is studied, demonstrating the power of the common manifold model and the \ac{AD} algorithm in the medical field. We will show that among different pairs of sensors, different physiological information is obtained. In Section \ref{Section:Extension}, the prototypical extension problem of manifold representations is considered and discussed for the common manifold model and the \ac{AD} algorithm. Discussion and future directions are outlined in Section \ref{Section:Conclusion}. The proofs of the theoretical results from Section \ref{Section:ADAnalysis} are presented in the Appendix.

\section{Common Manifold Model and Alternating Diffusion}
\label{Section:CommonManifoldAlternationDiffusion}

The success of the \ac{AD} algorithm has been shown in different problems, for example, for sleep depth analysis \cite{lederman2015icassp}. 
To analyze the algorithm, in \cite{lederman2015alternating}, the common variable as well as the nuisance variables specific for each sensor are modeled by metric measure spaces{; specifically, it is shown that the diffusion distance calculated from the data via \ac{AD} is equivalent to an ``effective alternating-diffusion distance'', which is defined only on one common variable shared by the two sensors \cite[Theorem 5 and Equation 79]{lederman2015alternating}}. However, in some problems, {the data may exhibit additional structures, which could be extracted and exploited.} In this paper, we consider such a case where the common variable has a distinct geometric structure, which is modelled by a manifold. Yet, even if we assume that the system we observe remains fixed during the observation, the data collected from different sensors might depend on the observation procedure or be contaminated by different irrelevant information. {In particular, the common manifold might be deformed differently by different sensors}. As a result, special focus is given to possible dependencies between the ``common manifold'' and the sensors. These considerations are manifested in the setting presented in this section, and our focus in the remainder of the paper is on studying the geometric information that can be obtained from \ac{AD}.

\begin{figure}[th]
\begin{centering} 
\includegraphics[width=.8\textwidth]{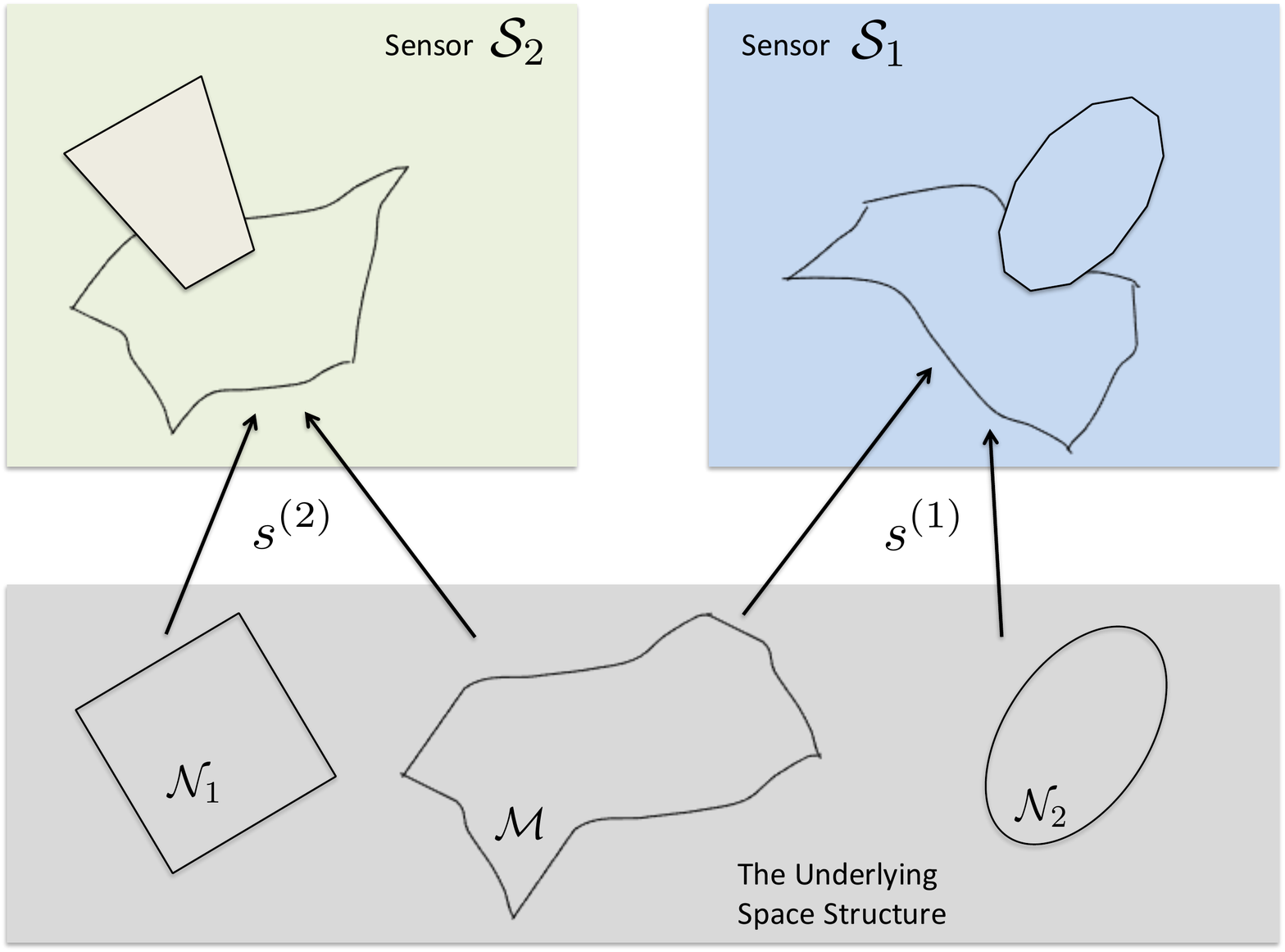}
\caption{A diagram illustrating our setting with a common manifold underlying two sensor observations.}
\label{fig:diagram}
\end{centering} 
\end{figure}

\subsection{Common manifold model}

Consider a setup consisting of multiple sensors observing a system or a phenomenon of interest {\em simultaneously}. To simplify the exposition, we focus on a setup with merely two sensors, noting that our analysis can be generalized to multiple sensors with slight modifications.

\subsubsection{Geometric model}

Suppose there exists a common structure underlying the two sensor observations, which is modeled by a low dimensional {\em common manifold}, and suppose that each sensor introduces various deformations and interferences, which are modeled by other irrelevant nuisance structures.
Mathematically, denote the common structure of interest, $\MM$, by a $d$-dim Riemannian manifold with the metric tensor $g^{(i)}$, which depends on the sensor, where $i=1,2$ is the sensor index, and let $d_{g^{(i)}}$ denote the distance function on $\MM$ associated with $g^{(i)}$. To accommodate possible discrepancies and deformations between the common manifold observed through the two sensors, we allow the metrics to be different, i.e., $g^{(1)}\neq g^{(2)}$.
For $i=1,2$, let $\mathcal{N}_i$ be a compact metric space with the distance function $d_{\mathcal{N}_i}$, which models the irrelevant nuisance structures. 
Let $\mathcal{S}_i$ denote the \emph{observable} metric space with the distance function $d_{\mathcal{S}_i}$, which models the space of the collected/accessible data by the $i$-th sensor.
Assume that 
\begin{align}\label{eq:sensors}
s^{(i)}:\MM\times\mathcal{N}_1\times \mathcal{N}_2\to \mathcal{S}_i
\end{align} 
is a smooth isometric embedding of $\MM\times\mathcal{N}_i$ into $\mathcal{S}_i$ modeling how the $i$-th sensor collects data, where $i=1,2$. Importantly, note that $s^{(1)}$ ignores $\mathcal{N}_2$ (as it represents the interferences specific to the the second sensor), and that $s^{(2)}$ ignores $\mathcal{N}_1$ (as it represents the interferences specific to the first sensor). 
In other words, each sensor typically acquires two structures: a deformation of the common manifold of interest and an additional nuisance structure. The model described here implies that we do not have access to the structures underlying each sensor (the deformed common manifold $\mathcal{M}_i$ and the nuisance structure $\mathcal{N}_i$), nor to the mapping to the observable sensor space \eqref{eq:sensors}.
In addition, \eqref{eq:sensors} means that for $(x,y,z), (x',y',z')\in \MM\times\mathcal{N}_1\times\mathcal{N}_2$, where $x,x'\in \mathcal{M}$, $y,y'\in\mathcal{N}_1$ and $z,z'\in\mathcal{N}_2$, we have 
\begin{align}
&d_{\mathcal{S}_1}(s^{(1)}(x,y,z),s^{(1)}(x',y',z'))^2=d_{g^{(1)}}(x,x')^2+d_{\mathcal{N}_1}(y,y')^2\label{Equation:MetricDefinition:RelationshipExact}\\
&d_{\mathcal{S}_2}(s^{(2)}(x,y,z),s^{(2)}(x',y',z'))^2=d_{g^{(2)}}(x,x')^2+d_{\mathcal{N}_2}(z,z')^2.\nonumber
\end{align}
In practice, we do not have the full access to the metric structure \eqref{Equation:MetricDefinition:RelationshipExact}, but only to the observable distance functions $d_{\mathcal{S}_i}$.

In summary, different sensors have access to the information of interest from the common geometric object $\MM$.
However, the acquired information is deformed by two different sources -- the nuisance variables modeled by $\mathcal{N}_i$, and the deformation induced by each sensor modeled by the different metric on the manifold $\MM$. 
 
A diagram of the geometric model is depicted in Figure \ref{fig:diagram}.

\subsubsection{Statistical model}

To model the discrete data sets collected by the sensors, let $(\Omega,\mathcal{F},P)$ be a probability space, where $\Omega$ is the event space, $\mathcal{F}$ is the sigma algebra on $\Omega$, and $P$ is a probability measure defined on $\mathcal{F}$. Consider a random vector $S:(\Omega,\mathcal{F},P)\to \MM\times \mathcal{N}_1\times \mathcal{N}_2$. The dataset sampled from the $i$-th sensor, $i=1,2$, is modeled by a random vector $S_i=s^{(i)}\circ S$. Note that via the first sensor, only $\MM$ and $\mathcal{N}_1$ are observed in $S_1$, while via the second sensor, only $\MM$ and $\mathcal{N}_2$ are observed in $S_2$.

We further assume that conditional on $\MM$, the nuisance variables modeled by $\mathcal{N}_1$ and $\mathcal{N}_2$ (introduced by the different sensors) are sampled independently. Let $\nu:=\nu_{\MM\times\mathcal{N}_1\times\mathcal{N}_2}={S}_*P$ denote the induced probability measure on $\MM\times \mathcal{N}_1\times\mathcal{N}_2$. By the assumption that conditional on $\MM$ the nuisance variables are independent, we have 
\begin{align}
\nu(x,y,z)=\nu_{\MM}(x)\nu_{\mathcal{N}_1|\MM}(y|x)\nu_{\mathcal{N}_2|\MM}(z|x), 
\end{align}
where $\nu_{\MM}(x)$ is the marginal distribution on $\MM$, and $\nu_{\mathcal{N}_i|\MM}(\cdot |x)$ is the conditional distribution on $\mathcal{N}_i$.

To enhance the readability of the paper, we summarize the notations in Table \ref{Table:Notation}.

\begin{table}[h] 
\begin{center}
\begin{tabular}{| l | l | }
\hline
Symbol ($i=1,2$) & Meaning\\
\hline\hline
$(\MM,g^{(i)})$ & $d$-dim Riemannian manifold with the {Riemannian metric} $g^{(i)}$\\
$(\mathcal{N}_i,d_{\mathcal{N}_i})$ & Compact metric space with the metric $d_{\mathcal{N}_i}$\\
$(\mathcal{S}_i,d_{\mathcal{S}_i})$ & Metric space with the metric $d_{\mathcal{S}_i}$\\
$s^{(i)}$ & Sensor collecting data \\
$\epsilon$ & Bandwidth parameter of the kernel function\\
$D^{(i)}$ & Observable diffusion operator on $\mathcal{S}_i$\\
$\tilde{P}^{(i)}$ & Kernel associated with $D^{(i)}$\\
$D$ & Observable AD operator starting from the first sensor\\
$\mathcal{E}$ & Marginalization operator \\
$D^{(e_i)}$ & Effective diffusion operator on $(\MM,g^{(i)})$\\
$\tilde{P}^{(e)}$ & Kernel associated with $D^{(e_i)}$\\
$D^{(e)}$ & Effective AD operator starting from the first sensor \\
$\tilde{P}^{(O_i,y)}$ & Observable diffusion kernel with the fixed $y$\\
\hline\hline
$\iota^{(i)}$ & Embed $\MM$ into $\RR^p$ so that $g^{(i)}$ is the induced metric via $\iota^{(i)}$\\
$\tilde{K}^{(i)}$ & Diffusion kernel on $(\MM,g^{(i)})$\\
$T$ & Reduced AD operator starting from $(\MM,g^{(1)})$\\
$\tilde{K}^{(e)}$ & Kernel associated with $T$\\
$\exp^{(i)}_{x}$ & Exponential map at $x$\\
$\nabla^{(i)}$ & Levi-Civita connection associated with the metric $g^{(i)}$\\
$\ud V^{(i)}$ & Volume form associated with the metric $g^{(i)}$\\
$\Ric^{(i)}$ & Ricci curvature of $(\MM,g^{(i)})$\\
$s^{(i)}$ & Scalar curvature of $(\MM,g^{(i)})$\\ 
$\Pi^{(i)}$ & second fundamental form of the embedding $\iota^{(i)}$\\ 
$\Delta^{(i)}$ & Laplace-Beltrami operator of $(\MM,g^{(i)})$\\
\hline
\end{tabular}
\end{center}
\caption{Summary of symbols used throughout the paper.}\label{Table:Notation}
\end{table}

\subsection{Alternating diffusion {under the common manifold model}}

Based on the common manifold model, we apply \ac{AD} to analyze data collected simultaneously from two sensors, possibly of different modalities. {The goal is to extract the common manifold from the observed data via AD.} {Throughout the paper, we use three sets of notations. The first set will be defined on the accessible data, typically using an observable kernel and an observable diffusion. The second set is intermediate and designed to describe the relationship between the observation and the hidden common manifold. The third set will be defined on the hidden common manifold using an (inaccessible) effective kernel and effective diffusion.}

\begin{defn}[Observable Diffusion Kernels]
Let $\tilde{P}^{(i)}\in C([0,\infty))$, $i=1,2$, be two kernels that decay sufficiently fast and are associated with the two sensors. 
Define
\begin{align}
&\tilde{P}^{(i)}_{\epsilon}((x,y,z),(z',y',z')):=\,\tilde{P}^{(i)}\Big(\frac{d_{\mathcal{S}^{(i)}}(s^{(i)}(x,y,z),s^{(i)}(x',y',z'))}{\sqrt{\epsilon}}\Big)\label{notation:Kh_def}\\
&P^{(i)}_{\epsilon}((x,y,z),(x',y',z')):=\,\frac{\tilde{P}^{(i)}_{\epsilon}((x,y,z),(x',y',z'))}{\int_{\MM\times\mathcal{N}_1\times\mathcal{N}_2} \tilde{P}^{(i)}_{\epsilon}((x,y,z),(x'',y'',z''))\ud\nu(x'',y'',z'')}.\nonumber
\end{align}
where $P^{(i)}_\epsilon$ is referred to as a {\em diffusion kernel with bandwidth $\epsilon$} associated with the observable metric measure space $\mathcal{S}_i = s^{(i)}(\MM\times\mathcal{N}_1\times\mathcal{N}_2)$. 
\end{defn}
Note that the diffusion kernel $P^{(i)}_\epsilon$ is in a normalized form, i.e.,
\[
	\int_{\MM\times\mathcal{N}_1\times\mathcal{N}_2} P^{(i)}_{\epsilon}((x,y,z),(x',y',z'))\ud\nu(x',y',z') = 1.
\]
The main benefit of a diffusion kernel in such a normalized form is that the normalization helps to eliminate unwanted non-intrinsic quantities, such as the terms depending on the specific kernel in the asymptotical analysis. See, for example, the difference between Lemma 8 and Proposition 10 in \cite{Coifman_Lafon:2006}\footnote{{Note that in \cite{lederman2015alternating} the analysis is carried out using forward diffusion, while in this paper we use backward diffusion for consistency with the standard diffusion maps framework presented in \cite{Coifman_Lafon:2006}.}}.
For each kernel, $\epsilon$ is referred to as the {\em bandwidth} of the kernel. Further note that while the bandwidth may depend on the sensor, for simplicity, we assume that the kernels share the same $\epsilon$ value.

In this section, we show that when studied under a suitable assumption, the influence of the nuisance variables, that is, each metric space $\mathcal{N}_i$ representing the sensor effects and observation specific influences, is erased by the \ac{AD} procedure, and thus, can be ignored. Note that in the special case when $g^{(1)}=g^{(2)}$, i.e., the common manifold is viewed similarly by different sensors, 
this statement was shown previously in \cite{lederman2015alternating}{, where the focus is on the diffusion distance}. 

%

We begin with the following definition.
\begin{defn}[{Observable} Diffusion Operators]
For a function $f\in C(\MM\times\mathcal{N}_1\times\mathcal{N}_2)$, let $D^{(i)}: C(\MM\times\mathcal{N}_1\times\mathcal{N}_2)\to C(\MM\times\mathcal{N}_1\times\mathcal{N}_2)$, for $i=1,2$, denote the diffusion operator on the $i$-th sensor, defined by
\begin{align}
D^{(i)}f(x,y,z) 
:=\int_{\MM\times\mathcal{N}_1\times\mathcal{N}_2} P_\epsilon^{(i)}((x,y,z),(x',y',z'))f(x',y',z') \ud \nu(x',y',z').
\end{align}
\end{defn}

\begin{defn}[{Observable} Alternating Diffusion Operator and Kernel]
The \ac{AD} operator starting from the first sensor is defined by
\begin{align}
D:=D^{(2)}D^{(1)},
\label{eq:AD_operator}
\end{align}
while the \ac{AD} operator starting from the second sensor is defined analogously. 
%
Let $P_\epsilon\in C((\MM\times\mathcal{N}_1\times\mathcal{N}_2)\times(\MM\times\mathcal{N}_1\times\mathcal{N}_2))$ be the {\em \ac{AD} kernel}, defined by 
\begin{align}
&P_\epsilon((x,y,z),(x'',y'',z''))\label{Definition:Pepsilon}\\
:=&\,\int_{\MM\times\mathcal{N}_1\times\mathcal{N}_2}P_\epsilon^{(2)}((x,y,z),(x',y',z'))P_\epsilon^{(1)}((x',y',z'),(x'',y'',z'')) \ud \nu(x',y',z').\nonumber
\end{align}
\end{defn}
Note that by definition, we can associate the {observable} \ac{AD} kernel and operator:
\begin{align}
(Df)(x,y,z)=\int_{\MM\times\mathcal{N}_1\times\mathcal{N}_2} P_\epsilon((x,y,z),(x'',y'',z'')) f(x'',y'',z'') \ud \nu(x'',y'',z'').
\end{align}


{The second set of notations concerns with the connection between the observations and the hidden common manifold.} Consider the {observable} \ac{AD} starting from the first sensor, and recall that the first sensor only sees $\MM\times \mathcal{N}_1$. Thus, by the {definition of $\tilde{P}^{(1)}$ in \eqref{Equation:MetricDefinition:RelationshipExact} and \eqref{notation:Kh_def}}, $P_\epsilon^{(1)}((x,y,z),(x',y',z'))$ only takes $(x,y)$ and $(x',y')$ into account. {We thus introduce the corresponding reduced intermediate diffusion kernels:
\begin{align}
\tilde{P}_\epsilon^{(I_1)}((x',y'),(x'',y'')) &:=\tilde{P}_\epsilon^{(1)}((x',y',z'),(x'',y'',z''))\nonumber\\
P_\epsilon^{(I_1)}((x',y'),(x'',y'')) &:=P_\epsilon^{(1)}((x',y',z'),(x'',y'',z''))\,;
\end{align} 
that is, the observable diffusion starting from the first sensor can be simplified by ignoring the contribution of $\mathcal{N}_2$. Similarly, we define
\begin{align}
\tilde{P}_\epsilon^{(I_2)}((x',z'),(x'',z'')) &:=\tilde{P}_\epsilon^{(2)}((x',y',z'),(x'',y'',z''))\nonumber\\
P_\epsilon^{(I_2)}((x',z'),(x'',z'')) &:=P_\epsilon^{(2)}((x',y',z'),(x'',y'',z''))\,.
\end{align}
Clearly, by defintion, for $f\in C(\MM\times\mathcal{N}_1\times\mathcal{N}_2)$, $D^{(1)}f(x,y,\cdot)$ is a constant function for a fixed $(x,y)$. %
Accordingly, define the {\em reduced intermediate function}:
\begin{align}
f^{(I_2)}(x,y) :=f(x,y,\cdot)\,,
\end{align}
which is used to describe the result of the application of $D^{(1)}$.
A similar argument implies that $D^{(2)}f(x,\cdot,z)$ is a constant function for fixed $(x,z)$, and hence the definition:
\begin{align}
f^{(I_1)}(x,z) :=f(x,\cdot,z)
\end{align}

Lastly, we define the third set of notations, which consists of the {\em effective counterparts} of the observable diffusion operators defined on the hidden common manifold. These notations are needed to show that in effect the AD erases the nuisance variables. We need the following auxiliary operator that integrates out the nuisance variables.
\begin{defn}[Marginalization Operator]
Let $\mathcal{E}:C(\MM\times\mathcal{N}_1\times\mathcal{N}_2)\to C(\MM)$ be the marginalization operator, which is defined by
\begin{align}
\mathcal{E}f(x):=\int_{\mathcal{N}_1\times\mathcal{N}_2}f(x,y,z)\ud\nu_{\mathcal{N}_1|\MM}(y|x)\ud\nu_{\mathcal{N}_2|\MM}(z|x).
\label{eq:marginalization_operator}
\end{align}
\end{defn}
Clearly, the marginalization operator $\mathcal{E}$ is a bounded linear operator, which evaluates the marginal distribution of the collected data.
Although a-priori the marginalization operator is unknown and cannot be computed from data (since we do not have access to the hidden structure of the data, an in particular, to the nuisance variables), we will show that {an equivalent operation is attainable by the observable \ac{AD}}.

With the above definition, when $f(x,\cdot,z)$ is a constant function for fixed $(x,z)$, \eqref{eq:marginalization_operator} is reduced to
\begin{align}
\mathcal{E}f(x)=\int_{\mathcal{N}_2}f^{(I_1)}(x,z)\ud\nu_{\mathcal{N}_2|\MM}(z|x).
\end{align}
Similarly, when $f(x,y,\cdot)$ is a constant function for fixed $(x,y)$, \eqref{eq:marginalization_operator} is reduced to
\begin{align}
\mathcal{E}f(x)=\int_{\mathcal{N}_1}f^{(I_2)}(x,y)\ud\nu_{\mathcal{N}_1|\MM}(y|x).
\end{align}}

\begin{defn}[Effective Diffusion Kernel and Operator]
Let $P^{(e_i)}_\epsilon(x,x')$ be the effective diffusion kernel associated with the $i$-th sensor and defined on $\MM$ by
\begin{align}
P^{(e_i)}_\epsilon(x,x'):=\int_{\mathcal{N}_i}\int_{\mathcal{N}_i}P^{(I_i)}_\epsilon((x,y),(x',y'))\ud\nu_{\mathcal{N}_i|\MM}(y|x)\ud\nu_{\mathcal{N}_i|\MM}(y'|x').
\label{eq:effective_diffusion_kernel}
\end{align}
and let $D^{(e_i)}$ be the corresponding effective diffusion operator, defined by 
\begin{align}
D^{(e_i)}\mathcal{E}f(x):=\int_{\MM} P^{(e_i)}_\epsilon(x,x')\mathcal{E}f(x')\ud\nu_{\MM}(x')\,.
\end{align}
\end{defn}
\begin{defn}[Effective Alternating Diffusion Kernel and Operator]
Let $P^{(e)}_\epsilon$ be the effective \ac{AD} kernel associated with \ac{AD} starting from the first sensor, defined by
\begin{align}
P^{(e)}_\epsilon(x,x''):=\int_{\MM} P^{(e_2)}_\epsilon(x,x') P^{(e_1)}_\epsilon(x',x'')\ud\nu_{\MM}(x').
\end{align}
and let $D^{(e)}$ be the corresponding effective \ac{AD} operator, defined by 
\begin{align}
D^{(e)}:=D^{(e_2)}D^{(e_1)}.\label{Definition:EffectiveADFromTheFirstSensor}
\end{align}
\end{defn}
{By definition we have
\begin{align}
(D^{(e)}\mathcal{E}f)(x)= \int_{\MM} P^{(e)}_\epsilon(x,x'') \mathcal{E}f(x'')\ud\nu_{\MM}(x'')\,.
\end{align}
The expansion of the {\em effective} diffusion kernel \eqref{eq:effective_diffusion_kernel} according to the definition of the diffusion kernel in \eqref{notation:Kh_def} deserves an additional discussion. Note that}
\begin{align}
P^{(e_1)}_\epsilon(x,x')&=\mathcal{E}\left[\frac{\int_{\mathcal{N}_1}\tilde{P}_\epsilon^{(1)}\Big(\frac{\sqrt{d_{g^{(1)}}(x,x')^2+d_{\mathcal{N}_1}(y,y')^2}}{\sqrt{\epsilon}}\Big)\ud\nu_{\mathcal{N}_1|\MM}(y'|x')}{\int_{\MM\times\mathcal{N}_1}\tilde{P}_\epsilon^{(1)}\Big(\frac{\sqrt{d_{g^{(1)}}(x,x'')^2+d_{\mathcal{N}_1}(y,y'')^2}}{\sqrt{\epsilon}}\Big)\ud \nu_{\MM}(x'')\ud\nu_{\mathcal{N}_1|\MM}(y''|x'') }\right] \label{Description:Pe1IsNotNormalized}.
\end{align}
Thus, $P^{(e_1)}_\epsilon(x,x')$ is a diffusion kernel on $\MM$; yet, for a general kernel function $\tilde{P}^{(1)}_{\epsilon}$, the kernel $P^{(e_1)}_\epsilon(x,x')$ cannot be further simplified. 

\begin{rem}
In the special case when $\tilde{P}^{(1)}_{\epsilon}$ is Gaussian {and the nuisance variables are independent}, it is possible to separate it into two terms consisting of the metric on the common manifold $d_{g^{(i)}}$ and the metric on the nuisance variable $d_{\mathcal{N}_i}$. 
In other words, in this special case, we have further access to the hidden structure of the data, and in particular, to the metric defined on each component.
Accordingly, the effective kernel can be simplified into the following normalized form
\begin{align}
P^{(e_i)}_\epsilon(x,x')=\frac{\tilde{P}^{(e_i)}\Big(\frac{d_{g_i}(x,x')}{\sqrt{\epsilon}}\Big)}{\int_{\MM}\tilde{P}^{(e_i)}\Big(\frac{d_{g_i}(x,x'')}{\sqrt{\epsilon}}\Big)\ud\nu_{\MM}(x'')}.
\end{align}
\end{rem}

With the above preparation, we are ready to state the main result of this section, showing that after marginalization, \ac{AD} constructed from the two-sensor obserations is {intimately related} to a diffusion process defined on the hidden common manifold.
\begin{thm}\label{Theorem:Commutation}
{For a fixed $i=1,2$,} when $f\in C(\MM\times \mathcal{N}_1\times\mathcal{N}_2)$ is constant on $\mathcal{N}_i$, we have 
\begin{align}
\mathcal{E}[D^{(i)}f]= D^{(e_i)}\mathcal{E}f\,.
\end{align}
Furthermore, we have
\begin{align}
\mathcal{E}[Df]=D^{(e)}\mathcal{E}f.
\label{eq:com}
\end{align}
\end{thm}

Before the proof, note that {Theorem \ref{Theorem:Commutation} implies that} the marginalization operator and the \ac{AD} operator commute.
{While the effective \ac{AD} operator $D^{(e)}$ cannot be directly computed from data, the observable \ac{AD} can be. However, in practice $\mathcal{E}$ is unknown, and we need to link $Df$ back to the effective \ac{AD} operator on the common manifold. 
We will further address this issue} in the sequel.

\begin{proof}
{Fix $i=1$. The proof for $i=2$ is analogous. By assumption, $f(x,\cdot,z)$ is a constant function for fixed $(x,z)$\footnotemark, so that we have}
\footnotetext{Such a function could be obtained as the result of applying $D^{(2)}$.}
\begin{align}
&D^{(1)}f(x,y,z)\label{Proof:Theorem1:FirstEquation}\\
=&\,\int_{\MM\times\mathcal{N}_1\times\mathcal{N}_2} P_\epsilon^{(1)}((x,y,z),(x'',y'',z''))f(x'',y'',z'')\ud \nu(x'',y'',z'')\nonumber\\
=&\,\int_{\MM\times\mathcal{N}_1\times\mathcal{N}_2} P_\epsilon^{(I_1)}((x,y),(x'',y''))f^{(I_1)}(x'',z'') \ud\nu(x'',y'',z'')\nonumber\\
=&\,\int_{\MM\times\mathcal{N}_1} P_\epsilon^{(I_1)}((x,y),(x'',y''))\mathcal{E}f(x'') \ud\nu_{\MM}(x'')\ud\nu_{\mathcal{N}_1|\MM}(y''|x'').\nonumber
\end{align}
Note that $D^{(1)}f(x,y,\cdot)$ is a constant function for {a fixed pair} $(x,y)$. Thus, we have
\begin{align}
&\mathcal{E}[D^{(1)}f](x)\\
=\,&\int_{\mathcal{N}_1\times\mathcal{N}_2} (D^{(1)}f)(x,y,z)\ud\nu_{\mathcal{N}_1|\MM}(y|x)\ud\nu_{\mathcal{N}_2|\MM}(z|x)\nonumber\\
=\,&\int_{\mathcal{N}_1\times\mathcal{N}_2} \int_{\MM\times\mathcal{N}_1} P_\epsilon^{(I_1)}((x,y),(x'',y''))\mathcal{E}f(x'')\nonumber\\
&\qquad\qquad\qquad\qquad\ud\nu_{\MM}(x'')\ud\nu_{\mathcal{N}_1|\MM}(y''|x'')\ud\nu_{\mathcal{N}_1|\MM}(y|x)\ud\nu_{\mathcal{N}_2|\MM}(z|x)\nonumber\\
=\,&\int_{\mathcal{N}_1} \int_{\MM\times\mathcal{N}_1} P_\epsilon^{(I_1)}((x,y),(x'',y''))\mathcal{E}f(x'')\ud\nu_{\MM}(x'')\ud\nu_{\mathcal{N}_1|\MM}(y''|x'')\ud\nu_{\mathcal{N}_1|\MM}(y|x)\nonumber\\
=\,&\int_{\MM} P^{(e_1)}_\epsilon(x,x'')\mathcal{E}f(x'')\ud\nu_{\MM}(x'').\nonumber
\end{align}
Define $D^{(e_1)}\mathcal{E}f(x):=\mathcal{E}[D^{(1)}f](x)$.
Similarly, if $f(x,y,\cdot)$ is a constant function for fixed $(x,y)$, the effective diffusion associated with the second sensor is given by
\begin{align}
\mathcal{E}[D^{(2)}f](x)=\int_{\MM} P^{(e_2)}_\epsilon(x,x')\mathcal{E}f(x')\ud\nu_{\MM}(x')
\end{align}
and define $D^{(e_2)}\mathcal{E}f(x):=\mathcal{E}[D^{(2)}f](x)$.
To finish the proof, note that
\begin{align}
\mathcal{E}[Df]=\mathcal{E}D^{(2)}D^{(1)}f=D^{(e_2)}\mathcal{E}D^{(1)}f=D^{(e_2)}D^{(e_1)}\mathcal{E}f=D^{(e)}\mathcal{E}f.
\end{align}
\end{proof}

{In light of Theorem \ref{Theorem:Commutation}, we can remark on the spectral behavior of the observable and effective \ac{AD} operators. Note that when $\epsilon>0$ is finite, although $D$ and $D^{(e)}$ are compact operators, they are not self-adjoint and hence limited spectral information can be exploited. In the sequel, we will show that asymptotically, when $\epsilon$ approaches $0$, these operators approximate a deformed Laplace-Beltrami operator. In particular, such a limiting operator is self-adjoint, and therefore, at least asymptotically, it gives the theoretical foundation to consider the spectral analysis of the operators at hand.} 
Suppose we have $D\phi=\lambda \phi$; that is, $\phi$ is the eigenfunction of {the observable \ac{AD} operator} $D$ associated with the eigenvalue $\lambda$. Then, by the commutativity {shown in Theorem \ref{Theorem:Commutation}} and the linearity of $\mathcal{E}$, we have
\begin{align}
D^{(e)}\mathcal{E}\phi=\mathcal{E}D\phi=\mathcal{E}\lambda \phi=\lambda \mathcal{E}\phi\,;
\end{align}
that is, if $\lambda$ is an eigenvalue of $D$ with the eigenspace $E_{\lambda}(D)$, then $\lambda$ is an eigenvalue of $D^{(e)}$ with the eigenspace $E_{\lambda}(D^{(e)})$ {containing $\mathcal{E}(E_{\lambda}(D))$}. While in practice we can only obtain an approximation of the eigenfunction $\phi\in E_{\lambda}(D)$ rather than $\mathcal{E}\phi\in E_{\lambda}(D^{(e)})$, the knowledge of $\phi$ carries information about $\mathcal{E}\phi$. 

{
\subsection{Approximating the effective \ac{AD} by the observable \ac{AD}}
}

Since $\mathcal{E}$ is {typically unknown, Theorem \ref{Theorem:Commutation} cannot be directly applied to extract the common manifold information from the given observations. To understand how the common manifold information can be obtained from the observable \ac{AD}, in this subsection, we further explore the relation between the effective \ac{AD} and the observable \ac{AD}, particularly, without taking the marginalization operator $\mathcal{E}$ into account}. 

{Before describing the theoretical result, observe that for a function $f$ defined on $\MM\times \mathcal{N}_1\times \mathcal{N}_2$, the diffusion on the second sensor, $D^{(2)}f$, is constant on $\mathcal{N}_1$; that is, $D^{(2)}$ integrates information on $\MM$ and $\mathcal{N}_2$ and no information about $\mathcal{N}_1$ is embodied in $D^{(2)}f$. As a result, intuitively, the additional application of $D^{(1)}$, namely, $D^{(1)}D^{(2)}f$, not only ignores the information from $\mathcal{N}_2$, but also should not bear information about $\mathcal{N}_1$ as well, since $D^{(2)}f$ is constant on $\mathcal{N}_1$. Yet, rigorous inspection shows that it is not true in general. Observe for example (\ref{Proof:Theorem1:FirstEquation}). We have
\begin{align}
D^{(1)}f(x,y,z)
&=\int_{\MM\times\mathcal{N}_1} P_\epsilon^{(I_1)}((x,y),(x'',y''))\mathcal{E}f(x'') \ud\nu_{\MM}(x'')\ud\nu_{\mathcal{N}_1|\MM}(y''|x'')\nonumber\\
&=\int_{\MM}\Big[\int_{\mathcal{N}_1} P_\epsilon^{(I_1)}((x,y),(x'',y''))\ud\nu_{\mathcal{N}_1|\MM}(y''|x'')\Big]\mathcal{E}f(x'') \ud\nu_{\MM}(x'')\nonumber\,,
\end{align}
which is constant on $z$, but may still depend on $y$. 
To take a closer look, we assume that the kernel $\tilde{P}^{(1)}$ is a Gaussian kernel so that we could decouple the $\MM$ and $\mathcal{N}_1$ and have
\begin{align}
&\int_{\mathcal{N}_1} P_\epsilon^{(I_1)}((x,y),(x'',y''))\ud\nu_{\mathcal{N}_1|\MM}(y''|x'')\nonumber\\
=&\,\frac{\int_{\mathcal{N}_1}\tilde{P}^{(1)}\Big(\frac{\sqrt{d_{g^{(1)}}(x,x'')^2+d_{\mathcal{N}_1}(y,y'')^2}}{\sqrt{\epsilon}}\Big)\ud\nu_{\mathcal{N}_1|\MM}(y''|x'')}
{\int_{\MM}\int_{\mathcal{N}_1}\tilde{P}^{(1)}\Big(\frac{\sqrt{d_{g^{(1)}}(x,x')^2+d_{\mathcal{N}_1}(y,y')^2}}{\sqrt{\epsilon}}\Big)\ud\nu_{\mathcal{N}_1|\MM}(y'|x')\ud \nu_{\MM}(x')}\nonumber\\
=&\,\frac{e^{-d_{g^{(1)}}(x,x'')^2/\epsilon}\int_{\mathcal{N}_1}e^{-d_{\mathcal{N}_1}(y,y'')^2/\epsilon}\ud\nu_{\mathcal{N}_1|\MM}(y''|x'')}{\int_{\MM}e^{-d_{g^{(1)}}(x,x')^2/\epsilon}\int_{\mathcal{N}_1}e^{-d_{\mathcal{N}_1}(y,y')^2/\epsilon}\ud\nu_{\mathcal{N}_1|\MM}(y'|x')\ud \nu_{\MM}(x')}\,,\nonumber
\end{align}
where $\int_{\mathcal{N}_1}e^{-d_{\mathcal{N}_1}(y,y'')^2/\epsilon}\ud\nu_{\mathcal{N}_1|\MM}(y''|x'')$ depends on $y$. 
Therefore, while $Df(x,y,z)=D^{(1)}D^{(2)}f(x,y,z)$ is constant in $z$, it is not necessarily constant in $y$ and it depends on the geometry of $\mathcal{N}_1$ via the kernel integration.
To better quantify this dependence, we introduce the following definition.
 } 

\begin{defn}[Nuisance-dependent diffusion kernel on the common manifold]
For a fixed $y\in\mathcal{N}_1$, let $P^{(\mathcal{N}_1,y)}_\epsilon$ be the {nuisance-dependent diffusion kernel} on $\MM$, which is defined in the normalized form by
\begin{align}
P^{(\mathcal{N}_1,y)}_\epsilon(x,x'):=\frac{\tilde{P}^{(\mathcal{N}_1,y)}\big(\frac{d_{g^{(1)}}(x,x')}{\sqrt{\epsilon}}\big) }{\int_{\MM}\tilde{P}^{(\mathcal{N}_1,y)}\big(\frac{d_{g^{(1)}}(x,x')}{\sqrt{\epsilon}}\big)\ud \nu_{\MM}(x')} ,
\end{align}
{where the superscript $\mathcal{N}_1$ stands for the dependence on the nuisance variable corresponding to the first sensor} and 
\begin{align}
\tilde{P}^{(\mathcal{N}_1,y)}\Big(\frac{d_{g^{(1)}}(x,x')}{\sqrt{\epsilon}}\Big)&:={\int_{\mathcal{N}_1} \tilde{P}_\epsilon^{(I_1)}((x,y),(x',y'))\ud\nu_{\mathcal{N}_1|\MM}(y'|x')}\\
&=\int_{\mathcal{N}_1}\tilde{P}^{(1)}\Big(\frac{\sqrt{d_{g^{(1)}}(x,x')^2+d_{\mathcal{N}_1}(y,y')^2}}{\sqrt{\epsilon}}\Big)\ud\nu_{\mathcal{N}_1|\MM}(y'|x') .\nonumber
\end{align}
Similarly, for a fixed $z\in\mathcal{N}_2$, the {nuisance-dependent diffusion kernel} on $\MM$, $P^{(\mathcal{N}_2,z)}_\epsilon$, is defined by
\begin{align}
P^{(\mathcal{N}_2,z)}_\epsilon(x,x'):=\frac{\tilde{P}^{(\mathcal{N}_2,z)}\big(\frac{d_{g^{(2)}}(x,x')}{\sqrt{\epsilon}}\big)}{\int_{\MM}\tilde{P}^{(\mathcal{N}_2,z)}\big(\frac{d_{g^{(2)}}(x,x')}{\sqrt{\epsilon}}\big)\ud \nu_{\MM}(x')} ,
\end{align}
where
\begin{align}
\tilde{P}^{(\mathcal{N}_2,z)}\Big(\frac{d_{g^{(2)}}(x,x')}{\sqrt{\epsilon}}\Big):=\int_{\mathcal{N}_2} \tilde{P}_\epsilon^{(I_2)}((x,z),(x',z'))\ud\nu_{\mathcal{N}_2|\MM}(z'|x') .
\end{align}
\label{def:nuisance_kernel}
\end{defn}
{Note that $\tilde{P}^{(\mathcal{N}_1,y)}$ is a family of kernels defined on $\MM$ that depend on the nuisance variable $y$,
and by definition, it is related to the effective diffusion kernel by} \begin{equation}
\mathcal{E}P^{(\mathcal{N}_1,y)}_\epsilon=P^{(e_1)}_\epsilon.
\end{equation} 
{A similar argument holds for $P^{(\mathcal{N}_2,z)}_\epsilon$, where we have} 
\begin{equation}
\mathcal{E}P^{(\mathcal{N}_2,z)}_\epsilon=P^{(e_2)}_\epsilon.
\end{equation}
With this definition, we summarize the relation between the {observable} \ac{AD} and the effective \ac{AD} in the following theorem.

{
\begin{thm}\label{Theorem:DataEffectiveRelationship}
Fix $(x,y,z)\in\MM\times\mathcal{N}_1\times \mathcal{N}_2$ and a continuous function $f$ that is constant in $\mathcal{N}_1$.
The {observable} \ac{AD} on $\MM\times\mathcal{N}_1\times \mathcal{N}_2$ starting from the first sensor satisfies
\begin{align}
Df(x,y,z)=&\,
\int_{\MM} \Big[\int_{\MM}P^{(\mathcal{N}_2,z)}_\epsilon(x,x')P^{(e_1)}(x',x'')\ud\nu_{\MM}(x')\Big] \mathcal{E}f(x'') \ud\nu_{\MM}(x'')\nonumber\,.
\end{align}
\end{thm}

\begin{proof}
By a direct expansion based on the definition of \ac{AD}, we have
\begin{align}
&Df(x,y,z)\label{Definition:AD:ExpressionForAnalysisOnM}\\
=&\,\int_{\MM\times\mathcal{N}_1\times\mathcal{N}_2} P_\epsilon^{(2)}((x,y,z),(x',y',z'))\int_{\MM\times\mathcal{N}_1\times\mathcal{N}_2}P_\epsilon^{(1)}((x',y',z'),(x'',y'',z''))\nonumber\\
&\qquad\times  f(x'',y'',z'') \ud \nu(x'',y'',z'')\ud \nu(x',y',z')  \nonumber\\
=&\,\int_{\MM\times\mathcal{N}_1\times\mathcal{N}_2} \Big\{\int_{\MM\times\mathcal{N}_1\times\mathcal{N}_2}P_\epsilon^{(I_2)}((x,z),(x',z'))P_\epsilon^{(I_1)}((x',y'),(x'',y''))\nonumber\\
&\qquad\times  \ud\nu_{\mathcal{N}_1|\MM}(y'|x')\ud\nu_{\mathcal{N}_2|\MM}(z'|x')\ud\nu_{\MM}(x')\Big\}   f(x'',y'',z'')\ud\nu_{\mathcal{N}_1|\MM}(y''|x'')\ud\nu_{\mathcal{N}_2|\MM}(z''|x'')\ud\nu_{\MM}(x'')\nonumber\\
=&\,\int_{\MM\times\mathcal{N}_1\times\mathcal{N}_2} \Big\{\int_{\MM}\Big[\int_{\mathcal{N}_2}P_\epsilon^{(I_2)}((x,z),(x',z')) \ud\nu_{\mathcal{N}_2|\MM}(z'|x') \int_{\mathcal{N}_1} P_\epsilon^{(I_1)}((x',y'),(x'',y''))\ud\nu_{\mathcal{N}_1|\MM}(y'|x')\Big]\nonumber\\
&\qquad\times  \ud\nu_{\MM}(x')\Big\}   f(x'',y'',z'')\ud\nu_{\mathcal{N}_1|\MM}(y''|x'')\ud\nu_{\mathcal{N}_2|\MM}(z''|x'')\ud\nu_{\MM}(x'')\nonumber\,,
\end{align}
which, by the definition of the nuisance-dependent diffusion kernel on the common manifold, could be reduced to
\begin{align}
&Df(x,y,z)=\int_{\MM\times\mathcal{N}_1\times\mathcal{N}_2} \Big\{\int_{\MM}P^{(\mathcal{N}_2,z)}_\epsilon(x,x') P^{(\mathcal{N}_1,y'')}_\epsilon(x',x'') \ud\nu_{\MM}(x')\Big\}  \nonumber\\
&\qquad\qquad\times f(x'',y'',z'')\ud\nu_{\mathcal{N}_1|\MM}(y''|x'')\ud\nu_{\mathcal{N}_2|\MM}(z''|x'')\ud\nu_{\MM}(x'')\nonumber\\
=&\int_{\MM} \int_{\MM}P^{(\mathcal{N}_2,z)}_\epsilon(x,x')\Big[\int_{\mathcal{N}_2} \int_{\mathcal{N}_1} P^{(\mathcal{N}_1,y'')}_\epsilon(x',x'')f(x'',y'',z'')\nonumber\\
&\qquad\qquad\times\ud\nu_{\mathcal{N}_1|\MM}(y''|x'') \ud\nu_{\mathcal{N}_2|\MM}(z''|x'')\Big]\ud\nu_{\MM}(x') \ud\nu_{\MM}(x'')\label{Expansion:RelateObservationCommon1}\,.
\end{align}
To further reduce (\ref{Expansion:RelateObservationCommon1}), note that by the assumption that $f$ is constant in $y''$, we have 
\[
f(x'',y'',z'')=\int_{\mathcal{N}_1}f(x'',y''',z'')\ud\nu_{\mathcal{N}_1|\MM}(y'''|x'')\,.
\]
Hence, the integrant quantity inside the bracket in (\ref{Expansion:RelateObservationCommon1}) could be simplified by
\begin{align}
&\int_{\mathcal{N}_2}\int_{\mathcal{N}_1} P^{(\mathcal{N}_1,y'')}_\epsilon(x',x'')f(x'',y'',z'')\ud\nu_{\mathcal{N}_1|\MM}(y''|x'')\ud\nu_{\mathcal{N}_2|\MM}(z''|x'')\nonumber\\
=&\,\int_{\mathcal{N}_2}\int_{\mathcal{N}_1} P^{(\mathcal{N}_1,y'')}_\epsilon(x',x'')\int_{\mathcal{N}_1}f(x'',y''',z'')\ud\nu_{\mathcal{N}_1|\MM}(y'''|x'')\ud\nu_{\mathcal{N}_1|\MM}(y''|x'')\ud\nu_{\mathcal{N}_2|\MM}(z''|x'')\nonumber\\
=&\,\int_{\mathcal{N}_1} P^{(\mathcal{N}_1,y'')}_\epsilon(x',x'')\ud\nu_{\mathcal{N}_1|\MM}(y''|x'')\int_{\mathcal{N}_2}\int_{\mathcal{N}_1}f(x'',y''',z'')\ud\nu_{\mathcal{N}_1|\MM}(y'''|x'')\ud\nu_{\mathcal{N}_2|\MM}(z''|x'')\nonumber\\
=&\,\Big[\int_{\mathcal{N}_1} P^{(\mathcal{N}_1,y'')}_\epsilon(x',x'')\ud\nu_{\mathcal{N}_1|\MM}(y''|x'')\Big]\mathcal{E}f(x'')\nonumber\\
=&\,P^{(e_1)}(x',x'') \mathcal{E}f(x'')\label{Expansion:RelateObservationCommon2}\,.
\end{align}
By plugging (\ref{Expansion:RelateObservationCommon2}) into (\ref{Expansion:RelateObservationCommon1}), $Df(x,y,z)$ is reduced to
\begin{align}
Df(x,y,z)=&\,
\int_{\MM} \Big[\int_{\MM}P^{(\mathcal{N}_2,z)}_\epsilon(x,x')P^{(e_1)}(x',x'')\ud\nu_{\MM}(x')\Big] \mathcal{E}f(x'') \ud\nu_{\MM}(x'')\nonumber .
\end{align}
\end{proof}
}
Theorem \ref{Theorem:DataEffectiveRelationship} implies that when the nuisance variables $y$ and $z$ are fixed, the {observable} \ac{AD} can be viewed as ``an ordinary" diffusion process on the common manifold {with the kernel $\int_{\MM}P^{(\mathcal{N}_2,z)}_\epsilon(x,x')P^{(e_1)}(x',x'')\ud\nu_{\MM}(x')$, which depends on $z$}. Hence, {by Theorem \ref{Theorem:Commutation},} the effective \ac{AD} is related to the {observable} \ac{AD} by { taking expectation with respect to the nuisance variable of the second sensor $z$, when using the {observable} \ac{AD} starting from the first sensor:
\begin{align}
[D^{(e)}\mathcal{E}f](x)=\mathcal{E}\int_{\MM} \Big[\int_{\MM}P^{(\mathcal{N}_2,z)}_\epsilon(x,x')P^{(e_1)}(x',x'')\ud\nu_{\MM}(x')\Big] \mathcal{E}f(x'') \ud\nu_{\MM}(x'')\,.\label{RelationshipEffectiveData:1}
\end{align}
This result indicates that we can gain access to the common manifold information via the effective \ac{AD} by viewing {the observable} \ac{AD} as a proxy. However, the information provided by {the observable} \ac{AD} depends on the data, and we would only achieve an extraction of undistorted information on the common manifold through the effective \ac{AD} if sufficient amount of data is available to ``average out'' the nuisance variable. Based on this result, in the next section, we} focus on studying the net result of \ac{AD} solely on the common manifold while ignoring the nuisance variables.

\begin{rem}
{{By the same argument we could obtain a parallel result as (\ref{RelationshipEffectiveData:1}) for the forward diffusion operator, which, if applied to the delta measure supported at two points, recovers the effective alternating-diffusion distance considered in \cite[Equation 79]{lederman2015alternating}. Notice that in \cite{lederman2015alternating}, it is the effective alternating-diffusion distance, which embodies the ``averaged behavior'' of the observable \ac{AD} kernel, that is discussed, while in this paper, we focus on the diffusion behavior at each point.}}
\end{rem}
{

\subsection{Alternating diffusion with nuisance variables modeled with manifold structures}\label{Section:ManifoldNuisance}

Before closing this section, we further study how the observable \ac{AD} depends on the data. Under a stronger condition that $\mathcal{N}_2$ is a $q$-dim closed Riemannian manifold (compact and without boundaries) with the metric $g_{\mathcal{N}_2}$ and that for each $x''$, $\ud\nu_{\mathcal{N}_2|\MM}(z''|x'')$ is absolutely continuous with related to the volume form $dV_{\mathcal{N}_2}$ associated with $g_{\mathcal{N}_2}$, we claim that asymptotically when $\epsilon$ is sufficiently small and the Radon-Nikodym derivative $\frac{\ud\nu_{\mathcal{N}_2|\MM}(z''|x'')}{dV_{\mathcal{N}_2}(z'')}$ is constant (uniform sampling conditional on $x''$), $Df(x,y,z)$ is almost constant in $z$ if we start \ac{AD} from the first sensor. We mention that while a more complicated condition could be considered, to simplify the discussion we focus on this assumption.
We have  
\begin{align}
&\int_{\mathcal{N}_2} \tilde{P}^{(2)}\Big(\frac{\sqrt{d_{g^{(2)}}(x,x'')^2+d_{\mathcal{N}_2}(z,z'')^2}}{\sqrt{\epsilon}}\Big) \ud V_{\mathcal{N}_2}(z'')\label{Calculation:FiberManifoldAssumption1}\\
=\,&\int_{B_z} \tilde{P}^{(2)}\Big(\frac{\sqrt{d_{g^{(2)}}(x,x'')^2+\|u\|^2}}{\sqrt{\epsilon}}\Big)  (1+\text{Ric}_z(i,j) u_iu_j+O(\|u\|^3))\ud u\nonumber
\end{align}
which stems from the change of variables with the normal coordinate at $z$, $B_z:=\exp_z^{-1}(\mathcal{N}_2\backslash C_z)\subset T_z\mathcal{N}_2$, $C_z$ is the cut locus of $z$, and $\text{Ric}_z$ is the Ricci curvature of $(\mathcal{N}_2,g_{\mathcal{N}_2})$ at $z$. 
Define $\varphi_\ell(a):=\int_{0}^\infty \tilde{P}^{(2)}\Big(\sqrt{a+s^2}\Big) s^{q-1+\ell} \ud s$ for $a\geq 0$ and $\ell=0,1,\ldots$.
By changing the Cartesian coordinates to polar coordinates on $T_z\mathcal{N}_2$, (\ref{Calculation:FiberManifoldAssumption1}) can be recast as
\begin{align}
&\epsilon^{(q-1)/2}\int_{\mathbb{R}^q} \tilde{P}^{(2)}\Big(\sqrt{\frac{d_{g^{(2)}}(x,x'')^2}{\epsilon}+s^2}\Big)  (s^{q-1}+\epsilon\text{Ric}_z(\theta,\theta) s^{q+1})\ud s\ud \theta+O(\epsilon^{(q+2)/2})\,\label{Calculation:FiberManifoldAssumption2}\\
=\,&\epsilon^{(q-1)/2}|S^{q-1}| \varphi_0(d_{g^{(2)}}(x,x'')^2/\epsilon)+\epsilon^{(q+1)/2}\frac{|S^{q-1}|}{q}s_y\varphi_2(d_{g^{(2)}}(x,x'')^2/\epsilon)+O(\epsilon^{(q+2)/2})\,,\nonumber
\end{align}
where we approximate the integration over $B_y$ by the integration over $\mathbb{R}^q$ exploiting the fast decay assumption of the kernel function, and $s_z$ is the scalar curvature of $(\mathcal{N}_2,g_{\mathcal{N}_2})$ at $z$.

Now, by the uniform sampling assumption, we have
\begin{align}
&\int_{\mathcal{N}_2} P_\epsilon^{(I_2)}((x,z),(x'',z''))\ud\nu_{\mathcal{N}_2|\MM}(z''|x'')\nonumber\\
&\quad=\,\frac{\int_{\mathcal{N}_2} \tilde{P}^{(2)}\Big(\frac{\sqrt{d_{g^{(2)}}(x,x'')^2+d_{\mathcal{N}_2}(z,z'')^2}}{\sqrt{\epsilon}}\Big) \ud V_{\mathcal{N}_2}(z'')}{\int_{\MM}\int_{\mathcal{N}_2} \tilde{P}^{(2)}\Big(\frac{\sqrt{d_{g^{(2)}}(x,x'')^2+d_{\mathcal{N}_2}(z,z'')^2}}{\sqrt{\epsilon}}\Big) \ud V_{\mathcal{N}_2}(z'')\ud \nu_{\MM}(x'')}\nonumber\,,
\end{align}
which by (\ref{Calculation:FiberManifoldAssumption2}) is reduced to
\begin{align}
&\int_{\mathcal{N}_2} P_\epsilon^{(I_2)}((x,z),(x'',z''))\ud\nu_{\mathcal{N}_2|\MM}(z''|x'')\nonumber\\
&\quad=\frac{\varphi_0(d_{g^{(2)}}(x,x'')^2/\epsilon)+\epsilon \frac{s_z}{q}\varphi_2(d_{g^{(2)}}(x,x'')^2/\epsilon)+O(\epsilon^{3/2})  }
{\int_{\MM}  \varphi_0(d_{g^{(2)}}(x,x'')^2/\epsilon)\ud \nu_{\MM}(x'') +\epsilon \frac{s_z}{q} \int_{\MM}\varphi_2(d_{g^{(2)}}(x,x'')^2/\epsilon)\ud \nu_{\MM}(x'')+O(\epsilon^{3/2})}\nonumber\\
&\quad=\,\frac{\varphi_0(d_{g^{(2)}}(x,x'')^2/\epsilon) }{\int_{\MM}  \varphi_0(d_{g^{(2)}}(x,x'')^2/\epsilon)\ud\nu_{\MM}(x'')}+\epsilon \frac{s_z}{q}G(x,x'')  +O(\epsilon^{3/2})\nonumber.
\end{align}
where
\begin{align}
G(x,x'') &= \frac{\varphi_2(d_{g^{(2)}}(x,x'')^2/\epsilon)}{\int_{\MM}  \varphi_0(d_{g^{(2)}}(x,x'')^2/\epsilon)\ud \nu_{\MM}(x'')}\nonumber\\
&-\frac{\varphi_0(d_{g^{(2)}}(x,x'')^2/\epsilon) \int_{\MM}\varphi_2(d_{g^{(2)}}(x,x'')^2/\epsilon)\ud \nu_{\MM}(x'') }{(\int_{\MM}  \varphi_0(d_{g^{(2)}}(x,x'')^2/\epsilon)\ud \nu_{\MM}(x''))^2} \nonumber
\end{align}
The above derivation implies that under the assumption that the nuisance variable lies on a compact manifold and is sampled uniformly, the quantity $\int_{\mathcal{N}_2} P_\epsilon^{(I_2)}((x,z),(x'',z''))\ud\nu_{\mathcal{N}_2|\MM}(z''|x'')$ is almost constant in $z$, and the dependence on $z$ is of order $\epsilon$. Hence, under this strong condition, we conclude from Theorem \ref{Theorem:DataEffectiveRelationship} and from Definition \ref{def:nuisance_kernel} that $Df(x,y,z)$ is almost constant in $z$ if we start \ac{AD} from the first sensor, and hence by (\ref{RelationshipEffectiveData:1}), we could recover the effective \ac{AD} via the observable \ac{AD} with a higher order error. It is worthwhile noting that an analogous result can be derived for \ac{AD} starting from the second sensor, if $\mathcal{N}_1$ is a manifold with the same assumptions.
}

\section{Analysis of the alternating diffusion {under the common manifold model}}
\label{Section:ADAnalysis}

In this section, we study the effective \ac{AD} under the manifold setup. 
{Based on the discussion in Section \ref{Section:CommonManifoldAlternationDiffusion}, we present the analysis and proof under a reduced setting which does not contain the nuisance variables. Yet, the results of the analysis under this setting are transferable to the general setting with the nuisance variables. }

{Before presenting the analytical results, we take a closer look at the available metrics when we analyze data.}
First, theoretically, if the metric $d_{\mathcal{S}_i}$ for the dataset faithfully reflects the geodesic distance on $\MM$, then the analysis becomes simple. However, it is usually not the case in practice and the best we could expect is that the metric $d_{\mathcal{S}_i}$ for the dataset provides a good approximation of the geodesic distance on $\MM$; that is, $d_{g^{(i)}}(x,x')$, for $i=1,2$, in (\ref{Equation:MetricDefinition:RelationshipExact}) can be approximated from the observations and their ambient metrics $d_{\mathcal{S}_i}$. Denote the approximation of $d_{g^{(i)}}(x,x')$ by $\bar{d}_{g^{(i)}}(x,x')$.
The error introduced by the difference between the metric we have for the dataset and the geodesic distance might not be easily quantified in general. Thus, for mathematical tractability, we make further assumptions regarding the discrepancy between $d_{g^{(i)}}(x,x')$ and $\bar{d}_{g^{(i)}}(x,x')$. When the metric is well designed, we can assume that for sufficiently close $x$ and $x'$, $d_{g^{(i)}}(x,x')$ and $\bar{d}_{g^{(i)}}(x,x')$ are close up to a higher order error. Moreover, in some situations we are able to directly quantify the error, which helps us to further quantify the information that can be extracted from \ac{AD}.

{We now recast the formulation of the problem and the definitions of the diffusion kernels and operators from Section \ref{Section:CommonManifoldAlternationDiffusion} under the reduced setting, and introduce notations used in this section for the analysis.}
Suppose the common manifold $\MM$ is smoothly embedded into $\RR^p$ via $\iota^{(i)}$, $i=1,2$, with the induced {metric $g^{(i)}$ and hence the induced distance function $d_{g^{(i)}}$}. We sample $\MM$ via two random vectors, $X_i:=\iota^{(i)}\circ X$, where $X:(\Omega,\mathcal{F},P)\to \MM$, with the induced measure on $\MM$, denoted by $\mu_{\MM}:=X_*P$. 
This model entails that under the reduced setting we are not able to access $\MM$ directly, but only through the two sensors via $X_i\in \RR^p$. In this case, 
\begin{align}
\bar{d}_{g^{(i)}}(x,x')=\|\iota^{(i)}(x)-\iota^{(i)}(x')\|, 
\end{align}
where $\|\cdot\|$ means the Euclidean distance between any two samples $x$ and $x'$ from $\MM$;
in other words, we only approximate the geodesic distance on $\MM$ using the Euclidean distance. In general, other approximations could be used, e.g. based on an embedding to a manifold or a well-designed metric space, but to simplify the analysis we focus on the Euclidean space. Clearly, $\iota^{(1)}(\MM)$ and $\iota^{(2)}(\MM)$ are diffeomorphic to each other via a differomorphism $\Phi$ so that $\Phi=\iota^{(2)}\circ {\iota^{(1)}}^{-1}$. 
See Figure \ref{fig:diagram2} for an illustration of this reduced setting.

\begin{figure}[h]
\begin{centering} 
\includegraphics[width=.8\textwidth]{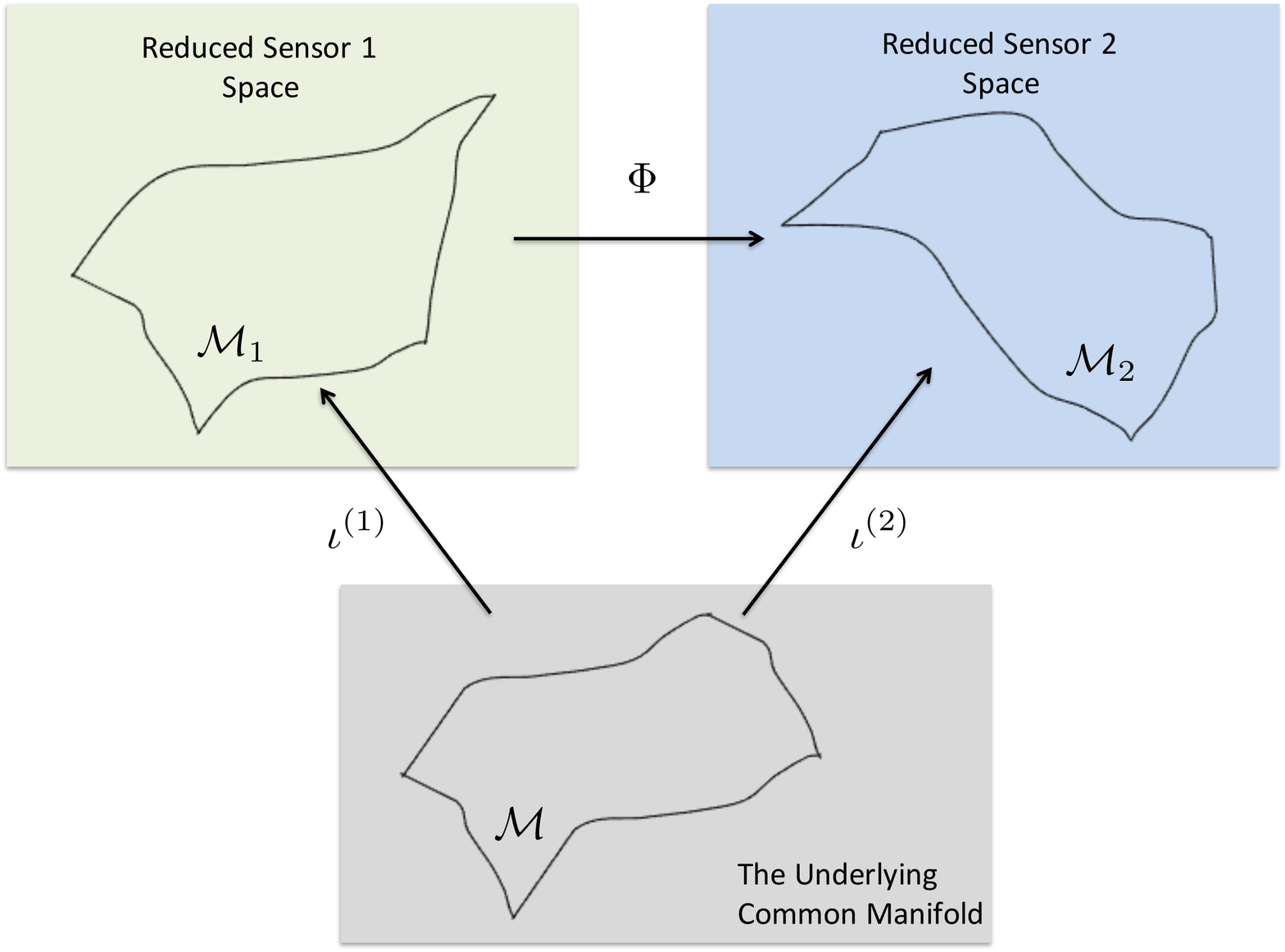}
\caption{A diagram illustrating the reduced setting without nuisance variables.}
\label{fig:diagram2}
\end{centering}
\end{figure}

We study the effective \ac{AD} starting from the first sensor, defined in (\ref{Definition:EffectiveADFromTheFirstSensor}), by studying the following diffusion process on $\MM$. Note that the effective \ac{AD} starting from the second sensor can be analyzed in an analogous way. 

\begin{defn}[Reduced Alternating Diffusion Operator]
Take two kernels $\tilde{K}^{(1)}$ and $\tilde{K}^{(2)}$, and define the reduced (without the nuisance variables) \ac{AD} operator $T:C(\MM)\to C(\MM)$ by 
\begin{align}
Tf(x):=&\int_{\MM}\frac{\tilde{K}_\epsilon^{(2)}(x,x')}{\int_{\MM}\tilde{K}_\epsilon^{(2)}(x,\bar{x}) \ud\mu_{\MM}(\bar{x})} \label{Definition:diffusionManifoldCaseT}\\
&\qquad\times \Big[\int_{\MM}  \frac{\tilde{K}_\epsilon^{(1)}(x',x'')}{\int_{\MM}\tilde{K}_\epsilon^{(1)}(x',\bar{x}) \ud\mu_{\MM}(\bar{x})}f(x'') \ud\mu_{\MM}(x'')\Big]\ud\mu_{\MM}(x')\nonumber
\end{align}
where $f\in C(\MM)$, and    
\begin{align}
\tilde{K}_\epsilon^{(i)}(x,x'):=\tilde{K}^{(i)}\left(\frac{\|\iota^{(i)}(x)-\iota^{(i)}(x')\|}{\sqrt{\epsilon}}\right)
\end{align}
for $i=1,2$.
\end{defn}
\begin{defn}[Reduced Alternating Diffusion Kernel]
Let $\tilde{K}_\epsilon^{(e)}(x,x'')$ be the reduced (without the nuisance variables) \ac{AD} kernel, defined by
\begin{align}
\tilde{K}_\epsilon^{(e)}(x,x''):=\int_{\MM} 
\frac{\tilde{K}_\epsilon^{(2)}(x,x') \tilde{K}_\epsilon^{(1)}(x',x'')}{\int_{\MM}\tilde{K}_\epsilon^{(1)}(x',\bar{x})\ud\mu_{\MM}(\bar{x})}\ud\mu_{\MM}(x').\label{Definition:effectiveKernelManifoldCase}
\end{align}
\end{defn}
By definition, we have
\begin{align}
\int_{\MM}\tilde{K}_\epsilon^{(e)}(x,\bar{x})\ud\mu_{\MM}(\bar{x})=\int_{\MM}\tilde{K}_\epsilon^{(2)}(x,\bar{x}) \ud\mu_{\MM}(\bar{x}).
\end{align}
and
\begin{align}
Tf(x)=\,\frac{\int_{\MM}\tilde{K}_\epsilon^{(e)}(x,x'')f(x'')\ud\mu_{\MM}(x'')}{\int_{\MM}\tilde{K}_\epsilon^{(e)}(x,\bar{x})\ud\mu_{\MM}(\bar{x}) }
\end{align}

The notation used in the remainder of this section is as follows. Note that $\iota^{(i)}(\MM)$, $i=1,2$, is now a sub-manifold of $\RR^p$, and the distance function $g^{(i)}$ is induced from the canonical metric of $\RR^p$. Let $\ud V^{(i)}$ denote the measure associated with the Riemannian volume form induced from $g^{(i)}$. We use the notation $\nabla^{(i)}$, $\exp^{(i)}$, $\Ric^{(i)}$, and $\Pi^{(i)}$ to denote the covariant derivative, the exponential map, the Ricci curvature and the second fundamental form associated with $\iota^{(i)}$, respectively. 

We start from the following assumptions.

\begin{assumption}\label{Assumption:A}
\begin{enumerate}
\item[(A1)] The manifold $\MM$ is $d$-dim, compact and without a boundary. It is embedded into $\RR^p$ via $\iota\in C^4(\MM,\RR^p)$ with the metric $g$ induced from the canonical metric of $\RR^p$. 
\item[(A2)] The kernel functions satisfy the following conditions. For $i=1,2$, $\tilde{K}^{(i)}\in C^2([0,\infty))$ are positive, decay exponentially fast and $\tilde{K}^{(i)}(0)>0$. Further, there exists $c_1,c_2>0$ so that $\tilde{K}^{(i)}(t)<c_1e^{-c_2t^2}$ and $|[\tilde{K}^{(i)}]'(t)|\leq c_1e^{-c_2t^2}$.
Denote $\mu^{(i)}_{l,k}:=\int_{\mathbb{R}^d}\|x\|^l \partial^k\tilde{K}^{(i)}(\|x\|)\ud x<\infty$, where $l\in\{0\}\cup \NN$, and $\partial^k$ is the $k$-th derivative, for $k=0,1,2$. Assume $\mu^{(i)}_{0,0}=1$.
\item[(A3)] The bandwidth of the kernel, denoted by $\epsilon$, satisfies $0<\sqrt{\epsilon}< \min\{\tau,\text{inj}(\MM)\}$, where $\tau$ is the reach of the manifold \cite{Niyogi2006} and $\text{inj}(\MM)$ is the injectivity of $\MM$ \cite[p 271]{doCarmo:1992}.
\item[(A4)] Assume that $\ud \mu_{\MM}$ is absolutely continuous with related to $\ud V^{(i)}$, and we denote $p_i:=\frac{\ud \mu_{\MM}}{\ud V^{(i)}}$ as the probability density function (p.d.f.) of $X$ on $\MM$ by the Radon-Nikodym theorem. Furthermore, assume that $p_i\in C^4(\MM)$ so that $0<\min p_i\leq\max p_i$, where $i=1,2$. 
\end{enumerate}
\end{assumption}

We mention that Assumption (A4) is necessary for the sake of analyzing the asymptotical behavior of \ac{AD} without the influence of the nuisance variables. In particular it contributes to the symmetric argument used in the proofs of the theorems.

With the above preparation, we are ready to state our main results.
The first theorem states that the effective \ac{AD} kernel defined under a model ignoring the nuisance variables behaves essentially like an ordinary diffusion kernel in the normalized form. In particular, when both kernels are Gaussian, the effective \ac{AD} kernel is Gaussian as well.  

\begin{thm}
Suppose Assumptions (A1)-(A4) hold. Take $f\in C^3(\MM)$, $0<\gamma<1/2$ and $x,x''\in \MM$ so that $x''=\exp^{(1)}_x v$, where $v\in T_x\MM$ and $\|v\|_{g^{(1)}}\leq 2\epsilon^\gamma$. 
Denote $R_x=[\ud\exp^{(2)}_x|_0]^{-1}[\ud \iota^{(2)}]^{-1}\nabla\Phi[\ud\iota^{(1)}][\ud\exp^{(1)}_x|_0]:\RR^d\to \RR^d$ Then, when $\epsilon$ is sufficiently small, the following holds:
\begin{align}
\tilde{K}^{(e)}_\epsilon(x,x'')=\int_{\RR^d}\tilde{K}^{(2)}\left(\|R_xw\|\right) \tilde{K}^{(1)}\left(\|w-v/\sqrt{\epsilon}\|\right) \ud w+\epsilon A_{2,\epsilon}(1,v) +O(\epsilon^{3/2})\nonumber
\end{align}
where $A_{2,\epsilon}(1,v)$ is defined in (\ref{Proof:KeyLemma:DefinitionH2eps}), which decays exponentially, $A_{2,\epsilon}(1,0)$ is of order $O(1)$, and $A_{2,\epsilon}(1,-v)=A_{2,\epsilon}(1,v)$. In particular, if $K^{(1)}(t)=K^{(2)}(t)=e^{-t^2}/\pi^{d/2}$, we have
\begin{align}
\int_{\RR^d}\tilde{K}^{(2)}\left(\|R_xw\|\right) \tilde{K}^{(1)}\left(\|w-v\|\right) \ud w=\frac{e^{-\|(I+R_x^2)^{-1/2}R_xv\|^2/\epsilon}}{\sqrt{\det(I+R_x^2)}}.
\end{align}
On the other hand, when $\|v\|_{g^{(1)}}> 2\epsilon^\gamma$,
\begin{align}
\tilde{K}^{(e)}_\epsilon(x,x'')=O(\epsilon^{3/2}).
\end{align}
\label{Theorem:ConvolutionKernel}
\end{thm}
Note that $R_x$ embodies the difference between the two metrics. \ac{AD} consists of two diffusion steps: the first, carried out by $\tilde{K}^{(1)}$, respects the metric $g^{(1)}$, and the second, carried out by $\tilde{K}^{(2)}$, respects the metric $g^{(2)}$. In order to study the integrated behavior of the two consecutive different diffusion steps, we quantify the overall effect using $g^{(1)}$ via $R_x$. When $\iota^{(1)}=\iota^{(2)}$, that is, when $\Phi$ is the identity map, then $R_x$ is reduced to the identity as well, and the common manifold setup is reduced to the setup considered in \cite{lederman2015alternating}.

Based on the behavior of the reduced effective \ac{AD} kernel $\tilde{K}^{(e)}$ studied in Theorem \ref{Theorem:ConvolutionKernel}, we study the asymptotic behavior of the reduced \ac{AD} operator $T$ (without the influence of the nuisance variables).
The second theorem states that asymptotically $T$ is a deformed Laplace-Beltrami operator defined on $(\MM, g^{(1)})$. 
\begin{thm}
Suppose $f\in C^3(\MM)$. Fix normal coordinates around $x$ associated with $g^{(1)}$ and $g^{(2)}$ so that $\{E_i\}_{i=1}^d\subset T_x\MM$ is orthonormal associated with $g^{(1)}$. Set $R_x=[\ud\exp^{(2)}_x|_0]^{-1}[\ud \iota^{(2)}]^{-1}\nabla\Phi[\ud\iota^{(1)}][\ud\exp^{(1)}_x|_0]$ and by the \ac{SVD} $R_x=U_x\Lambda_xV_x^T$, where $\Lambda_x=\text{diag}[\lambda_1,\ldots,\lambda_d]$. Then, when $\epsilon$ is sufficiently small, the \ac{AD} without the nuisance variables starting from $g^{(1)}$ satisfies
\begin{align}
Tf(x)=\,& f(x)+\frac{\epsilon \mu^{(2)}_{2,0}}{2d^2}\sum_{i=1}^d\lambda_i\big[{\nabla^{(1)}}^2_{E_i,E_i}f(x)+\frac{2\nabla^{(1)}_{E_i}f(x)\nabla^{(1)}_{E_i}p_1(x)}{p_1(x)}\big] \nonumber\\
&\qquad+\frac{\epsilon \mu^{(1)}_{2,0}}{2d^2}\Big[\Delta^{(1)}f(x)+\frac{2\nabla^{(1)} f(x)\cdot\nabla^{(1)} p_1(x)}{p_1(x)}\Big]+O(\epsilon^{3/2})\nonumber.
\end{align}
In particular, when $\iota^{(1)}=\iota^{(2)}$, that is, $R_x=I_d$, for every $x\in\MM$ we have
\begin{align}
Tf(x)=\,f(x)+\epsilon {\frac{\mu^{(1)}_{2,0}+\mu^{(2)}_{2,0}}{2d^2}}\Big[\Delta^{(1)} f(x)+\frac{\nabla^{(1)} p_1(x)\cdot \nabla^{(1)} f(x)}{p_1(x)}\Big]+O(\epsilon^{3/2})
\end{align}
\label{Theorem:MainTheorem2}
\end{thm}

This theorem implies that when $\Phi$ is not the identity map, then the obtained infinitesimal generator of the \ac{AD} process is a deformed Laplace-Beltrami operator of $\mathcal{M}$ associated with $g^{(1)}$. In particular, when $K^{(1)}=K^{(2)}$, we have
\begin{align}
Tf(x)=\,& f(x)+\frac{\epsilon \mu^{(1)}_{2,0}}{2d^2}\sum_{i=1}^d(1+\lambda_i)\big[{\nabla^{(1)}}^2_{E_i,E_i}f(x)+\frac{2\nabla^{(1)}_{E_i}f(x)\nabla^{(1)}_{E_i}p_1(x)}{p_1(x)}\big]+O(\epsilon^{3/2})\nonumber.
\end{align}
In addition, when $\Phi$ is the identity map, then the infinitesimal generator of the \ac{AD} process is precisely the Laplace-Beltrami operator of $\mathcal{M}$ associated with $g^{(1)}$. 

The proofs of Theorem \ref{Theorem:ConvolutionKernel} and Theorem \ref{Theorem:MainTheorem2} appear in \ref{Section:Appendix:Proof}. 
The theorems immediately lead to the following corollary, which describes the asymptotical behavior of the effective \ac{AD} operator in a model including the nuisance variables, studied in (\ref{Definition:EffectiveADFromTheFirstSensor}) and is associated with the \ac{AD} defined from data, starting from the first sensor (\ref{Definition:AD:ExpressionForAnalysisOnM}). We mention that since in general we are not able to convert (\ref{Description:Pe1IsNotNormalized}) into a normalized kernel, we need to directly study the observable diffusion kernel $P^{(\mathcal{N}_1,y)}_\epsilon$.

\begin{cor}\label{Theorem:MainCorollary}
Suppose $f\in C(\MM\times \mathcal{N}_1\times \mathcal{N}_2)$ so that $\mathcal{E}f\in C^3(\MM)$ and $\nu_{\MM}$ is absolutely continuous with respect to the Riemannian measure induced from $g^{(1)}$ so that $p_1=\frac{\ud\nu_{\MM}}{\ud V^{(1)}}\in C^4(\MM)$. Then by the definitions of the \ac{AD} operators in \eqref{eq:AD_operator} and (\ref{Definition:EffectiveADFromTheFirstSensor}), and by the commutativity from \eqref{eq:com}, we obtain
\begin{align}
&\mathcal{E}Df(x)=\mathcal{E}f(x)+\frac{\epsilon C_2}{2d^2}\sum_{i=1}^d\lambda_i\big[{\nabla^{(1)}}^2_{E_i,E_i}\mathcal{E}f(x)+\frac{2\nabla^{(1)}_{E_i}\mathcal{E}f(x)\nabla^{(1)}_{E_i}p_1(x)}{p_1(x)}\big] \\
&\qquad+\frac{\epsilon C_1}{2d^2}\Big[\Delta^{(1)}\mathcal{E}f(x)+\frac{2\nabla^{(1)} \mathcal{E}f(x)\nabla^{(1)} p_1(x)}{p_1(x)}\Big]+O(\epsilon^{3/2}).\nonumber
\end{align}
where $C_i$, $i=1,2$, are constants depending only of the chosen kernel $\tilde{P}^{(i)}$.
\end{cor}
The proof of this corollary appears in \ref{Section:Appendix:Proof} as well. The corollary states that after marginalization, the \ac{AD} operator computed from data is asymptotically (when $\epsilon$ is small) a deformed Laplace-Beltrame operator on the common manifold $\mathcal{M}$ associated with $g^{(1)}$. 

We end this section with two closing remarks. First, we note the importance of the order of the kernels consisting the \ac{AD} kernel, as implied by the analytic results (Theorem \ref{Theorem:ConvolutionKernel} and Theorem \ref{Theorem:MainTheorem2}). Second, these results further show that in order to compare the outcome of \ac{AD} starting from the first sensor to those of \ac{AD} starting from the second sensor, the deformation of the common manifold in each of the two sensors has to be taken into account. This issue is illustrated by an example in Section \ref{Section:Seasonality}.


{
\section{Alternating diffusion algorithm}\label{Section:ADalgorithm}
}

The \ac{AD} algorithm is summarized in Algorithm \ref{algo:AD1},
{ which is a direct discretization of the observable \ac{AD} operator.
}

\begin{algorithm}[t]
\caption{Alternating Diffusion algorithm}
\label{algo:AD1}
\begin{itemize}
\item[\textbf{Input:}] Two data sets $\mathcal{X}_l:=\{ x_{l,i} \}_{i=1}^n\subset\mathbb{R}^p$, where $l=1,2$, are given. $x_{1,i}$ and $x_{2,i}$ are sampled simultaneously from the sensors, for all $i=1,\ldots,n$. 

\item[\textbf{Parameters:}] Pick two positive kernels {$\tilde{P}^{(1)}$ and $\tilde{P}^{(2)}$} which decay fast enough. Fix {a positive integer $K\leq n$ and $\epsilon>0$}.

\item[\textbf{Output:}] The first $K$ {singular} values and {singular} vectors of the \ac{AD} starting with the first or the second sensor.

\end{itemize}

\begin{enumerate}
\item Build the first affinity matrix $\mathbf{W}_1 \in \mathbb{R}^{n \times n}$ based on the first data set $\mathcal{X}_1$ by ${\mathbf{W}_1(i,j)}=\tilde{P}^{(1)}(\|x_{1i}-x_{1j}\|/\sqrt{\epsilon})$.

\item Build the second affinity matrix $\mathbf{W}_2 \in \mathbb{R}^{n \times n}$ based on the second data set $\mathcal{X}_2$ by ${\mathbf{W}_2(i,j)}=\tilde{P}^{(2)}(\|x_{2i}-x_{2j}\|/\sqrt{\epsilon})$.

\item Build the first diffusion kernel $\mathbf{P}_1 \in \mathbb{R}^{n \times n}$ based on the first data set $\mathcal{X}_1$ by $\mathbf{P}_1(i,j)=\frac{\mathbf{W}_1(i,j)}{\sum_{l=1}^n\mathbf{W}_1(l,j)}$.

\item Build the second diffusion kernel $\mathbf{P}_2 \in \mathbb{R}^{n \times n}$ based on the second data set $\mathcal{X}_2$ by $\mathbf{P}_2(i,j)=\frac{\mathbf{W}_2(i,j)}{\sum_{l=1}^n\mathbf{W}_2(l,j)}$.

\item Run the \ac{AD} starting with the first sensor with the diffusion operator $\mathbf{P}_2\mathbf{P}_1$ and obtain the first $K$ {right singular vectors} $u^{(1)}_i \in \mathbb{R}^n$, $i=1,\ldots,K$. 

\item Run the \ac{AD} starting with the second sensor with the diffusion operator $\mathbf{P}_1\mathbf{P}_2$ and obtain the first $K$ {right singular vectors} $u^{(2)}_i \in \mathbb{R}^n$, $i=1,\ldots,K$.

\end{enumerate}
\end{algorithm}

{
\subsection{Some facts about the spectrum}
The discretization of \ac{AD} is implemented by a direct multiplication of two normalized affinity matrices (kernels), i.e., $\mathbf{P}:=\mathbf{P}_1\mathbf{P}_2$, and it is not obvious that the spectral theorem can be applied to $\mathbf{P}$. %
While both $\mathbf{P}_1$ and $\mathbf{P}_2$ are symmetrizable, the asymmetric matrix $\mathbf{P}$ is not symmetrizable or normal in general, so that it does not necessarily have a real spectrum or a complete eigen-basis. 
Note that $\mathbf{P}$ is a row stochastic matrix, so that the operator norm of $\mathbf{P}$ is bounded by $1$, as both $\mathbf{P}_1$ and $\mathbf{P}_2$ are bounded by $1$. 
Moreover, by the Perron-Frobenius theory if both $\mathbf{P}_1$ and $\mathbf{P}_2$ are primitive, $\mathbf{P}$ has an eigenvalue $1$, which is simple and is the only eigenvalue with radius $1$, corresponding to the eigenvector $[1,\ldots,1]^T/n\in \mathbb{R}^n$. Since we use two positive kernels $\tilde{P}^{(1)}$ and $\tilde{P}^{(2)}$ and since all pairwise affinities are evaluated, these assumptions are satisfied. 
To be more precise, define $\mathbf{W}:={\mathbf{W}}_1\mathbf{D}_2^{-1}{\mathbf{W}}_2$, where $\mathbf{D}_i:=\text{diag}({\mathbf{W}}_i\mathbf{1})$ for $i=1,2$ and $\mathbf{1}=[1,\ldots,1]^T\in \mathbb{R}^n$. In general, $\mathbf{W}$ is asymmetric but non-negative. We thus have
\begin{equation}
\mathbf{P}=\mathbf{D}_1^{-1}\mathbf{W},
\end{equation}
and note that $\mathbf{D}_1=\text{diag}(\mathbf{W}\mathbf{1})$. In other words, the discretized \ac{AD} operator $\mathbf{P}$ could be viewed as a stochastic diffusion operator on a directed graph with imbalanced weights on edges. 
In general, although $\mathbf{W}_1$ and $\mathbf{W}_2$ could be made positive definite if the chosen kernel is Gaussian by the Bochner theorem when there are finite points, the spectrum of $\mathbf{W}$ might not be real, and hence the spectrum of $\mathbf{P}$.

\subsection{Alternating diffusion map algorithm and alternating diffusion distance via SVD}

For the asymmetric matrix $\mathbf{P}\in \mathbb{R}^{n\times n}$, we can always consider the \ac{SVD}; that is, $\mathbf{P}=\mathbf{U}\Lambda \mathbf{V}^T$, where $\mathbf{U}=\begin{bmatrix}u_1,\ldots,u_n\end{bmatrix}\in O(n)$ and $\mathbf{V}=\begin{bmatrix}v_1,\ldots,v_n\end{bmatrix}\in O(n)$ contain the left and right singular vectors, $v_i$ and $u_i$, $i=1,\ldots,n$, and $\Lambda=\text{diag}[\sigma_1,\ldots,\sigma_n]$ is a $n\times n$ diagonal matrix with the corresponding singular values, $\sigma_i$, $i=1,\ldots,n$, on the diagonal entries.
Therefore, the matrix $\mathbf{P}$ can be decomposed as $\mathbf{P}=\sum_{\ell=1}^n\sigma_\ell u_\ell v_\ell^T$. By observing that 
\begin{equation}
\|\mathbf{P}e_i-\mathbf{P}e_j\|^2=\sum_{\ell=1}^n(\sigma_\ell v_\ell(i)-\sigma_\ell v_\ell(j))^2,
\end{equation}
where $e_i\in\mathbb{R}^n$ is the unit vector with the $i$-th entry $1$, consider the following {\em \ac{AD} map}:
\begin{equation}
\Phi:i\mapsto [\sigma_1 v_1(i),\,\sigma_2 v_2(i),\ldots,\sigma_n v_n(i)]^T\in \mathbb{R}^n\,,
\end{equation}
and the corresponding {\em \ac{AD} distance} between $i$ and $j$ defined as $\|\Phi(i)-\Phi(j)\|$.
Note that we have 
\begin{equation}
\|\Phi(i)-\Phi(j)\|=\|\mathbf{P}e_i-\mathbf{P}e_j\|\,,
\end{equation}
and $\|\mathbf{P}e_i-\mathbf{P}e_j\|$ is analogous to the effective alternating-diffusion distance considered in \cite{lederman2015alternating}.
Yet, the \ac{AD} map and distance here are based on the \ac{SVD} of $\mathbf{P}$ rather than on the \ac{EVD} of $\mathbf{P}$ as presented in \cite{Coifman_Lafon:2006}\footnotemark.
\footnotetext{As described above, in contrast to the standard construction of the kernel, in \ac{AD}, $\mathbf{P}$ does not necessarily have real eigenvectors.}

The main benefit of considering $\|\Phi(i)-\Phi(j)\|$ as the \ac{AD} diffusion distance lies in the capability to handle noise. 
For example, based on recent advances in analyzing the \ac{SVD} in random matrix setup \cite{Donoho_Gavish_Johnstone:2013}, one could balance between the accuracy of the required \ac{AD} distance and the noise influence on the result via the truncation scheme. A systematic study of this direction will be reported in future work.

In general, to realize the idea of ``diffusion'' \cite{Coifman_Lafon:2006}, consider $\mathbf{P}^t$, where $t>0$. Broadly, if $\mathbf{P}$ is viewed as a transition probability matrix of some Markov chain defined on the samples, $\mathbf{P}^t$ consists of the transition probabilities in $t$ steps.
In contrast to the standard diffusion geometry framework \cite{Coifman_Lafon:2006}, which relies on the \ac{EVD} of $\mathbf{P}$ as well as on the tight connection between the \ac{EVD} of $\mathbf{P}$ and of $\mathbf{P}^t$, the singular values and singular vectors of $\mathbf{P}^t$ are not directly related to those of $\mathbf{P}$.
Denote the \ac{SVD} of $\mathbf{P}^t$ as $\mathbf{U}_t\Lambda_t \mathbf{V}_t^T$, where $\mathbf{U}_t=\begin{bmatrix}u_{t,1},\ldots,u_{t,n}\end{bmatrix}\in O(n)$, $\mathbf{V}_t=\begin{bmatrix}v_{t,1},\ldots,v_{t,n}\end{bmatrix}\in O(n)$ contain left and right singular vectors and $\Lambda_t=\text{diag}[\sigma_{t,1},\ldots,\sigma_{t,n}]$ is a $n\times n$ diagonal matrix with the corresponding singular values on the diagonal entries. Then, by the same argument as the above, we could consider the following {\em \ac{AD} map with time $t>0$}:
\begin{equation}
\Phi_t:i\mapsto [\sigma_{t,1} v_{t,1}(i),\,\sigma_{t,2} v_{t,2}(i),\ldots,\sigma_{t,n} v_{t,n}(i)]^T\in \mathbb{R}^n\,,
\end{equation}
and the corresponding {\em \ac{AD} distance with diffusion time $t>0$} as $\|\Phi_t(i)-\Phi_t(j)\|$.
As before, we have $\|\mathbf{P}^te_i-\mathbf{P}^te_j\|=\|\Phi_t(i)-\Phi_t(j)\|$. 
We remark that a similar approach is also considered in \cite{2016arXiv160803628M}, where the authors consider the time-coupled diffusion maps.

\subsection{Non-uniform sampling issue}
}

Typically, the dataset is sampled non-uniformly from the common manifold; that is, $p_1$ or/and $p_2$ might be non-constant. In the manifold learning society, it has been well known that the non-uniform sampling effect has possible negative effect \cite{Coifman_Lafon:2006}, whereas in some cases it is beneficial \cite{Nadler_Lafon_Coifman:2006}. One way to reduce the influence of the non-uniform sampling is the $\alpha$-normalization proposed in \cite{Coifman_Lafon:2006}. While the application of this normalization to our \ac{AD} setup is straightforward, here we summarize the procedure and refer readers with interest to \cite{Coifman_Lafon:2006,Singer_Wu:2017} for details. For $0\leq \alpha\leq 1$, $\epsilon>0$ and a probability density function $p$ defined on $\MM$, we define the following functions for $i=1,2$: 
\begin{align}
p^{(i)}_{\epsilon}(x) &:=\,\int_{\MM} \tilde{K}^{(i)}_{\epsilon}(x,y)p_i(y)\ud V^{(i)}(y),\quad
\tilde{K}^{(i)}_{\epsilon,\alpha}(x,y) :=\,\frac{\tilde{K}^{(i)}_{\epsilon}(x,y)}{{p^{(i)}}^\alpha_{\epsilon}(x) {p^{(i)}}^\alpha_{\epsilon}(y)},\\
d^{(i)}_{\epsilon,\alpha}(x) &:=\,\int_{\MM} \tilde{K}^{(i)}_{\epsilon,\alpha}(x,y)p_i(y)\ud V^{(i)}(y),\quad
K^{(i)}_{\epsilon,\alpha}(x,y):=\,\frac{\tilde{K}^{(i)}_{\epsilon,\alpha}(x,y)}{d^{(i)}_{\epsilon,\alpha}(x)}.  \nonumber
\end{align}
Here, $p^{(i)}_{\epsilon}(x)$ is related to the estimation of the p.d.f. $p_i$ at $x$, denoted by $\epsilon^{-d/2}p^{(i)}_{\epsilon}(x)$. 
The practical meaning of $\tilde{K}^{(i)}_{\epsilon,\alpha}(x,y)$ is a {\em new} kernel function at $(x,y)$ adjusted by the estimated p.d.f. at $x$ and $y$; that is, the kernel is ``normalized'' to reduce the influence of the non-uniform p.d.f. $p_i$. The kernel $K^{(i)}_{\epsilon,\alpha}(x,y)$ is thus another diffusion kernel associated with $\tilde{K}^{(i)}_{\epsilon,\alpha}(x,y)$. Since the proof of the $\alpha$-normalization follows the same lines as those in \cite{Coifman_Lafon:2006,Singer_Wu:2017}, we do not present it here.

\section{Application to seasonal pattern detection}
\label{Section:Seasonality}

The ability to extract the common latent manifold underlying multiple manifolds gives rise to a new approach for detecting latent seasonal patterns in time series \cite{de2011forecasting,harvey1993forecasting,harvey1997modeling,taylor2003short,gould2008forecasting}. 
While most existing methods for seasonal pattern detection are based on trigonometric or Fourier-based analysis, time frequency analysis and parametric estimation, we take a geometric data analysis standpoint. Our approach allows the detection of seasonal patterns hidden in the data and obscured by the observation modality in addition to noise. Specifically, it is designed to accommodate nonlinearities masking the information of interest. 
For example, consider a simple $1$-dimensional periodic pattern, represented by the harmonic function $\cos(2\pi\omega_0 x)$ with a single base frequency $\omega_0>0$. Suppose this seasonal pattern is distorted by an unknown nonlinear observation function $\iota$, which takes the form of $\sqrt{y}$ when $y=\cos(2\pi\omega_0 x)>0$ and $y^2$ otherwise. Even in this caricature example, neither by trigonometric function matching, Fourier-based analysis, nor parametric estimation, the wrong frequency information might be recovered. In contrast, typical time frequency analysis approaches might provide redundant information, causing ambiguity.

The primary idea relies on the observation that the {\em geometric manifold representation} of a pure $1$-dimensional seasonal pattern, which is usually represented by a simple sinusoidal process, is a $1$-dimensional sphere $S^1:=\{(\cos(\theta),\sin(\theta))^T\in\RR^2|\,\theta\in[0,2\pi)\}$, where $\theta$ denotes the intrinsic phase of the observed seasonal oscillation; that is, the intrinsic manifold of interest, which is to be recovered in order to discover the seasonal dynamics, is $\MM=S^1$.
The time series sampled from the underlying $1$-dimensional seasonal pattern, denoted by $f(t)$, can converted to a high dimensional time series, denoted by $\mathcal{X}$, in an observable space $\mathcal{S}$ by using a lag map \cite{Takens:1981}. Clearly, the geometric structure of the observed (now) high dimensional points might be different from $\MM$, as the data might be contaminated by the observation modality and the embedding process. Using the notation from Section \ref{Section:CommonManifoldAlternationDiffusion}, the sampled high dimensional data $\mathcal{X}$ can be modeled as a smooth embedding $\iota: \MM \times \mathcal{N} \rightarrow \mathcal{S}$, where $\mathcal{N}$ is a metric space describing various nuisance/interference variables as well as artifacts introduced by the conversion.
We note that while the exposition here focuses on detecting $1$-dimensional patterns, the extension of the formulation and the detection algorithm to higher dimensional seasonal patterns is straightforward.
For example, $d$ independent seasonal patterns can be represented by a $d$-torus: $\mathcal{T}^d = S^1 \times \cdots \times S^1$.

\begin{figure}[t]
    \centering
    \subfigure[]{
        \includegraphics[width=50mm,angle=0]{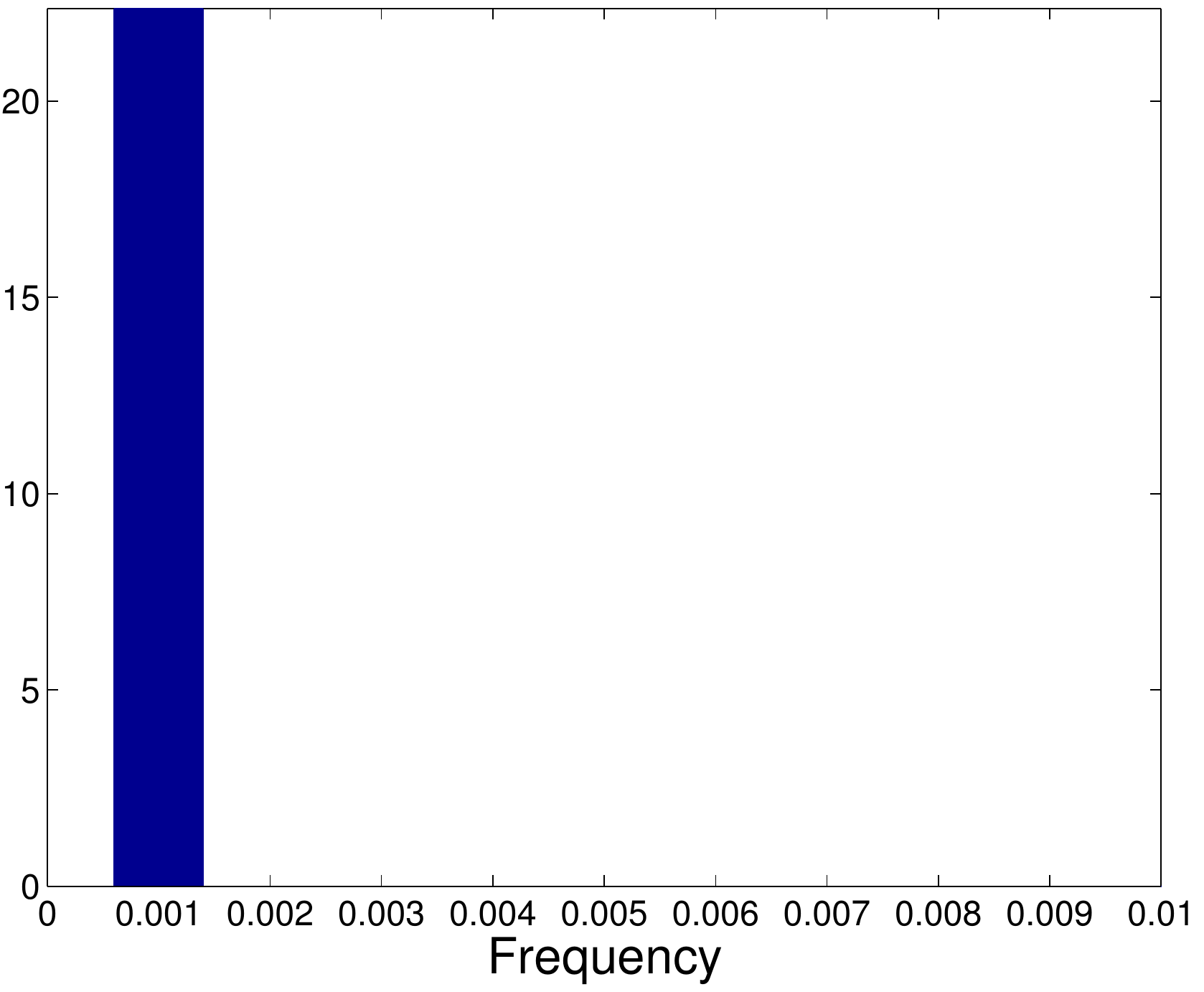}}
    \subfigure[]{
        \includegraphics[width=50mm,angle=0]{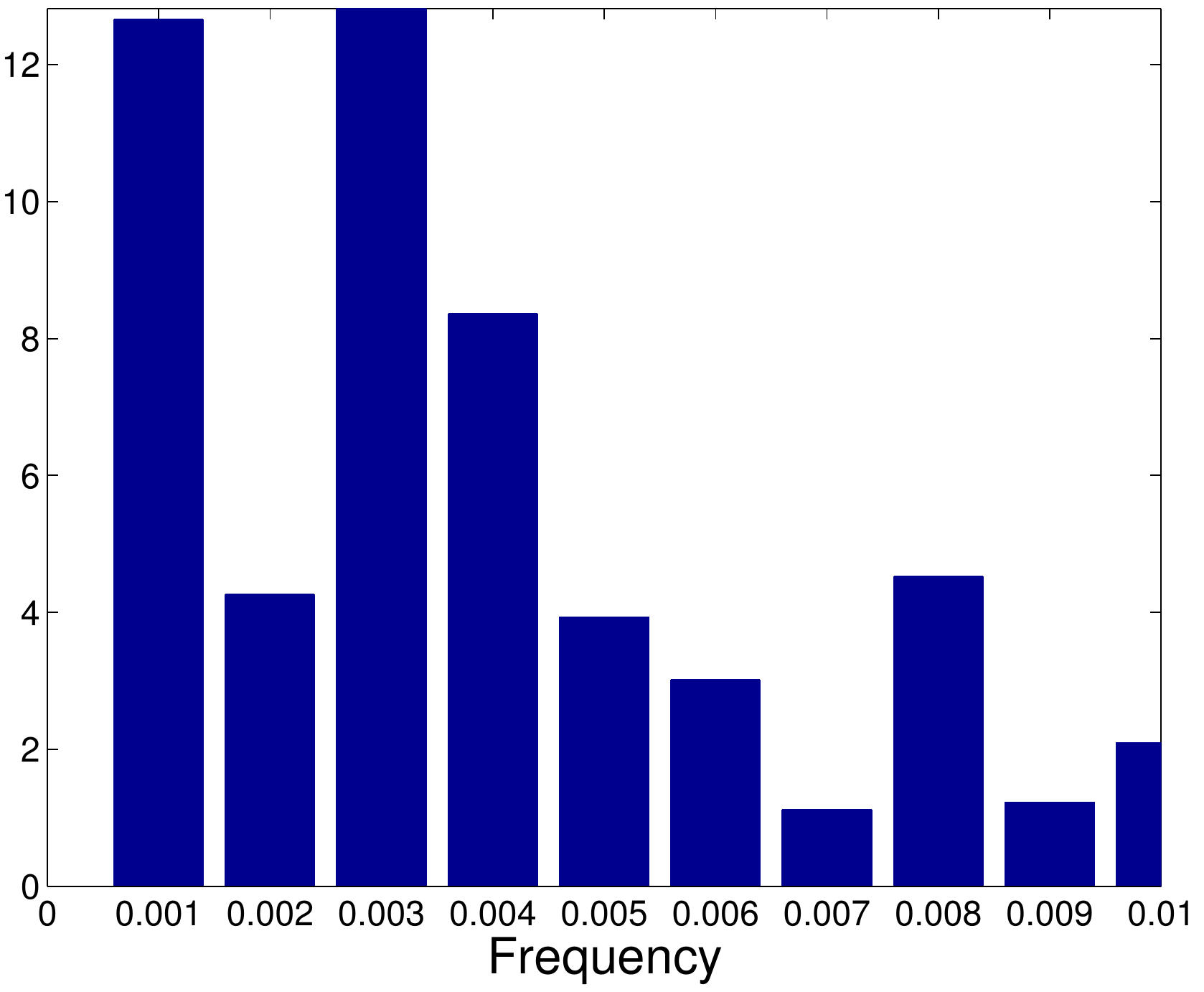}}
    \subfigure[]{
        \includegraphics[width=50mm,angle=0]{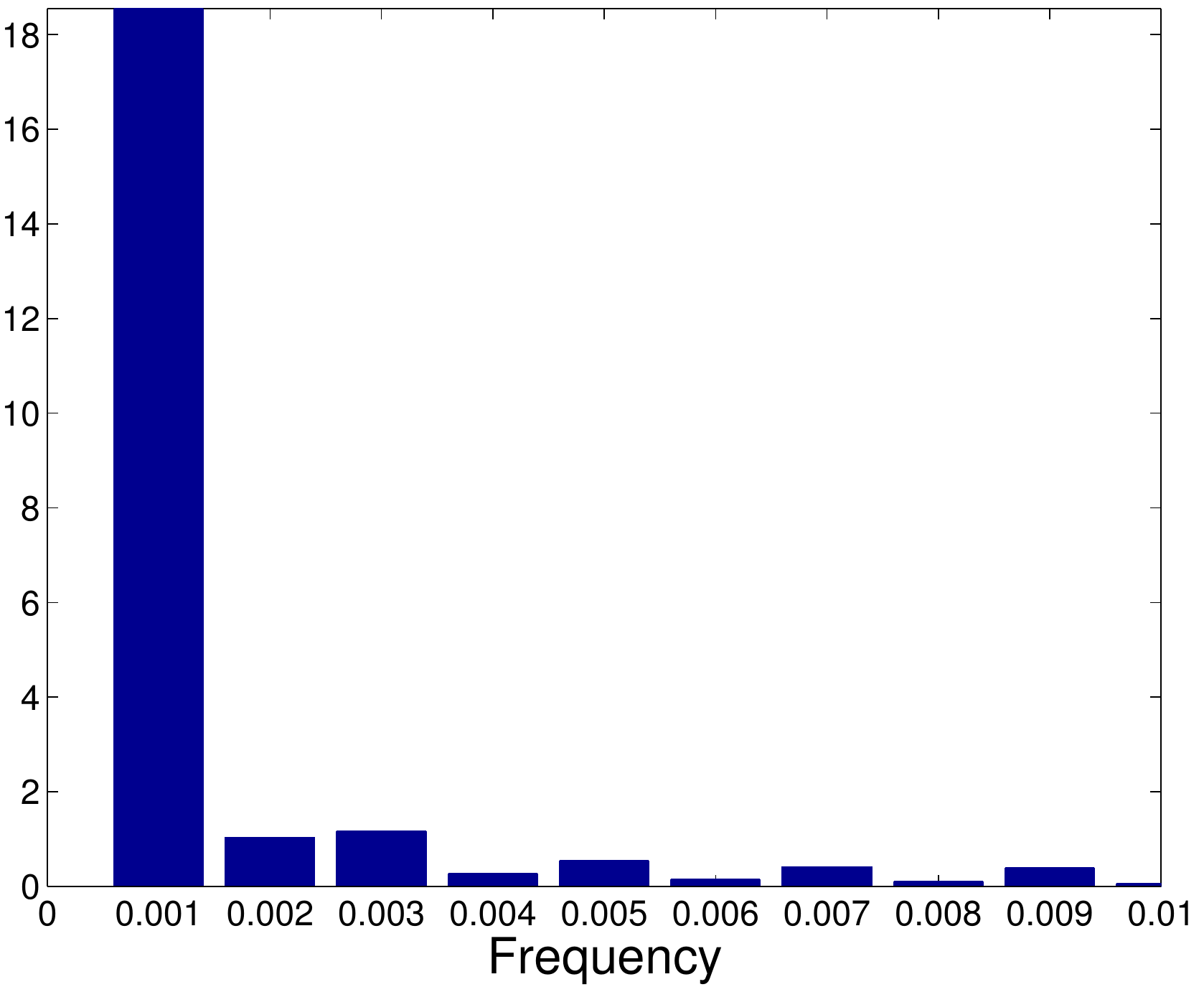}}
    \subfigure[]{
        \includegraphics[width=50mm,angle=0]{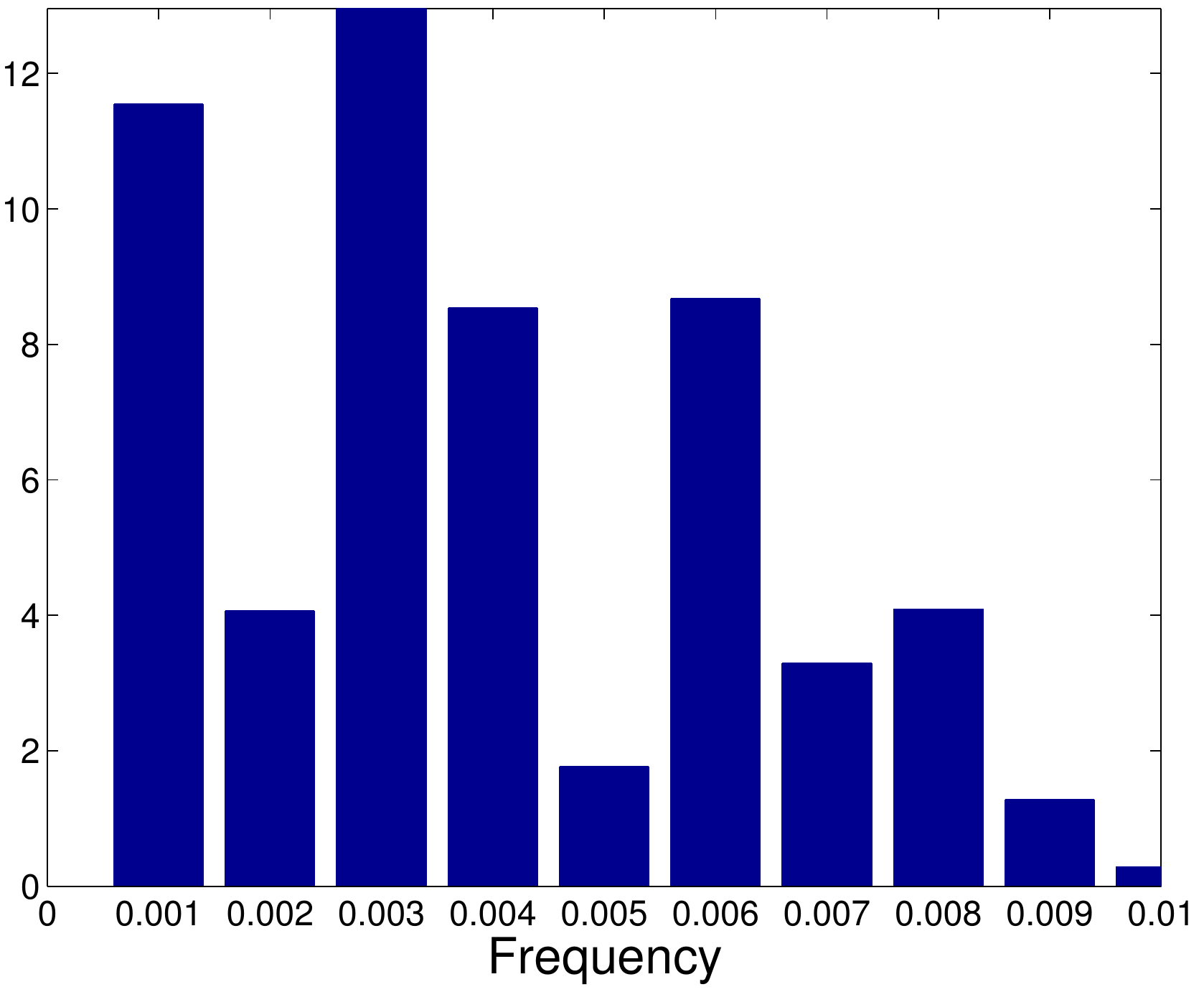}}
    \caption{The results of the application of \ac{AD} to the toy problem. (a) The Fourier transform of the eigenvector associated with the nontrivial largest eigenvalue obtained from the diffusion maps application to the reference data set $\mathcal{Y}$. We observe that indeed the Fourier transform of the eigenvector approximates the function $\delta(\omega - \omega_0)$ where $\omega _0 = 1/1000$. (b) The Fourier transform of the eigenvector associated with the largest nontrivial eigenvalue obtained from the diffusion maps application to the data set $\mathcal{X}$. We observe that indeed the eigenvector is deformed and does not only consists of a single frequency $\omega _0 = 1/1000$. (c) The Fourier transform of the eigenvector associated with the largest nontrivial eigenvalue obtained from \ac{AD} application starting from the reference set. We observe that indeed, as indicated by Theorem \ref{Theorem:MainTheorem2}, the Fourier transform of the eigenvector approximates the function $\delta(\omega - \omega_0)$ where $\omega _0 = 1/1000$. (d) The Fourier transform of the eigenvector associated with the largest nontrivial eigenvalue obtained from \ac{AD} application starting from the data set. In contrast to (c), and as indicated by Theorem \ref{Theorem:MainTheorem2}, here the eigenvector is deformed and does not solely consists of $\omega _0 = 1/1000$.}
    \label{fig:toy_seasonality}
\end{figure}

Our approach uses the following rationale: if the data at hand admit a pure $1$-dimensional seasonal pattern, then, according to Corollary \ref{Theorem:MainCorollary}, we can design a reference manifold to detect its seasonality. We consider {\em a generated reference sample set} from the {\em reference} canonical $\MM:=S^1$ associated with the frequency $\omega>0$, denoted by $\mathcal{Y}$, which is the lag map of the time series $\cos(2\pi\omega t)$. This reference data set from $S^1$ represents a pure seasonal pattern. 
To be more precise, the mapping associated with the first sensor $s^{(1)}:\MM\to\RR^p$ is determined by the lag map, where $p\in\NN$ is the chosen lag step and $\mathcal{N}_1=\emptyset$; the mapping associated with the second sensor $s^{(2)}:\MM\times \mathcal{N}_2\to\RR^p$ is the composition of the procedure generating the recorded time series from $\MM$ and the lag map with the lag step $p$, where $\mathcal{N}_2$ represents a space of possible interferences and nuisance variables introduced during the data acquisition procedure and the lag embedding process. Here, the given sample set is in $\mathcal{X}\subset \RR^p$ and the goal is to recover $\MM$, namely, the underlying seasonal pattern, using \ac{AD} from $\mathcal{X}$, the space of the data, and $\mathcal{Y}$, the space of the reference samples. 
If the given sample set $\mathcal{X}$ embodies a seasonal pattern with the base frequency $\omega$, then it assumed to lie in a space consisting of an image of $S^1$ (representing the underlying seasonal pattern contaminated by the interference/nuisance variables). As a result, by applying \ac{AD} to obtain a parameterization of the common manifold, due to the pure seasonal pattern in the generated reference sample set, the common manifold extracted by the application of \ac{AD} to the given data set $\mathcal{X}$ and the reference set $\mathcal{Y}$ is $S^1$, i.e., the desired seasonal pattern underlying the data.
Since the given data set does not necessarily exhibit a seasonal pattern with the base frequency $\omega$ used to generate the reference set, the difference between the resulting quantities associated with the \ac{AD} over the common manifold model and the diffusion maps over the reference manifold $S^1$ can be used to define a {\em seasonality index}. 
Denote the top nontrivial eigenvector of \ac{AD} over the common manifold model as $\psi$, and the top nontrivial eigenvector of DM over the reference manifold associated with the frequency $\omega>0$ as $\psi_{\text{r}}$. Then, define the seasonality index for the inherent periodicity as
\begin{align}
\text{SI}(\omega)=\|\hat{\psi}-\hat{\psi}_{\text{r}}\|^2,
\end{align}
where $\hat{\psi}$ means the Fourier transform of $\psi$. Note that this approach is feasible since the dataset is well-ordered in time. This choice of difference is simple and efficient, yet it is shown empirically to provide sufficiently accurate results in the tested applications.
Our proposed algorithm for computing the seasonality index for the detection of $1$-dimensional seasonal patterns based on \ac{AD} is summarized in Algorithm \ref{algo:seasonality}.

\begin{algorithm}[t]
\caption{Seasonality Index Computation Based on Alternating Diffusion}

\noindent\textbf{Input:} a data set $\{ x_i \}_{i=1}^n$ in $\mathbb{R}^d$ and a tested seasonality frequency $\omega$.

\textbf{Output:} a seasonality index $\textrm{SI}(\omega)$ and a \emph{baseline} seasonality index $\textrm{SI}_{\textrm{bl}} (\omega)$.

\begin{enumerate}

\item
Build a Gaussian kernel $\widetilde{\mathbf{P}} \in \mathbb{R}^{n \times n}$ based on the given data $\{ x_i \}_{i=1}^n$.

\item
Optional, for illustrative purposes: \\ \noindent
Apply diffusion maps to $\widetilde{\mathbf{P}}$ and obtain the leading eigenvector $\psi _{bl} \in \mathbb{R}^n$.

\item
Create a synthetic reference manifold:

\begin{itemize}

\item
Simulate reference samples $r_i = \cos ( 2 \pi \omega i ), i = 1, \ldots, n$.

\item
Create a lag map of samples $y_i \in \mathbb{R}^l$ with arbitrary lag $l$ from $r_i$.

\item
Build a Gaussian kernel $\widetilde{\mathbf{P}}_r \in \mathbb{R}^{n \times n}$ from $\{ y_i \}_{i=1}^n$.

\item
Apply diffusion maps to $\widetilde{\mathbf{P}}_r$ and obtain the leading eigenvector $\psi _r \in \mathbb{R}^n$.
\end{itemize}

\item
Apply \ac{AD} to $\widetilde{\mathbf{P}}_r$ and $\widetilde{\mathbf{P}}$ and obtain the leading eigenvector $\psi \in \mathbb{R}^n$.

\item
Compute the seasonality index:
\begin{equation}\label{eq:si}
	\textrm{SI} (\omega) = \| \hat{\psi} - \hat{\psi}_r \|^2
\end{equation}
where $\hat{v}$ denotes the discrete Fourier transform of $v$.

\item
Optional, for illustrative purposes: \\ \noindent
Compute the {\em baseline} seasonality index:
\begin{equation}\label{eq:sanity_si}
	\textrm{SI}_{\textrm{bl}} (\omega) = \| \hat{\psi}_{bl} - \hat{\psi}_r \|^2
\end{equation}

\end{enumerate}
\label{algo:seasonality}
\end{algorithm}

\begin{figure}[t]
    \subfigure[]{
        \includegraphics[width=130mm,angle=0]{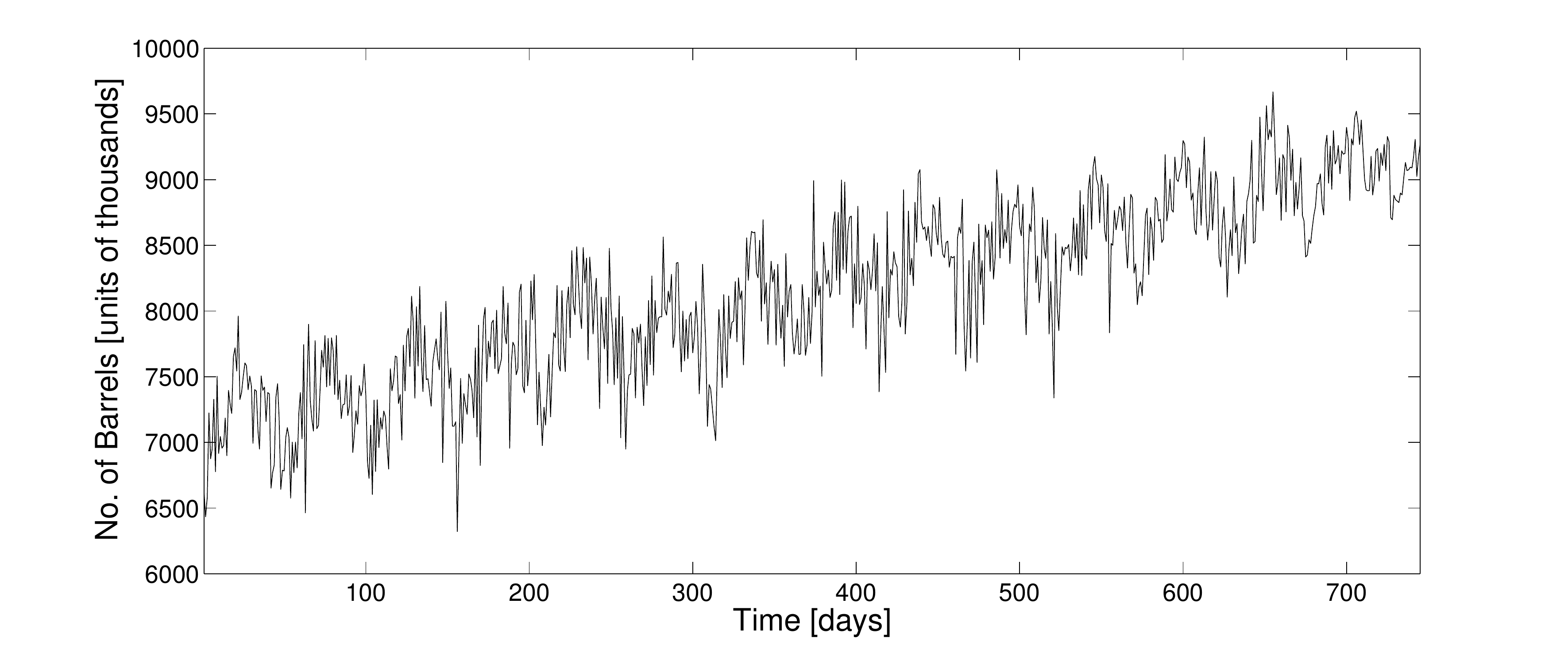}} \vfill
    \subfigure[]{
        \includegraphics[width=38mm,angle=0]{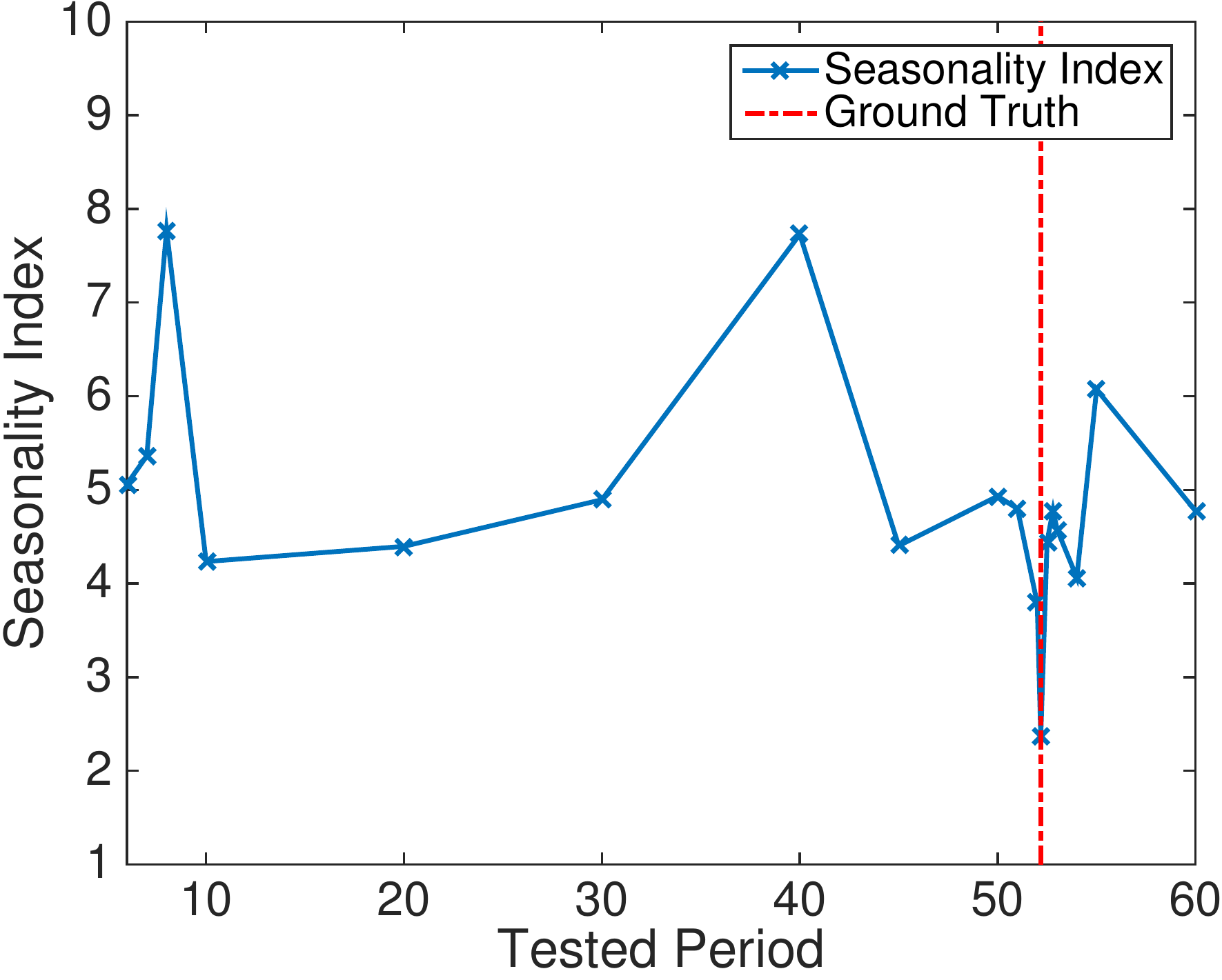}}
    \subfigure[]{
        \includegraphics[width=38mm,angle=0]{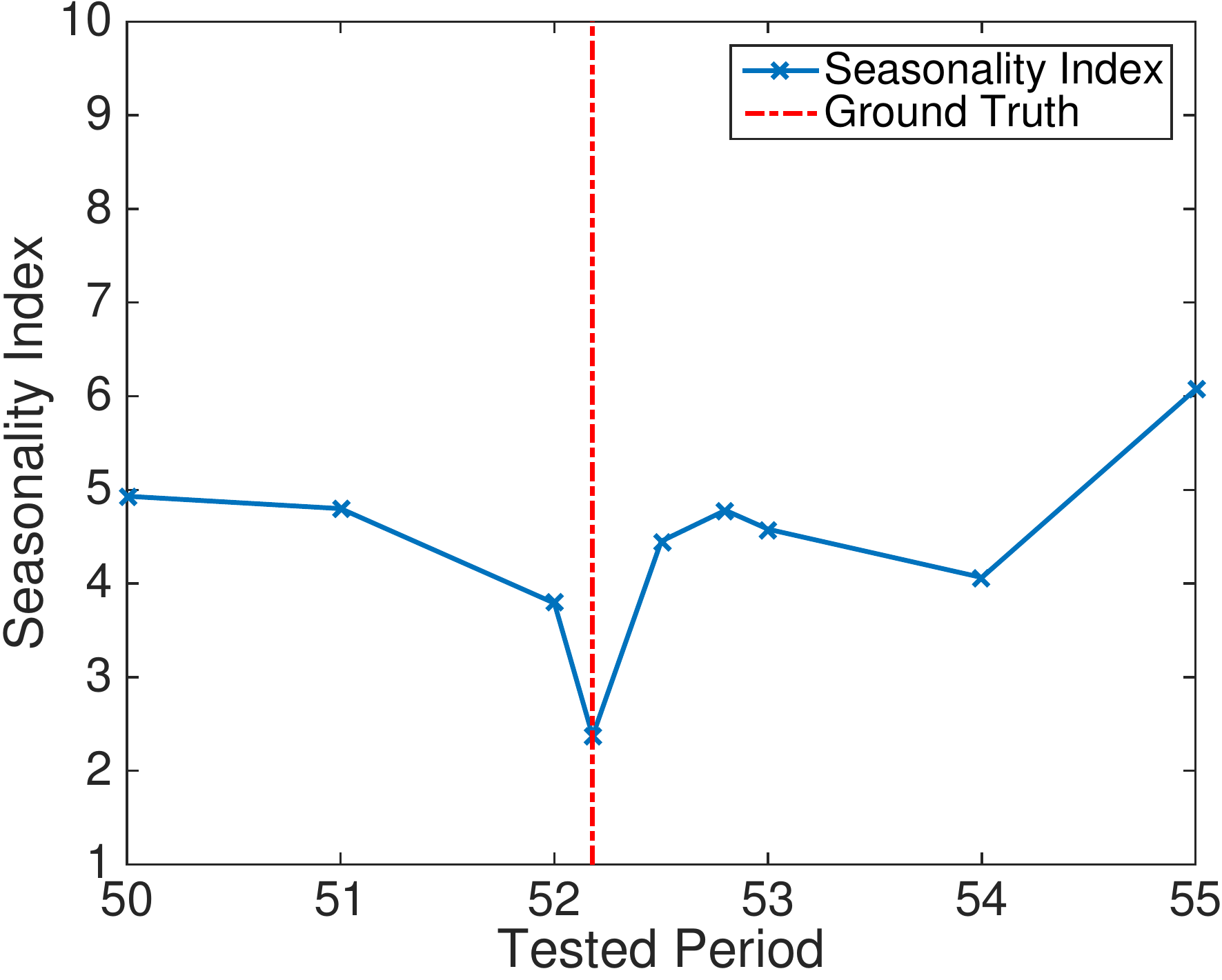}}
    \subfigure[]{
        \includegraphics[width=38mm,angle=0]{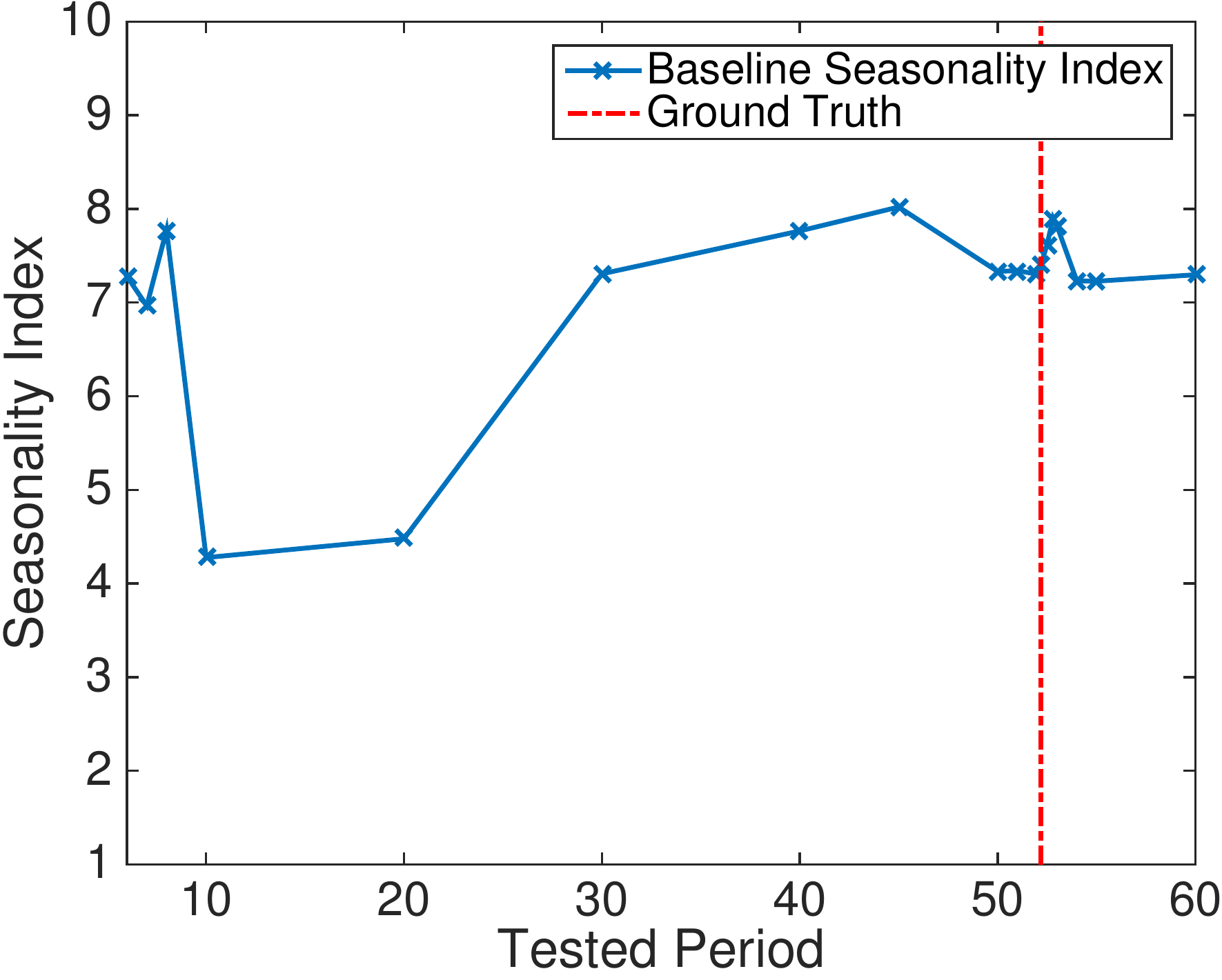}}
    \caption{The results of the application of Algorithm \ref{algo:seasonality} to the time series of weekly U.S. finished motor gasoline products supplied from February $1991$ to July $2005$ in units of thousands of barrels per day. (a) The time series of number of barrels. (b) The seasonality index as a function of the tested frequency $\omega$. (c) A zoom into the range of interest around the true frequency ($\omega_0=52.179$). (d) The baseline seasonality index as a function of the tested frequency.}
    \label{fig:gasoline}
\end{figure}

To exemplify the algorithm, as well as the theoretical results, we first test our algorithm on a toy problem.
Consider a data set of $1000$ samples from a 2-dim time series given by
\[
	x_n = (\cos(2 \pi (n \Delta t + \sin(2 \pi n \Delta t))), \sin(2 \pi (n \Delta t + \sin(2 \pi n \Delta t))))
\]
where $\Delta t = 1$. Clearly, $\{ x_n \}_{n=1}^{1000}$ can be viewed as a set of samples from $S^1$ with nonuniform sampling. Alternatively, we may view the set as a pure 1-dim seasonal pattern with base frequency $\omega_0 = 1/1000$, distorted by a periodic function.

According to Algorithm \ref{algo:seasonality}, we create the following $1000$ reference samples 
\[
	r_n = (\cos ( 2 \pi \omega_0 n ), \sin ( 2 \pi \omega_0 n ))
\]
which can be viewed as uniform samples from $S^1$.

By Theorem \ref{Theorem:MainTheorem2}, where the common manifold is $\mathcal{M} = S^1$, the metrics are equal ($\iota ^{(1)} = \iota ^{(2)}$), and the sampling density on the reference manifold is uniform, the effective \ac{AD} operator starting from the reference set $\mathcal{Y}$ is given by
\[
Tf(x)=\,f(x)+\epsilon \frac{\mu^{(1)}_{2,0}}{2}\Delta^{(1)} f(x)+O(\epsilon^{3/2})
\]
which is the Laplace-Beltrami operator. As a result, the eigenvectors of the \ac{AD} operator approximate the eigenfunctions of the Laplace-Beltrami operator, and in particular in 1-dim, the eigenvector associated with the largest eigenvalue of the \ac{AD} operator (or the smallest eigenvalue of the Laplace-Beltrami operator) is $\psi(n) \approx \sin (\omega _0 n)$, and hence, $\hat{\psi} (\omega) \approx \delta (\omega - \omega _0)$, where $\delta$ is the Delta function.

In contrast, the effective \ac{AD} operator, starting from the data set $\mathcal{X}$, where the nonuniform density is given by $p(x) = 1 / (2 \pi (1 +  2 \pi \cos (2 \pi x))$ is given according to Theorem \ref{Theorem:MainTheorem2} by
\[
Tf(x)=\,f(x)+\epsilon \frac{\mu^{(2)}_{2,0}}{2}\Big[\Delta^{(2)} f(x)+\frac{\nabla^{(2)} p(x)\cdot \nabla^{(2)} f(x)}{p(x)}\Big]+O(\epsilon^{3/2})
\]
which is a deformed Laplace-Beltrami operator. As a result, the eigenvector associated with the largest eigenvalue of the \ac{AD} operator is no longer a sinusoid with a single base frequency $\omega _0$.

Figure \ref{fig:toy_seasonality} presents the experimental results of such two \ac{AD} applications, exemplifying our theoretical results, and in particular, the effect of the order of the single-view operators composing the \ac{AD} operator. In addition, it shows the capability of Algorithm \ref{algo:seasonality} to accurately extract a seasonal pattern from a data set despite being masked.
Note that as discussed in Section \ref{Section:ADAnalysis}, since the common manifold components observed in the data set and in the reference set are not identical, when comparing the eigenvectors, we need to pull back the resulting eigenvector according to the diffeomorphism. {The Matlab code of the numerical implementation is available here: https://github.com/ronenta2/alternating-diffusion.git.}

Next, we test our algorithm on a time series of weekly U.S. finished motor gasoline products supplied from February
$1991$ to July $2005$ in units of thousands of barrels per day. Such a time series has an annual seasonal pattern with non-integer period $\omega_0 = 365.25/7 \approx 52.179$ \cite{de2011forecasting}, { and it is depicted in Figure \ref{fig:gasoline}~(a)}. 

In Fig. \ref{fig:gasoline}~(b) we plot the seasonality index, computed according to \eqref{eq:si}, as a function of the tested frequency $\omega$, ranging from $5$ to $60$. Figure \ref{fig:gasoline}~(c) zooms into the range of interest around the true base frequency ($\omega_0=52.179$). As we observe, our index indeed attains the minimal value at the true (non-integer) based frequency, and thus, accurately identifies the seasonal pattern underlying this data set.
To highlight the contribution of the \ac{AD} in extracting the seasonality pattern hidden in the data, we compare the proposed seasonality index with a naive baseline index, computed based on the comparison between the leading eigenvectors of the reference manifold and the data obtained by two separate DM applications (without applying \ac{AD}). This baseline index is designed to show that the seasonal pattern is hidden in the given data and cannot be identified simply from the leading eigenvector obtained by, for example, a direct DM application to the given data set. Figure \ref{fig:gasoline}~(d) depicts this naive index, computed according to \eqref{eq:sanity_si}, as a function of the tested frequency. In contrast to our index, this index does not capture the underlying annual seasonal pattern.

In Figs. \ref{fig:gasoline}~(b)-(c) we observe that the correct seasonality pattern is distinctly identified. In contrast, Fig. \ref{fig:gasoline}~(d) shows that the application of \ac{AD} is critical; it implies that the seasonality pattern is not the only information hidden in the data and cannot be simply extracted without the ``filtering" of possible nuisance variables and acquisition procedure deformations attained by the \ac{AD} procedure.

\section{Application in sleep research}\label{Section:Sleep}

\begin{figure}[t]
    \centering
    \subfigure[]{
        \includegraphics[width=120mm,angle=0]{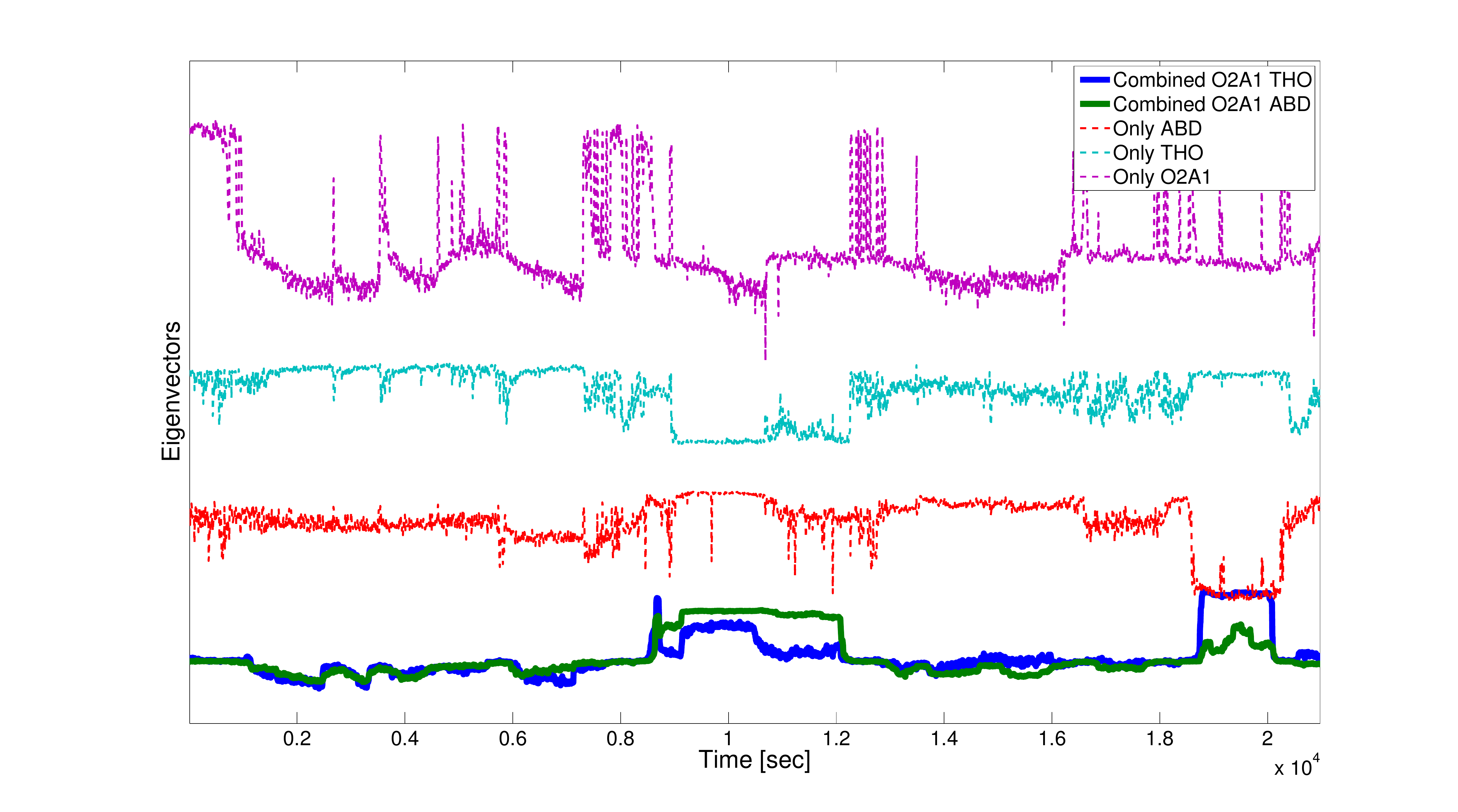}}
    \subfigure[]{
        \includegraphics[width=120mm,angle=0]{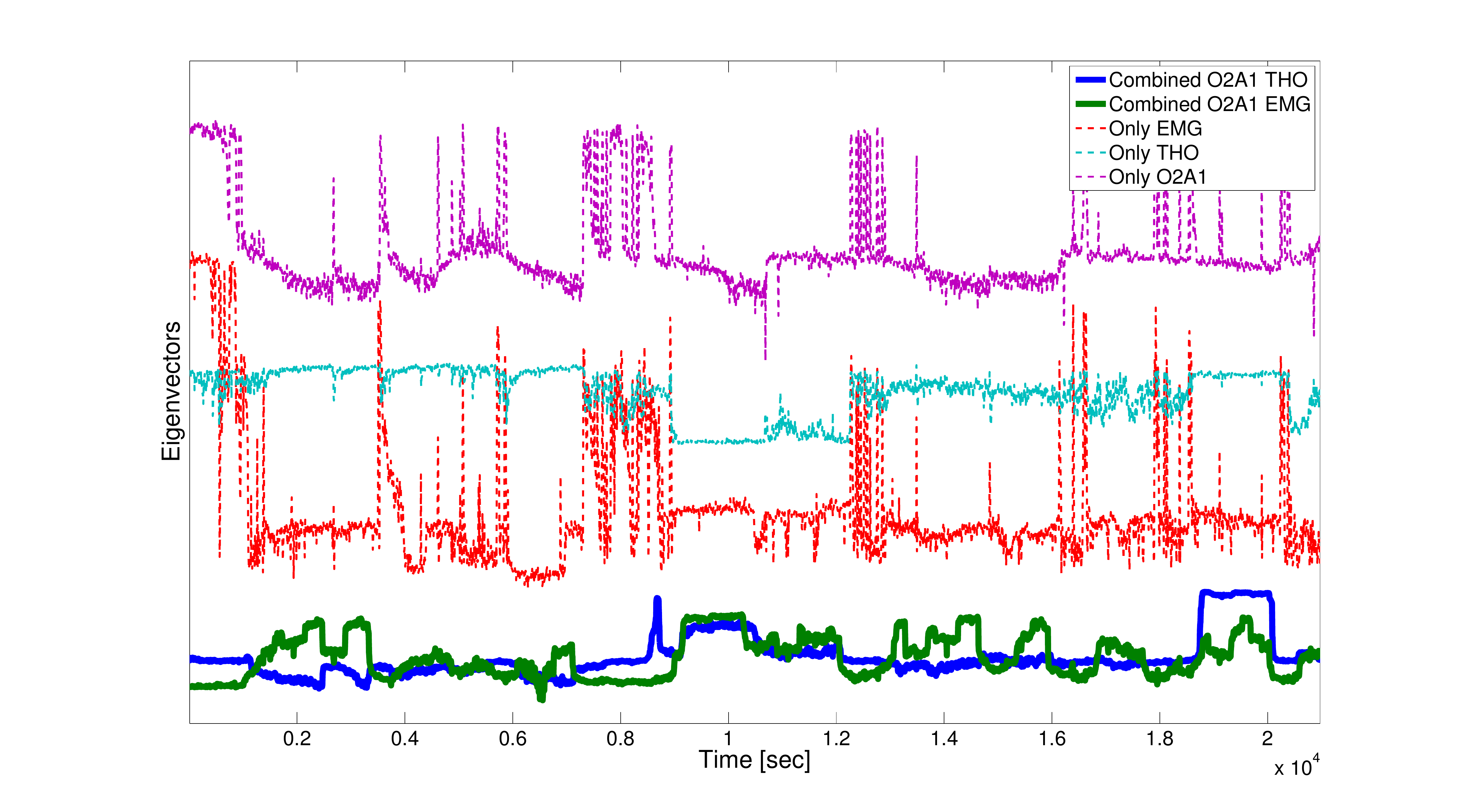}}
    \caption{The leading nontrivial eigenvectors obtained from \ac{AD} applied to EEG and respiratory signals as a function of time. (a) The dashed lines represent the leading nontrivial eigenvectors obtained by diffusion maps applied separately to an EEG recording from a single electrode (O2A1), a signal acquired by a motion belt located on the chest (THO), and a signal acquired by a motion belt located on the abdomen (ABD). The solid blue line is the leading nontrivial eigenvector resulting from \ac{AD} applied to the EEG recording and the signal from the motion belt on the chest. The solid green line is the leading nontrivial eigenvector resulting from \ac{AD} applied to the EEG recording and the signal from the motion belt on the abdomen.
    (b) The dashed lines represent the leading nontrivial eigenvectors obtained by diffusion maps applied separately to an EEG recording from a single electrode (O2A1), a signal acquired by a motion belt located on the chest (THO), and an EMG recording measuring muscle movements. The solid blue line is the leading nontrivial eigenvector resulting from \ac{AD} applied to the EEG recording and the signal from the motion belt on the chest. The solid green line is the leading nontrivial eigenvector resulting from \ac{AD} applied to the EEG recording and the EMG recording.}
    \label{fig:sleep}
\end{figure}

\begin{figure}[th]
    \centering
    \subfigure[]{
        \includegraphics[width=120mm,angle=0]{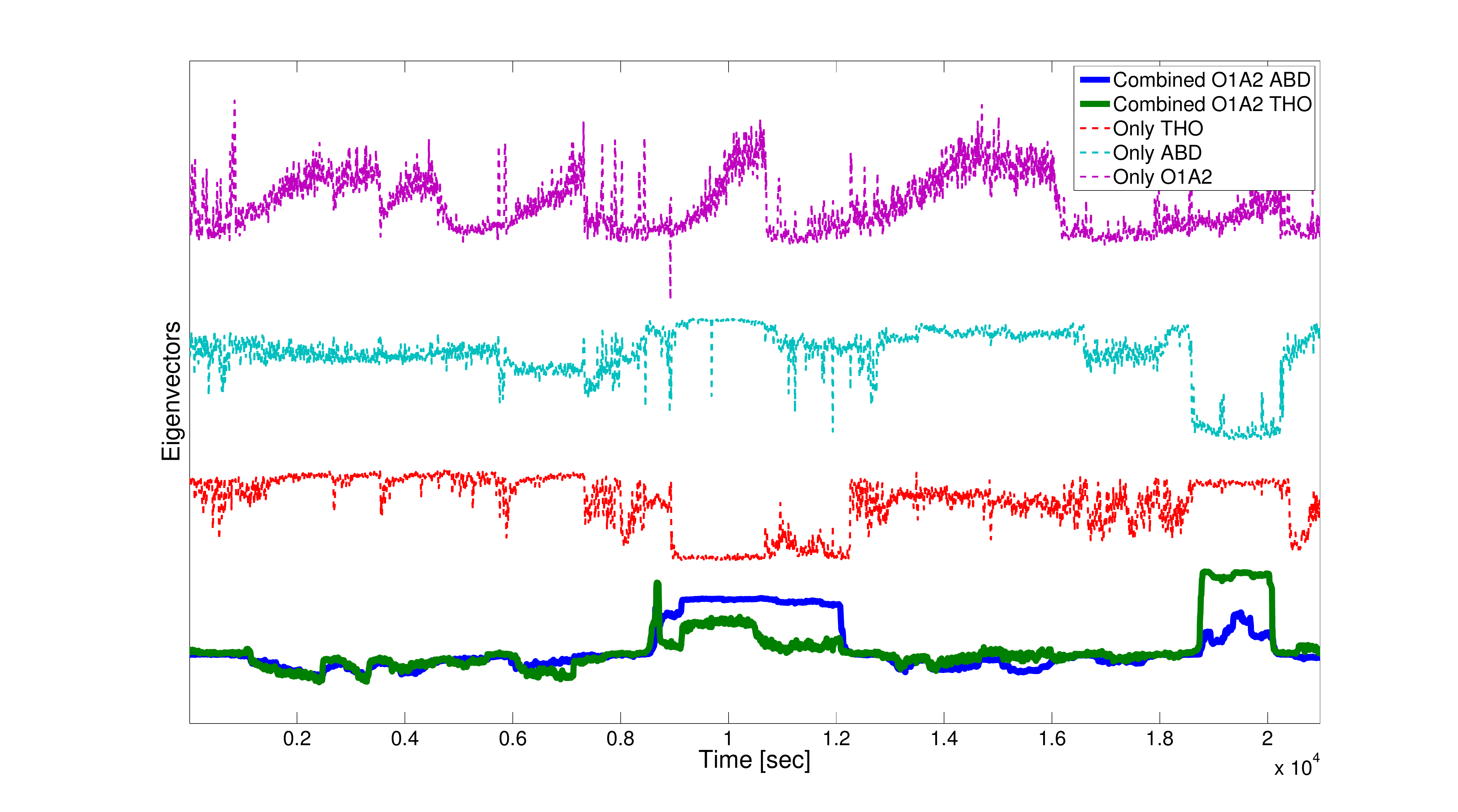}}
    \subfigure[]{
        \includegraphics[width=120mm,angle=0]{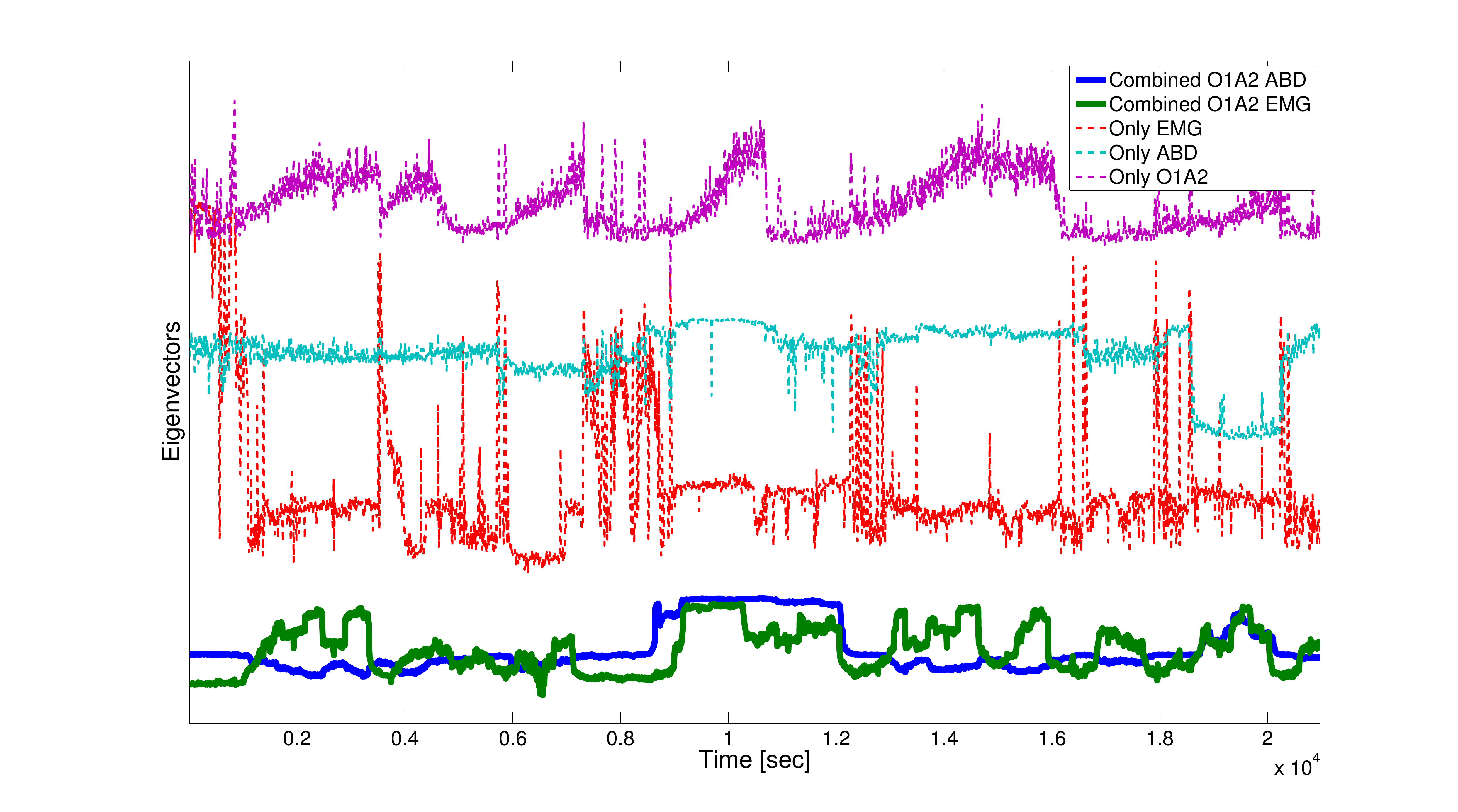}}
    \caption{The leading nontrivial eigenvectors obtained from \ac{AD} applied to EEG and respiratory signals as a function of time. Same as in Figure~\ref{fig:sleep} only with a different EEG channel (O1A2).}
    \label{fig:sleep_b}
\end{figure}

The extraction of the common manifold underlying two (or more) data sets from different sensors can be viewed as, and used for, nonlinear manifold filtering. 
To present the main idea and demonstrate its potential, we apply our technique to sleep data.
Sleep is global and systematic physiological activity, which manifests complicated temporal dynamics \cite{Lee-Chiong:2008}. Since the sound dynamics associated with sleep is associated with versatile diseases, a full understanding of the sleep dynamics is critical. For this purpose, different signals from different sensors are typically recorded, aiming to measure different physiological aspects. For example, in many specialized sleep labs, various signals such as EEG, ECG, respiration or EMG are recorded.

Recently, we studied the problem of automatic sleep stage identification from a geometric analysis/manifold learning perspective.
In \cite{wu2015assess}, we considered a model in which there exists a hidden process restricted to a low-dimensional Riemannian manifold governing these measured signals, whereas we only have access to the recorded time series/signals, which represent physiological processes deformed by the observation procedures. This model, which is referred to as the empirical intrinsic geometry (EIG), is motivated by the assumption that each measurement (e.g., an EEG channel measuring brain activity or a chest belt measuring respiration) can be affected by numerous factors related to the equipment (e.g., the specific type of sensors and their exact positions) or to noise, but we have interest in the true intrinsic variable associated with the sleep cycle.

In \cite{lederman2015icassp}, we extended this model to multiple sensors and showed that through \ac{AD}, 
we can better capture information on the sleep dynamics from multimodal respiratory signals, compared with the analysis based only on a single respiratory signal.
More specifically, we applied \ac{AD} to extract the common source of variability in abdominal motions and in airflow.
We showed that the common variable of these two respiratory signals (measured by different types of instruments) recovers a true physiological hidden process that is well correlated with the sleep stage. The underlying assumption is that physiologically there is a dominant controller of the respiratory process, which cannot be explicitly modeled or accessed in practice. Yet, this controller can be observed by different sensors, for example, from the chest belt movement or the air flow signal. On the one hand, different observations capture different, complementing aspects of the information about this controller. On the other hand, this information is deformed in different manners depending on each specific sensor.

Here, we further demonstrate the usefulness of \ac{AD}, now from a (nonlinear manifold) filtering perspective. Moreover, we attempt to devise a proof-of-concept example and to associate the common manifold and the output of \ac{AD} with physiology.
We apply \ac{AD} (separately) twice to two pairs of channels.
In the first application, we use an EEG channel (O2A1) and a belt sensor located on the chest measuring respiration movements (THO). In the second application, we use the same EEG channel (O2A1) and a different belt sensor located on the abdomen (ABD), designed to measure respiration movements as well. See \cite[Section III.A]{wu2015assess} for more details about the experimental setup and measurements.

The leading eigenvectors resulting from the two \ac{AD} applications are presented in Figure~\ref{fig:sleep}(a).
It is visually evident that the two leading eigenvectors resulting from the two applications follow the same patterns, 
%
a fact that reflects the validity of the assumption that there is a common controller guiding the entire system. In other words, while the common controller information is masked by each modality (measurement/sensor), the common controller is extracted by \ac{AD}.
%
%
Importantly, note that in both applications, we can view the \ac{AD} procedure as a nonlinear filtering of the manifold underlying the EEG channel from nuisance factors by using the respective manifolds underlying the two respiratory signals. 
%

To further examine this model, we repeat the procedure and apply \ac{AD} to a third pair of measurements. The new pair consists of the same EEG channel (O2A1) and an EMG channel measuring muscle movements.
Figure~\ref{fig:sleep}(b) presents the leading eigenvectors of \ac{AD} applied to the first pair (O2A1 and THO) and the third pair (O2A1 and EMG).
In contrast to Figure~\ref{fig:sleep}(a), we now observe that the two leading eigenvectors resulting from the two applications of \ac{AD} share less similar patterns. For comparison purposes, we report the $\ell _2$ distance between the obtained eigenvectors: the distance between the two eigenvectors in Figure~\ref{fig:sleep}(a) is $0.8$ and the distance between the two eigenvectors in Figure~\ref{fig:sleep}(b) $1.2$.
According to our manifold filtering interpretation, this result indicates that the manifold filtering of the EEG channel via the respiratory signal and the manifold filtering via the EMG signal are different. This finding can be physiologically explained -- the controller common to the EEG signal and the EMG signal is different from the controller common to the EEG signal and the respiratory signal, since these sensors are monitoring different parts of our physiological system.
In summary, the results indicate that the respiratory signal (THO) and the EMG signal lie on different manifolds, which contain different physiological information with respect to the EEG signal, or more generally, with respect to the brain activity. 

To validate the consistency of the results and the above statement, in Figure~\ref{fig:sleep_b} we repeat the experiment and apply the filtering to a different EEG channel (O1A2), which is different from O2A1, and as expected we observe similar results.
In Figure~\ref{fig:sleep_b} as well, the two applications of the manifold filtering of the EEG channel through the two respiratory signals yield a similar representation, whereas the manifold filtering of the EEG channel through the EMG signal yields a different representation. This again implies the existence of a common manifold hosting the common controller underlying the two different respiratory signals associated with the respiratory system. In addition, this common controller is different from the common controller of the EEG and the EMG signals.
We report that the same procedure was applied to the same signal recordings from $10$ different subjects/cases and the results were consistent.

{
The application to sleep data demonstrates the potential of \ac{AD} in filtering. More results illustrating the advantage of \ac{AD} in obtaining signal representation that is well correlated with the sleep stage are presented in \cite{lederman2015icassp}. Further details, discussions, and the presentation of the full experimental study are beyond the scope of this paper and will appear in a future publication dedicated to this sleep application. 
}

\section{Extending the Common Manifold Representation}
\label{Section:Extension}

One of the key questions in manifold learning is how to extend the manifold representation to new data samples, once the manifold representation has been built from an initial sample set, without applying the entire, usually computationally demanding construction procedure \cite{coifman2006geometric,lafon2006data}. Since many of the manifold learning representations are given by the eigenvectors of a particular kernel, the Nystr\"{o}m extension is typically used to achieve this goal \cite{nystrom1929praktische,fowlkes2004spectral}. 
Consider data on a manifold $\MM$ with probability measure $\nu(x)$, and a pairwise affinity kernel $P$. Assume that the kernel has an \ac{EVD},
\begin{equation}\label{eq:evd}
	\psi_k (x) = \frac{1}{\lambda_k} \int _{\MM} P(x, x') \psi_k (x') d\nu(x'), \ x \in \MM, 
\end{equation}
where $\psi_k$ and $\lambda_k$ are the eigenfunctions and eigenvalues of $K$, respectively. 
We can see from \eqref{eq:evd} that one of the properties of the \ac{EVD} is that the eigenfunction at every point can be written as a (nontrivial) linear combination of the eigenfunction values. Nystr\"{o}m extension exploits this property; given a new data sample $\tilde{x}$, either on $\mathcal{M}$ or close to $\mathcal{M}$, each eigenfunction at that point is therefore approximated by
\begin{equation}\label{eq:nystrom}
	\psi _k (\tilde{x}) = \frac{1}{\lambda_k} \int _{\MM} P(\tilde{x}, x') \psi _k (x') d\nu (x').
\end{equation}
While the Nystr\"{o}m extension is usually applied in order to reduce the number of computations stemming from the \ac{EVD}, in the context of \ac{AD} it gives rise to an important benefit in addition to being more efficient computationally.
In the context of \ac{AD}, we are given two data sets on two observable manifolds, which are functions of realizations of triplets of hidden variables $(x,y,x)$ from $\mathcal{M} \times \mathcal{N}_1 \times \mathcal{N}_2$ with joint measure $\nu(x,y,z)$. Recall that we denote by $(x,y,z), (x',y',z'), (x'', y'', z'')$ three realizations of the hidden variables, by $s_1=s^{(1)}(x,y,z), s_1'=s^{(1)}(x',y',z'), s_1''=s^{(1)}(x'',y'',z'')$ their corresponding samples in the observable space $\mathcal{S}^{(1)}$, and by $s_2=s^{(2)}(x,y,z), s_2'=s^{(2)}(x',y',z'), s_2''=s^{(2)}(x'',y'',z'')$ their corresponding samples in the the observable space $\mathcal{S}^{(2)}$.
By \eqref{Definition:Pepsilon}, the \ac{AD} kernel is written as a composition of two kernels
\begin{align}\label{eq:AD_kernel2}
&P((x,y,z),(x'',y'',z''))\\
=&\,\int_{\MM\times\mathcal{N}_1\times\mathcal{N}_2}P^{(2)}((x,y,z),(x',y',z'))P^{(1)}((x',y',z'),(x'',y'',z'')) d\nu(x',y',z'),\nonumber
\end{align}
where, by \eqref{notation:Kh_def}, $P^{(1)}$ is computed based on the observations $s_1'$ and $s_1''$, and $P^{(2)}$ is computed based on the observations $s_2$ and $s_2'$.

Substituting \eqref{eq:AD_kernel2} into the Nystr\"{o}m extension \eqref{eq:nystrom} yields
\begin{align}\label{eq:alt_diff_nystrom}
	\psi _k (\tilde{x},\tilde{y},\tilde{z}) = \frac{1}{\lambda_k} \int \int_{\MM\times\mathcal{N}_1\times\mathcal{N}_2}&P^{(2)}((\tilde{x},\tilde{y},\tilde{z}),(x',y',z'))P^{(1)}((x',y',z'),(x'',y'',z'')) \\
	& \times d\nu(x',y',z') \psi _k (x'', y'', z'') d\nu (x'', y'', z'') \nonumber
\end{align}
Thus, in order to get an approximation of the eigenfunction of the \ac{AD} kernel at a new sample $(\tilde{x},\tilde{y},\tilde{z})$, we only need to compute the kernel values $P^{(2)} ((\tilde{x},\tilde{y},\tilde{z}),(x',y',z'))$, i.e., the kernel between the new sample $(\tilde{x},\tilde{y},\tilde{z})$ in the observable space $\mathcal{S}^{(2)}$ and all the existing samples $(x',y',z')$ in $\mathcal{S}^{(2)}$.
In other words, the extension only requires a new sample from {\em only one of the sensors}, which is an important benefit in many applications, in particular, when one sensor is more difficult to obtain than the other.

\begin{figure}[t]
    \centering
    \subfigure[]{
        \includegraphics[width=60mm,angle=0]{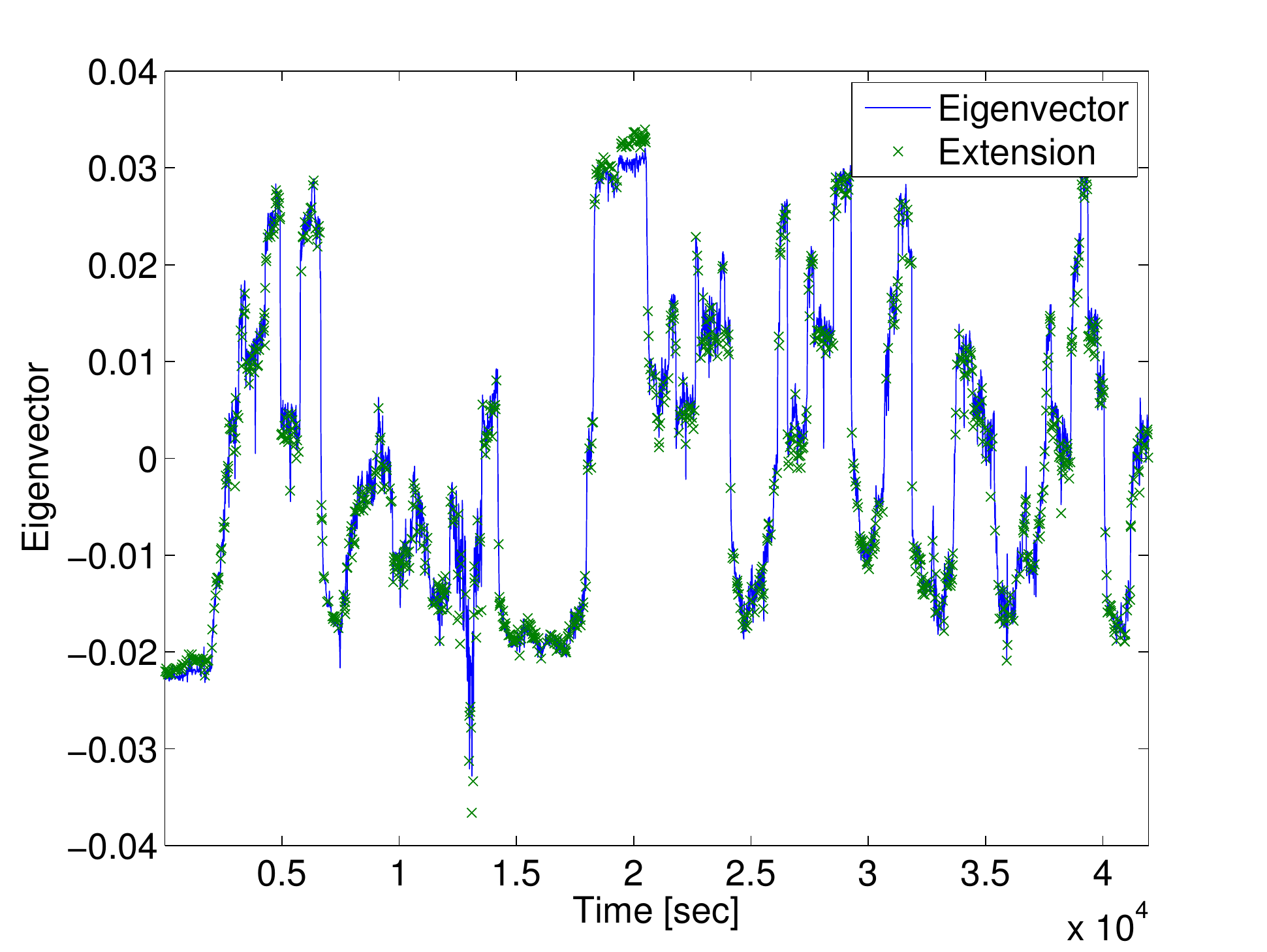}}
    \subfigure[]{
        \includegraphics[width=60mm,angle=0]{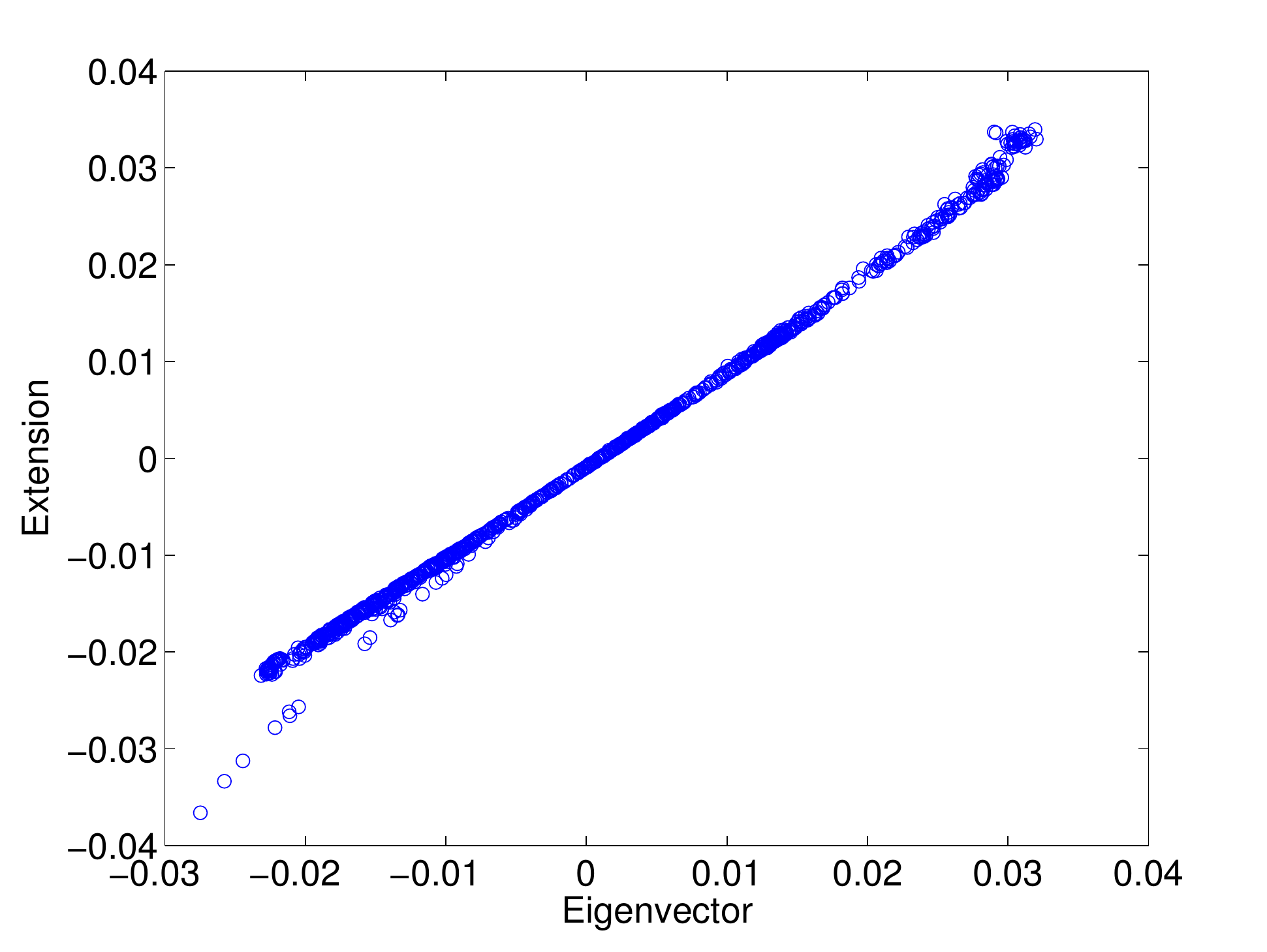}}
    \caption{The extension of the leading nontrivial eigenvector of alternating diffusion applied to an EEG recording (O2A1) and the EMG recording. In (a) we plot the true eigenvector and its extension as a function of time, and in (b) we plot the scatter plot of the true eigenvector samples vs. the extended eigenvector samples.}
    \label{fig:ext}
\end{figure}
We demonstrate the extension on the sleep application presented in Section \ref{Section:Sleep}.
We randomly take only $75\%$ of the samples from an EEG recording (O2A1) and an EMG recording. In the initialization stage, we construct the common manifold representation based on these samples. Then, in an extension stage, we recover the common manifold representation for the remaining $25\%$ from solely the EMG samples. In Figure \ref{fig:ext} we present both the true manifold representation and the extension we obtain. We observe that after the initial stage, which requires both measurements, we can accurately recover the common manifold representation from only one of the measurements; that is, we do not have to evaluate the extension for two kernels, but rather just for one. Specifically, in the sleep application, this result may have practical implications since EEG is considered the most informative recording, yet it is also difficult to collect and usually requires the help of a specialized technician. In addition, the filtering of the EEG recording via the EMG recording enables us to clean the EEG signal from artifacts and noise and to show high correspondence with the sleep stage. The extension result presented here implies that after a ``calibration" stage, in which both recordings are required, we can get access to the filtering result only based on the EMG recording, which is easier to collect.

\section{Conclusions}\label{Section:Conclusion}

In this paper, we introduced a common manifold model underlying sensor observations for the purpose of fusion multimodal sensor data. To study this problem, we proposed a method based on alternating diffusion and provided theoretical analysis under the common manifold model. Compared with traditional methods, this method is able to capture the nonlinear manifold structure in the sense that both the topological and geometrical structures are simultaneously preserved. Several applications were provided, demonstrating the power of the proposed method and the analytic tools.

One important aspect of the current framework is its extension to more than two sensors. Note that a straightforward extension of the proposed method to fuse data from multiple sensors discovers the common manifold underlying all the sensors. Despite being efficient and appropriate for certain purposes, in some sensor fusion applications it might erase substantial information. 
In addition, the notion of nonlinear manifold filtering, which is only introduced here, along with its filtering capabilities call for further analysis and more substantial mathematical foundations. For example, based on the results presented here, in a future work, we will present a filtering scheme that combines the common manifold representations obtained from all the possible pairs of sensors in a multimodal experiment. This scheme enables us to combine relevant information from multimodal sensors in a pure data-driven manner, while filtering nuisance phenomena measured only in single sensors, and thus, regarded as less important.

Another future directions will address applications and will focus on the following three issues. First, we have only shown a proof-of-concept to the potential association between the common manifold and physiological properties. A large scale study with statistical analysis is needed to confirm the findings. Second, {we may consider a natural generalization by taking the symmetrization into account. Based on the discussion in Section \ref{Section:ManifoldNuisance}, it is clear that the proposed \ac{AD} scheme could provide a more accurate information when the \ac{AD} ends in a sensor with the nuisance that can be well modeled by a manifold. The comparison of different schemes will be further explored. Third,} following the automatic annotation work in \cite{wu2015assess}, we will investigate the correspondence between the obtained representations of the common manifolds, and the sleep stage and we will devise improved automatic stage identification methods.

\section{Acknowledge}

The authors would like to thank Professor Ronald Coifman {and Dr. Roy Lederman} for fruitful discussions, {and Professor Yu-Lun Lo for sharing the sleep dataset}. 
Hau-tieng Wu acknowledges the support of Sloan Research Fellow FR-2015-65363. 
Ronen Talmon was supported by the European Union's Seventh Framework Programme (FP7) under Marie Curie Grant 630657 and by the Israel Science Foundation (grant no. 1490/16).

\bibliographystyle{plain}
\bibliography{ADonCommonManifold_RT}

\clearpage

\appendix

\setcounter{equation}{0}
\renewcommand{\theequation}{A.\arabic{equation}}
    \setcounter{lemma}{0}
    \renewcommand{\thelemma}{\Alph{section}\arabic{lemma}}

\section{Proof of Theorems \ref{Theorem:ConvolutionKernel} and \ref{Theorem:MainTheorem2}}\label{Section:Appendix:Proof}

Before proceeding, we need a discussion about the diffeomorphism $\Phi$. Denote $N_x\iota^{(i)}(\MM)$ to be the $(p-d)$-dim subspace of $\RR^p$ which contains all vectors perpendicular to $d\iota^{(i)}T_x\MM$, where $d\iota^{(i)}$ is the total differential of the map $\iota^{(i)}:M\to \RR^p$, and $d\iota^{(i)}T_x\MM:=\{d\iota^{(i)}(v)|\,v\in T_x\MM\}$; that is, $u^Tv=0$ for all $u\in N_x\iota^{(i)}(\MM)$ and $v\in d\iota^{(i)}T_x\MM$. As $\MM$ is compact and smoothly embedded in $\RR^p$, for $i=1,2$, we could find a tubular neighborhood $\text{Tub}_\tau (\iota^{(i)}(\MM))\subset \RR^p$ of $\iota^{(i)}(\MM)$, where $\tau>0$ is chosen so that $\text{Tub}_\tau (\iota^{(i)}(\MM))$ is homotopic to $\iota^{(i)}(\MM)$. Then $\Phi$ could be extended to $\text{Tub}_\tau (\iota^{(1)}(\MM))$ and becomes a diffeomorphism between $\text{Tub}_\tau (\iota^{(1)}(\MM))$ and $\text{Tub}_\tau (\iota^{(2)}(\MM))$, where we use the same notation to denote the extension. After extension, $\nabla\Phi$ could be defined on $\text{Tub}_\tau (\iota^{(1)}(\MM))$. For $x\in\MM$, $\nabla\Phi|_{\iota^{(1)}(x)}$ is a linear map from $T_{\iota^{(1)}(x)}\RR^p$ to $T_{\iota^{(2)}(x)}\RR^p$, which maps $d\iota^{(1)}T_x\MM$ to $d\iota^{(2)}T_{x}\MM$ and maps $N_x\iota^{(1)}(\MM)$ to $N_{x}\iota^{(2)}(\MM)$. Note that $\Pi^{(i)}_x(u,u)\in N_x\iota^{(i)}(\MM)$ for all $u\in T_x\MM$.

Recall the following two Lemmas. Below, we will adopt Einstein summation convention to simplify the notation.

\begin{lemma}\label{volexpansion}
Fix $i=1,2$. In the normal coordinate around $x\in\MM$, when $\|v\|_{g^{(i)}}\ll 1$, $v\in T_x\MM$, the Riemannian measure satisfies
\begin{align}
\ud V^{(i)}(\exp^{(i)}_xv)=\left(1-\frac{1}{6}\Ric^{(i)}_{kl}(x)v^kv^l+O(\|v\|^3)\right)\ud v^1\wedge \ud v^2\ldots\wedge \ud v^n.
\end{align}
In the polar coordinate $v=t\theta$, where $\theta\in T_x\MM$, $\|\theta\|_{g^{(i)}}=1$, $t>0$, we have
\begin{align}
\ud V^{(i)}(\exp^{(i)}_xt\theta)=\left(t^{d-1}+t^{d+1}\Ric^{(i)}(\theta,\theta)+O(t^{d+2})\right)\ud t\ud\theta.
\end{align}
\end{lemma}
\begin{proof}
See, for example, \cite{Singer_Wu:2012}.
\end{proof}

\begin{lemma}\label{relateexp}
Fix $x\in\MM$ and $y=\exp^{(1)}_x(v)$, where $v\in T_x\MM$ with $\|v\|_{g^{(1)}}\ll 1$. We have
\begin{equation}\label{relateexp1}
\iota^{(1)}(y)=\iota^{(1)}(x)+\ud \iota^{(1)}(v)+Q^{(1)}_2(v)+Q^{(1)}_3(v)+O(\|v\|_{g^{(i)}}^4),
\end{equation}
where $\Pi^{(1)}$ is the second fundamental form of $\iota^{(1)}$, $Q^{(1)}_2(v):=\frac{1}{2}\Pi^{(1)}(v,v)$ and  $Q^{(1)}_3(v):=\frac{1}{6}\nabla^{(1)}_v\Pi^{(1)}(v,v)$.
Further, for $z=\exp^{(1)}_x(u)$, where $u\in T_x\MM$ with $\|u\|_{g^{(1)}}\ll 1$, we have
\begin{align}
\|\iota^{(1)}(z)-\iota^{(1)}(y)\|=&\,\|u-v\|_{g^{(1)}}+\frac{\|Q^{(1)}_2(u)-Q^{(1)}_2(v)\|^2}{2\|u-v\|_{g^{(1)}}}\label{relateexp1Diff}\\
&+\frac{\langle\ud \iota^{(1)}(u-v),Q_3^{(1)}(u)-Q_3^{(1)}(v)\rangle}{\|u-v\|_{g^{(1)}}}+O(\|u\|_{g^{(1)}}^4,\|v\|_{g^{(1)}}^4).\nonumber
\end{align}
In particular, suppose $v=t\theta$, where $\|\theta\|_{g^{(1)}}=1$, we have
\begin{equation}\label{relateexpGeodesic1}
\|\iota^{(1)}(y)-\iota^{(1)}(x)\|=t-\frac{\|Q^{(1)}_2(\theta)\|^2}{6}t^3+O(t^4).
\end{equation}
\end{lemma}
\begin{proof}
See, for example, \cite{Singer_Wu:2012} for the proof of (\ref{relateexp1}). By (\ref{relateexp1}) and a direct expansion, we have (\ref{relateexp1Diff}). Note that (\ref{relateexpGeodesic1}) comes from (\ref{relateexp1}) since $\|\Pi^{(1)}(v,v)\|^2+\langle \ud \iota^{(1)}(v), \nabla^{(1)}_v\Pi^{(1)}(v,v)\rangle=0$.
\end{proof}

An immediate consequence of Lemma \ref{relateexp} is the distance between two points $x$ and $y=\exp^{(1)}_x(v)$, where $v\in T_x\MM$, when measured by the ambient metric; that is, $\|\iota^{(1)}(\exp^{(1)}_x(v))-\iota^{(1)}(x)\|$. We could call the distance $\|\iota^{(1)}(\exp^{(1)}_x(v))-\iota^{(1)}(x)\|$ the {\em ambient distance} between $x$ and $y$. Next, we discuss the difference between two metrics $g^{(1)}$ and $g^{(2)}$. Since $g^{(i)}$ is induced from $(\RR^p,\texttt{can})$ via $\iota^{(i)}$, and $\iota^{(i)}$, $i=1,2$, is a smooth embedding of $\MM$ into $\RR^p$, we could find a diffeomorphism $\Phi:\RR^p\to\RR^p$ so that $\Phi\circ \iota^{(1)}=\iota^{(2)}$. Hence, the metrics $g^{(1)}$ and $g^{(2)}$ are related by  
\begin{align}
\|u\|_{g^{(2)}}^2=\langle \nabla\Phi|_{\iota^{(1)}(x)}\ud \iota^{(1)}(u),\nabla\Phi|_{\iota^{(1)}(x)}\ud \iota^{(1)}(u)\rangle=\|\nabla\Phi|_{\iota^{(1)}(x)}\ud \iota^{(1)}(u)\|^2,
\end{align}
where $u\in T_x\MM$. In the next lemma, we evaluate the ambient distance of two points on $\MM$ measured by a different metric.
\begin{lemma}\label{Lemma:distortedDistance}
Fix $x\in\MM$. To simplify the notation, we ignore the subscription of $\nabla^{(1)}\Phi|_{\iota^{(1)}(x)}$ and use $\nabla\Phi$. Similar simplification holds for ${\nabla^{(1)}}^2\Phi$, $\ud \iota^{(1)}$, $\Pi^{(1)}$, etc. Suppose $y=\exp^{(1)}_xu$, where $u\in T_x\MM$ and $\|u\|_{g^{(1)}}$ is small enough. Then we have
\begin{align}
\iota^{(2)}(y)=\iota^{(2)}(x)+\nabla\Phi\ud\iota^{(1)}u+Q^{(2)}_2(u)+Q^{(2)}_3(u)+ O(\|u\|_{g^{(1)}}^4)
\end{align}
where $Q^{(2)}_2$ and $Q^{(2)}_3$ are quadratic and cubic polynomials respectively defined by
\begin{align}
Q^{(2)}_2(u):=&\frac{1}{2}\Pi^{(2)}(u,u),\\
Q^{(2)}_3(u):=&\frac{1}{6}\nabla\Phi\nabla^{(1)}_u\Pi^{(1)}(u,u)+ \nabla^2\Phi\left(u,\Pi^{(1)}(u,u)\right). \nonumber
\end{align}
Further, for $z=\exp^{(1)}_xv$, where $\|v\|_{g^{(1)}}$ is also small enough, we have
\begin{align}
\|\iota^{(2)}(y)-\iota^{(2)}(z)\|=\,& \|\nabla\Phi\ud\iota^{(1)}(u-v)\|+\frac{\left\langle \nabla\Phi\ud\iota^{(1)}(u-v),Q^{(2)}_2(u)-Q^{(2)}_2(v)\right\rangle}{\|\nabla\Phi\ud\iota^{(1)}(u-v)\|}\nonumber\\
&\,+\frac{2\left\langle \nabla\Phi\ud\iota^{(1)}(u-v), Q^{(2)}_3(u)-Q^{(2)}_3(v)\right\rangle+\|Q^{(2)}_2(u)-Q^{(2)}_2(v)\|^2}{\|\nabla\Phi\ud\iota^{(1)}(u-v)\|}\nonumber\\
&\,+O(\|u\|_{g^{(1)}}^4,\|v\|_{g^{(1)}}^4)\nonumber.
\end{align}
In particular, suppose $v=t\theta$, where $\|\theta\|_{g^{(1)}}=1$, we have
\begin{align}
\|\iota^{(2)}(y)-\iota^{(2)}(x)\|=\,& t\|\nabla\Phi\ud\iota^{(1)}(\theta)\|+t^2\frac{\left\langle \nabla\Phi\ud\iota^{(1)}(\theta),Q^{(2)}_2(\theta)\right\rangle}{\|\nabla\Phi\ud\iota^{(1)}(\theta)\|}\label{relateexpGeodesic2}\\
&\,+t^3\frac{2\left\langle \nabla\Phi\ud\iota^{(1)}(\theta), Q^{(2)}_3(\theta)\right\rangle+\|Q^{(2)}_2(\theta)\|^2}{\|\nabla\Phi\ud\iota^{(1)}(\theta)\|}+O(t^4)\nonumber.
\end{align}
\end{lemma}

\begin{proof}
By Lemma \ref{relateexp}, 
\begin{align}
&\iota^{(2)}(\exp^{(1)}_xu)=\Phi(\iota^{(1)}(\exp^{(1)}_xu))\label{proof:mainTheorem:Match:eq2}\\
=&\,\Phi\left(\iota^{(1)}(x)+\ud \iota^{(1)}u+\frac{1}{2}\Pi^{(1)}(u,u)+\frac{1}{6}\nabla^{(1)}_u\Pi^{(1)}(u,u)+O(\|u\|_{g^{(1)}}^4)\right)\nonumber\\
=&\,\iota^{(2)}(x)+\nabla\Phi\left(\ud\iota^{(1)}u+\frac{1}{2}\Pi^{(1)}(u,u)+\frac{1}{6}\nabla^{(1)}_u\Pi^{(1)}(u,u) + O(\|u\|_{g^{(1)}}^4)\right) \nonumber\\
&\,+ \frac{1}{2}\nabla^2\Phi\left(u,u\right)+ \nabla^2\Phi\left(u,\Pi^{(1)}(u,u)\right)+O(\|u\|^4)\nonumber\\
=&\,\iota^{(2)}(x)+\nabla\Phi\ud\iota^{(1)}u+Q^{(2)}_2(u)+Q^{(2)}_3(u)+ O(\|u\|_{g^{(1)}}^4),\nonumber
\end{align}
where $\Pi^{(2)}(u,u)=\nabla\Phi\Pi^{(1)}(u,u)+ \nabla^2\Phi\left(u,u\right)$ by the chain rule.
Thus, we have
\begin{align}
&\|\iota^{(2)}(\exp^{(1)}_xu)-\iota^{(2)}(\exp^{(1)}_xv)\|^2\\
=&\, \big\|\nabla\Phi\ud\iota^{(1)}(u-v)\big\|^2\nonumber\\
&\,+2\left\langle \nabla\Phi\ud\iota^{(1)}(u-v), Q^{(2)}_3(u)-Q^{(2)}_3(v)\right\rangle+\|Q^{(2)}_2(u)-Q^{(2)}_2(v)\|^2\nonumber\\
&\,+O(\|u\|_{g^{(1)}}^4\|u-v\|_{g^{(1)}},\|v\|_{g^{(1)}}^4\|u-v\|_{g^{(1)}})\nonumber.
\end{align}
Note that $\|\nabla\Phi\ud\iota^{(1)}(\theta)\|$ is uniformly bounded below from zero and above for any chosen $x\in\MM$. Thus, by taking the square root, we have (\ref{relateexpGeodesic2}).
\end{proof}

Note that in (\ref{relateexpGeodesic1}) the geodesic distance between $y$ and $x$, $\big\|\ud\iota^{(1)}(u)\big\|$, is different from the ambient distance with the error term of order $3$, while in (\ref{relateexpGeodesic2}), the error term is of order $2$. This difference comes from the diffeomorphism $\Phi$, in particular its Hessian $\nabla^2\Phi$ showing up in $Q^{(2)}_2(u)$.

\begin{lemma}\label{Lemma:KeyExpansion}
Fix a symmetric matrix $S\in \RR^{d\times d}$ and $w\in\RR^d$. Then we have the following integrations:
\begin{align}
&\int_{\RR^d}e^{-\|Su\|^2}e^{-\|u-w\|^2}\ud u=\frac{\pi^{d/2}}{\sqrt{\det(I+S^2)}}e^{-\|(I+S^2)^{-1/2}Sw\|^2}
\end{align}
\end{lemma}
\begin{proof}
Note that these integrations could be viewed as a convolution of two Gaussian functions, so we could obtain the results by applying the Fourier transform. 
\end{proof}

We need the following technical lemmas to analyze the effective kernel $\eK$. For $h\geq 0$, denote 
\begin{align}
\tilde{B}^{(1)}_{h}(x):=\exp^{(1)}_x(B^{(1)}_{h}),
\end{align}
where $B^{(1)}_{h}=\{u\in T_x\MM|\,\|u\|_{g^{(1)}}\leq h\}\subset T_x\MM$ is a $d$-dim disk with the center $0$ and the radius $h$. 

\begin{lemma}\label{Lemma:ConvolutionKernelExtraPart}
Suppose Assumptions (A1)-(A4) hold, $F\in L^\infty(\MM)$ and $0<\gamma<1/2$. Then, when $\epsilon$ is small enough, for all pairs of $x,x''\in \MM$ so that $x''=\exp^{(1)}_xv$ and $\|v\|_{g^{(1)}}>2\epsilon^\gamma$, the following holds:
\begin{align}
\left|\int_{\MM} \tilde{K}^{(e_2)}_{\epsilon}(x,y)\tilde{K}^{(e_1)}_{\epsilon}(y,x'')F(y)\ud V^{(1)}(y)\right|=O(\epsilon^{d/2+3/2}).
\end{align}
\end{lemma}
\begin{proof}
By the assumption that $F\in L^\infty(\MM)$, we immediately have
\begin{align}
&\left|\int_{\MM} \tilde{K}^{(2)}_{\epsilon}(x,y)\tilde{K}^{(1)}_{\epsilon}(y,x'')F(x')\ud V^{(1)}(y) \right|\label{proof:lemma1:eq1}\\
\leq \,& \|F\|_{L^\infty}\left|\int_{\tilde{B}^{(1)}_{\epsilon^{\gamma}}(x)}\tilde{K}^{(2)}_{\epsilon}(x,y)\tilde{K}^{(1)}_{\epsilon}(y,x'')F(x')\ud V^{(1)}(y) \right|\nonumber\\
\quad&+\|F\|_{L^\infty}\left|\int_{\MM\backslash\tilde{B}^{(1)}_{\epsilon^{\gamma}}(x)}\tilde{K}^{(2)}_{\epsilon}(x,y)\tilde{K}^{(1)}_{\epsilon}(y,x'')F(x')\ud V^{(1)}(y) \right|.\nonumber
\end{align}
To bound the first integration in (\ref{proof:lemma1:eq1}), note that the assumption $d_1(x,x'')>2\epsilon^\gamma$ implies that $d_1(y,x'')\geq \epsilon^\gamma$ for all $y\in \tilde{B}^{(1)}_{\epsilon^{\gamma}}(x)$. 
Thus, by Lemma \ref{relateexp}, we know that when $\epsilon$ is small enough, $\|\iota^{(1)}(y)-\iota^{(1)}(x'')\|_{\RR^p}\geq (1+c)\epsilon^\gamma$, where $|c|<1$ depends on the second fundamental form of $\iota^{(1)}$. Thus, when $\epsilon$ is small enough, by the exponential decay assumption of $K^{(1)}$, we have
\begin{equation}
\tilde{K}^{(1)}\left(\frac{\|\iota^{(1)}(y)-\iota^{(1)}(x'')\|}{\sqrt{\epsilon}}\right)\leq c_2e^{-c_1(1+c)\epsilon^{2\gamma-1}}=O(\epsilon^{d/2+3/2})
\end{equation}
for all $y\in \tilde{B}^{(1)}_{\epsilon^{\gamma}}(x)$. Since $\tilde{K}^{(2)}\left(\frac{\|\iota^{(2)}(x)-\iota^{(2)}(y)\|}{\sqrt{\epsilon}}\right)\tilde{K}^{(1)}\left(\frac{\|\iota^{(1)}(y)-\iota^{(1)}(x'')\|}{\sqrt{\epsilon}}\right)$ is bounded by $1$, we have
\begin{align}
\left|\int_{\tilde{B}^{(1)}_{\epsilon^{\gamma}}(x)}\tilde{K}^{(2)}_{\epsilon}(x,y)\tilde{K}^{(1)}_{\epsilon}(y,x'')F(x') \ud V^{(1)}(y)\right|=\, O(\epsilon^{d/2+3/2}).
\end{align}
The second integration in (\ref{proof:lemma1:eq1}) could be bounded directly by taking Lemma \ref{Lemma:distortedDistance} and the fact that $\tilde{K}^{(1)}\left(\frac{\|\iota^{(1)}(y)-\iota^{(1)}(x'')\|}{\sqrt{\epsilon}}\right)\leq1$ into account. Denote $y=\exp^{(1)}_xv$ so that $\|v\|_{g^{(1)}}\geq \epsilon^{\gamma}$. 
Note that since $\Phi$ is a diffeomorphism, the smallest singular value of $\nabla\Phi|_{\iota^{(1)}(x)}$ is bounded from below, say $c_0>0$. Thus, when $\|v\|_{g^{(1)}}\geq \epsilon^\gamma$ and $\epsilon$ is small enough, we know 
\begin{align}
&\|\iota^{(2)}(x)-\iota^{(2)}(y)\|^2 \\
=&\,\|\nabla\Phi|_{\iota^{(1)}(x)}v\| ^2+2(\nabla\Phi|_{\iota^{(1)}(x)}v)^T \nabla^2\Phi|_{\iota^{(1)}(x)}(v,v)+O(\|v\|_{g^{(1)}}^4)\nonumber\\
\geq&\, c_0(1+c')\|v\|_{g^{(1)}}^2/2.\nonumber
\end{align} 
Hence, when $\epsilon$ is small enough, there exists $|c'|<1$ so that
\begin{align}
\tilde{K}^{(2)}\left(\frac{\|\iota^{(2)}(x)-\iota^{(2)}(y)\|}{\sqrt{\epsilon}}\right) \leq c_2e^{-c_1c_0(1+c')\epsilon^{2\gamma-1}/2}=O(\epsilon^{d/2+3/2}),
\end{align}
and hence
\begin{align}
&\left| \int_{\MM\backslash\tilde{B}_{1,\epsilon^{\gamma}}(x)}\tilde{K}^{(2)}_{\epsilon}(x,y)\tilde{K}^{(1)}_{\epsilon}(y,x'')F(x')\ud V^{(1)}(y) \right|=O(\epsilon^{d/2+3/2}),
\end{align}
which leads to the proof.
\end{proof}

\begin{lemma}\label{Lemma:ConvolutionKernel}
Suppose Assumptions (A1)-(A4) hold, $F\in C^3(\MM)$ and $0<\gamma<1/2$. Take a pair of $x,x''\in \MM$ so that $x''=\exp_x v$, where $v\in T_x\MM$ and $\|v\|_{g^{(1)}}\leq 2\epsilon^\gamma$. Fix normal coordinates around $x$ associated with $g^{(1)}$ and $g^{(2)}$, set \\$R_x=[\ud\exp^{(2)}_x|_0]^{-1}[\ud \iota^{(2)}]^{-1}\nabla\Phi[\ud\iota^{(1)}][\ud\exp^{(1)}_x|_0]:\RR^d\to \RR^d$ and by the \ac{SVD} $R_x=U_x\Lambda_xV_x^T$, where $U_x,V_x\in O(d)$. Then, when $\epsilon$ is small enough, the following holds:
\begin{align}
&\int_{\MM}\tilde{K}^{(2)}_{\epsilon}(x,y)\tilde{K}^{(1)}_{\epsilon}(y,x'')F(y)\ud V^{(1)}(y)\\
=&\,\epsilon^{d/2}\Big[F(x)A_0(v)+\epsilon A_{2}(F,v) +O(\epsilon^{3/2})\Big]\nonumber,
\end{align}
where 
\begin{align}
A_0(v):=\int_{\RR^d}\tilde{K}^{(2)}\left(\|R_xw\|\right) \tilde{K}^{(1)}\left(\|w-v/\sqrt{\epsilon}\|\right) \ud w,
\end{align}
$A_0(0)>0$, $A_0$ decays exponentially and for $i=1,2$, $A_{2}$ are defined in (\ref{Proof:KeyLemma:DefinitionH2eps}); both $A_0$ and $A_2$ decay exponentially as $\|v\|$ increases. 
Further, $A_0(-v)=A_0(v)$ and $A_{2}(F,-v)=A_{2}(F,v)$. 

When $K^{(1)}$ and $K^{(2)}$ are both Gaussian, that is, $K^{(1)}(t)=K^{(2)}(t)=e^{-t^2}/\sqrt{\pi}$, we have
\begin{align}
A_0(v)=\frac{\pi^{d/2}}{\sqrt{\det(I+\Lambda_x^2)}}e^{-\|(I+\Lambda_x^2)^{-1/2}\Lambda_xV_x^Tv\|^2/\epsilon}.\label{Lemma:Key:Hepsilon:BothAreGaussian}
\end{align}
\end{lemma}

This Lemma essentially says that the effective kernel associated with the \ac{AD} actually still enjoys the exponential decay property of the kernel functions we favor. Further, if the chosen kernels are both Gaussian, then the effective kernel is Gaussian. Note that due to the diffeomorphism, $H_\epsilon(v)$ is not isotropic, no matter if the kernels are Gaussian or not.

We remark more about the decomposition $R_x=U_x\Lambda_xV_x^T$. Note that $\ud\exp^{(i)}_x|_0:T_0T_x\MM\to T_x\MM$, where $i=1,2$, are unitary and we could view $T_0T_x\MM$ as $\RR^d$, so we would not distinguish $T_x\MM$ from $\RR^d$ under the chosen normal coordinates. Thus, we could apply the \ac{SVD} on $R_x$ and obtain $R_x=U_x\Lambda_xV_x^T$, where $U_x,V_x\in O(d)$ and $\Lambda_x$ is a diagonal matrix with non-zero eigenvalues. Clearly, by the diffeomorphism assumption, we know that the smallest eigenvalue of $\Lambda_x$ is uniformly bounded for $x\in\MM$ away from $0$ and from above. 

\begin{proof}
Since $\MM$ is compact and $K$ is positive, we have
\begin{align}
&\left|\int_{\MM}\tilde{K}^{(2)}_{\epsilon}(x,y)\tilde{K}^{(1)}_{\epsilon}(y,x'')F(y) \ud V^{(1)}(y)\right. \label{proof:lemma2:eq1}\\
&\left.\qquad\qquad-\int_{\tilde{B}^{(1)}_{\epsilon^{\gamma}}(x)}\tilde{K}^{(2)}_{\epsilon}(x,y)\tilde{K}^{(1)}_{\epsilon}(y,x'')F(y) \ud V^{(1)}(y)\right|\nonumber\\
\leq\,&\|F\|_{L^\infty}\left|\int_{\MM\backslash\tilde{B}^{(1)}_{\epsilon^{\gamma}}(x)}\tilde{K}^{(2)}_{\epsilon}(x,y)\tilde{K}^{(1)}_{\epsilon}(y,x'') \ud V^{(1)}(y)\right|\nonumber.
\end{align}
Since $K^{(1)}_{\epsilon}(y,x'')\leq 1$, by the same bound as that of the second integration of (\ref{proof:lemma1:eq1}), (\ref{proof:lemma2:eq1}) is bounded by $\epsilon^{d/2+3/2}$ when $\epsilon$ is small enough.
Next, we handle the main term we have interest: 
\begin{align}
&\int_{\tilde{B}^{(1)}_{\epsilon^{\gamma}}(x)}\tilde{K}^{(2)}_{\epsilon}(x,y)\tilde{K}^{(1)}_{\epsilon}(y,x'') F(y) \ud V^{(1)}(y)\label{Proof:KeyEstimation:1}\\
=\,&\int_{\tilde{B}^{(1)}_{\epsilon^{\gamma}}(x)}\tilde{K}^{(2)}\left(\frac{\|\iota^{(2)}(x)-\iota^{(2)}(y)\|}{\sqrt{\epsilon}}\right)\tilde{K}^{(1)}\left(\frac{\|\iota^{(1)}(y)-\iota^{(1)}(x'')\|}{\sqrt{\epsilon}}\right)F(y) \ud V^{(1)}(y).\nonumber
\end{align}

Suppose $y=\exp_x^{(1)}u$, where $u\in T_x\MM$ with $\|u\|_{g^{(1)}}\ll 1$ and $t\geq 0$. By (\ref{relateexpGeodesic2}) in Lemma \ref{Lemma:distortedDistance}, we have
\begin{align}
\|\iota^{(2)}(x)-\iota^{(2)}(y)\|=&\,  \|\nabla\Phi\ud\iota^{(1)}(u)\|+\tilde{Q}^{(2)}_3(u)+O(\|u\|^4),
\end{align}
where 
\begin{align}
&\tilde{Q}^{(2)}_3(u):=\frac{2\left\langle \nabla\Phi\ud\iota^{(1)}(u), Q^{(2)}_3(u)\right\rangle+\|Q^{(2)}_2(u)\|^2}{\|\nabla\Phi\ud\iota^{(1)}(u)\|}.
\end{align}
Note that $\tilde{Q}^{(2)}_2$ is an odd function and $\tilde{Q}^{(2)}_3$ is an even function.
Thus, by the Taylor expansion, the first kernel function in (\ref{Proof:KeyEstimation:1}) becomes
\begin{align}
&\tilde{K}^{(2)}\left(\frac{\|\iota^{(2)}(x)-\iota^{(2)}(y)\|}{\sqrt{\epsilon}}\right)\\
=\,&\tilde{K}^{(2)}\left(\frac{\|\nabla\Phi\ud\iota^{(1)}(u)\|}{\sqrt{\epsilon}}\right)+[\tilde{K}^{(2)}]'\left(\frac{\|\nabla\Phi\ud\iota^{(1)}(u)\|}{\sqrt{\epsilon}}\right)\frac{\tilde{Q}^{(2)}_3(u)}{\sqrt{\epsilon}}+O\left(\frac{\|u\|^5}{\epsilon}\right).\nonumber
\end{align}
Similarly, by applying (\ref{relateexp1Diff}) in Lemma \ref{relateexp}, when $x''=\exp^{(1)}_x(v)$, where $v\in T_x\MM$ with $\|v\|_{g^{(1)}}\ll1$, we have
\begin{align}
\|\iota^{(1)}(y)-\iota^{(1)}(x'')\|=\|\ud\iota^{(1)}(u-v)\|+\tilde{Q}^{(1)}_3(u,v)+O(\|u\|^4,\|v\|^4).
\end{align}
where
\begin{align}
\tilde{Q}^{(1)}_3(u,v)&=\frac{\|Q^{(1)}_2(u)-Q^{(1)}_2(v)\|^2+2\langle\ud \iota^{(1)}(u-v),Q_3^{(1)}(u)-Q_3^{(1)}(v)\rangle}{2\|\ud\iota^{(1)}(u-v)\|}.
\end{align}
Note that $\tilde{Q}^{(1)}_3(-u,-v)=\tilde{Q}^{(1)}_3(u,v)$.
Thus, by the Taylor expansion, the second kernel function in (\ref{Proof:KeyEstimation:1}) becomes
\begin{align}
\tilde{K}^{(1)}&\left(\frac{\|\iota^{(1)}(y)-\iota^{(1)}(x'')\|}{\sqrt{\epsilon}}\right)=\,\tilde{K}^{(1)}\left(\frac{\|\ud\iota^{(1)}(u-v)\|}{\sqrt{\epsilon}}\right)\\
&\qquad+[\tilde{K}^{(1)}]'\left(\frac{\|\ud\iota^{(1)}(u-v)\|}{\sqrt{\epsilon}}\right)\frac{\tilde{Q}^{(1)}_3(u,v)}{\sqrt{\epsilon}}+O\left(\frac{\|u\|^4}{\sqrt{\epsilon}},\frac{\|v\|^4}{\sqrt{\epsilon}}\right).\nonumber
\end{align}
By the same argument as that of Lemma \ref{Lemma:ConvolutionKernelExtraPart}, we could replace the integral domain, $B^{(1)}_{\epsilon^{\gamma}}\subset T_x\MM$, in (\ref{Proof:KeyEstimation:1}) by $T_x\MM$ with the error of order $O(\epsilon^{d/2+3/2})$, as $\epsilon$ is small enough. As a result, with the Taylor expansion of $F$, (\ref{Proof:KeyEstimation:1}) becomes
\begin{align}
&\int_{\RR^d}\Big[\tilde{K}^{(2)}\left(\frac{\|R_xu\|}{\sqrt{\epsilon}}\right)+[\tilde{K}^{(2)}]'\left(\frac{\|R_xu\|}{\sqrt{\epsilon}}\right)\frac{\tilde{Q}^{(2)}_3(u)}{\sqrt{\epsilon}}+O\left(\frac{\|u\|^5}{\epsilon}\right)\Big]\label{Proof:KeyTheorem:Expansion0}\\
&\quad\times \Big[\tilde{K}^{(1)}\left(\frac{\|u-v\|}{\sqrt{\epsilon}}\right)+[\tilde{K}^{(1)}]'\left(\frac{\|u-v\|}{\sqrt{\epsilon}}\right)\frac{\tilde{Q}^{(1)}_3(u,v)}{\sqrt{\epsilon}}+O\left(\frac{\|u\|^4}{\sqrt{\epsilon}},\frac{\|v\|^4}{\sqrt{\epsilon}}\right) \Big]\nonumber\\
&\quad\times  \big[F(x)+\nabla^{(1)}_u F(x)+\frac{{\nabla_{u,u}^{(1)}}^2F(x)}{2}+O(\|u\|^3)\big]\nonumber\\
&\quad\times \big[1-\Ric^{(1)}_{ij}(x)u^iu^j+O(\|u\|^3)\big]\ud u+O(\epsilon^{d/2+3/2})\nonumber\\
=\,&\epsilon^{d/2}\big[F(x)A_{0,\epsilon}(v)+\epsilon^{1/2}A_{1,\epsilon}(F,v)+\epsilon A_{2,\epsilon}(F,v)+O(\epsilon^{3/2})\big],\nonumber
\end{align}
where 
\begin{align}
A_{0,\epsilon}(v)&:=\epsilon^{-d/2}\int_{\RR^d}\tilde{K}^{(2)}\left(\frac{\|R_xu\|}{\sqrt{\epsilon}}\right) \tilde{K}^{(1)}\left(\frac{\|u-v\|}{\sqrt{\epsilon}}\right) \ud u\nonumber\\
A_{1,\epsilon}(F,v)&:=\epsilon^{-d/2-1/2}\int_{\RR^d}\tilde{K}^{(2)}\left(\frac{\|R_xu\|}{\sqrt{\epsilon}}\right)\tilde{K}^{(1)}\left(\frac{\|u-v\|}{\sqrt{\epsilon}}\right)\nabla^{(1)}_u F(x)\ud u\\
A_{2,\epsilon}(F,v)&:=\epsilon^{-d/2-1}[A_{21}(F,v)+A_{22}(F,v)+A_{23}(F,v)+A_{24}(F,v)]\nonumber
\end{align}
and
\begin{align}
A_{21}(F,v)&:=F(x)\int_{\RR^d} \Big[[\tilde{K}^{(2)}]'\left(\frac{\|R_xu\|}{\sqrt{\epsilon}}\right)\frac{\tilde{Q}^{(2)}_3(u)}{\sqrt{\epsilon}}\tilde{K}^{(1)}\left(\frac{\|u-v\|}{\sqrt{\epsilon}}\right) \ud u \nonumber\\
A_{22}(F,v)&:=F(x)\int_{\RR^d} \tilde{K}^{(2)}\left(\frac{\|R_xu\|}{\sqrt{\epsilon}}\right) [\tilde{K}^{(1)}]'\left(\frac{\|u-v\|}{\sqrt{\epsilon}}\right)\frac{\tilde{Q}^{(1)}_3(u,v)}{\sqrt{\epsilon}} \ud u\\
A_{23}(F,v)&:=\int_{\RR^d}\tilde{K}^{(2)}\left(\frac{\|R_xu\|}{\sqrt{\epsilon}}\right) \tilde{K}^{(1)}\left(\frac{\|u-v\|}{\sqrt{\epsilon}}\right)  {\nabla_{u,u}^{(1)}}^2F(x)\ud u\nonumber\\
A_{24}(F,v)&:=F(x)\int_{\RR^d}\tilde{K}^{(2)}\left(\frac{\|R_xu\|}{\sqrt{\epsilon}}\right) \tilde{K}^{(1)}\left(\frac{\|u-v\|}{\sqrt{\epsilon}}\right)  \Ric^{(1)}_{ij}(x)u^iu^j\ud u\nonumber.
\end{align}
Here, we sort the terms according to the order of $\epsilon$, and we claim that  $A_i(F,v)$ is of order $O(1)$, for $i=0,1,2$.
By a change of variable $w=u/\sqrt{\epsilon}$, we have
\begin{align}
A_{0,\epsilon}(v)=\int_{\RR^d}\tilde{K}^{(2)}\left(\|R_xw\|\right) \tilde{K}^{(1)}\left(\|w-v/\sqrt{\epsilon}\|\right) \ud w.
\end{align}
By the assumption of the kernels, $H(0)>0$ and $H$ decays exponentially.
Similarly, we have
\begin{align}
A_{1,\epsilon}(F,v)=&\int_{\RR^d}\tilde{K}^{(2)}\left(\|R_xu\|\right)\tilde{K}^{(1)}\left(\|u-v/\sqrt{\epsilon}\|\right)\nabla^{(1)}_u F(x)\ud u\label{Proof:KeyLemma:DefinitionH1eps}
\end{align}
and
\begin{align}
A_{2,\epsilon}(F,v)=&F(x)\int_{\RR^d} [\tilde{K}^{(2)}]'\left(\|R_xu\|\right) \tilde{K}^{(1)}\left(\|u-v/\sqrt{\epsilon}\|\right) \tilde{Q}^{(2)}_3(u)\ud u\label{Proof:KeyLemma:DefinitionH2eps}\\
&+F(x)\int_{\RR^d} \tilde{K}^{(2)}\left(\|R_xu\|\right) [\tilde{K}^{(1)}]'\left(\|u-v/\sqrt{\epsilon}\|\right)\tilde{Q}^{(1)}_3(u,v/\sqrt{\epsilon})\ud u\nonumber \\
&+\int_{\RR^d}\tilde{K}^{(2)}\left(\|R_xu\|\right) \tilde{K}^{(1)}\left(\|u-v/\sqrt{\epsilon}\|\right)  {\nabla_{u,u}^{(1)}}^2F(x)\ud u\nonumber\\
&+F(x)\int_{\RR^d}\tilde{K}^{(2)}\left(\|R_xu\|\right) \tilde{K}^{(1)}\left(\|u-v/\sqrt{\epsilon}\|\right)  \Ric^{(1)}_{ij}(x)u^iu^j\ud u\nonumber
\end{align} 
Note that since $\tilde{K}^{(1)}$ and $\tilde{K}^{(2)}$ both decay exponentially fast, we know that $A_{i,\epsilon}(F,v)$ also decay exponentially fast as $\|v\|$ increases. 

By the symmetric properties of $\tilde{Q}^{(1)}_3(u,v)$ and $\tilde{Q}^{(2)}_3(u)$, $A_{1,\epsilon}$ is anti-symmetric associated with $v$ and $A_{2,\epsilon}(F,v)$ is symmetric associated with $v$; that is, $A_{1,\epsilon}(F, -v)=-A_{1,\epsilon}(F, v)$ and $A_{2,\epsilon}(F, -v)=A_{2,\epsilon}(F, v)$.

Finally, when both kernels are Gaussian, that is, $\tilde{K}^{(1)}(t)=\tilde{K}^{(2)}(t)=e^{-t^2}/\sqrt{\pi}$, by the fact that $R_x=U_x\Lambda_xV_x^T$ and Lemma \ref{Lemma:KeyExpansion}, the leading order term in (\ref{Proof:KeyTheorem:Expansion0}) 
\begin{align}
A_{0,\epsilon}(v)\,&=\int_{\RR^d}e^{-\|\Lambda_x u\|^2}e^{-\|u-V_x^Tv/\sqrt{\epsilon}\|^2}\ud u\\
&\,=\frac{1}{\sqrt{\det(I+\Lambda_x^2)}}e^{-\|(I+\Lambda_x^2)^{-1/2}\Lambda_xV_x^Tv\|^2/\epsilon}.\nonumber
\end{align}
\end{proof}

With Lemma \ref{Lemma:ConvolutionKernelExtraPart} and Lemma \ref{Lemma:ConvolutionKernel}, we can finish the study of the effective kernel $\eK$. Before proving Theorem \ref{Theorem:ConvolutionKernel}, we mention that in the proof, we focus on how the \ac{AD} process behaves on $(\MM,g^{(1)})$ by converting most ``measurements'' on $(\MM,g^{(2)})$ back to $(\MM,g^{(1)})$. To do so, note that by Assumption 3.1 (A4), $\ud\mu_{\MM}=p_1\ud V^{(1)}=p_2\ud V^{(2)}$, (\ref{Definition:diffusionManifoldCaseT}) becomes
\begin{align}
Tf(x)&\,=\frac{\int_{\MM}\tilde{K}_\epsilon^{(e)}(x,x'')f(x'')p_1(x'')\ud V^{(1)}(x'')}{\int_{\MM}\tilde{K}_\epsilon^{(2)}(x,\bar{x})p_2(\bar{x})\ud V^{(2)}(\bar{x}) }\\
&\,=\frac{\int_{\MM}\tilde{K}_\epsilon^{(e)}(x,x'')f(x'')p_1(x'')\ud V^{(1)}(x'')}{\int_{\MM}\tilde{K}_\epsilon^{(2)}(x,\bar{x})p_1(\bar{x})\ud V^{(1)}(\bar{x}) } \nonumber
\end{align}
and (\ref{Definition:effectiveKernelManifoldCase}) becomes
\begin{align}
\tilde{K}_\epsilon^{(e)}(x,x'')&\,=\int_{\MM} 
\frac{\tilde{K}_\epsilon^{(2)}(x,x') \tilde{K}_\epsilon^{(1)}(x',x'')}{\int_{\MM}\tilde{K}_\epsilon^{(1)}(x',\bar{x})p_1(\bar{x})\ud V^{(1)}(\bar{x})}p_2(x')\ud V^{(2)}(x')\\
&\,=\int_{\MM} 
\frac{\tilde{K}_\epsilon^{(2)}(x,x') \tilde{K}_\epsilon^{(1)}(x',x'')}{\int_{\MM}\tilde{K}_\epsilon^{(1)}(x',\bar{x})p_1(\bar{x})\ud V^{(1)}(\bar{x})} p_1(x')\ud V^{(1)}(x').\label{Definition:effectiveKernelManifoldCase2}
\end{align}
As we will see shortly, although the \ac{AD} works on two different metrics, most quantities could be converted to $(\MM,g^{(1)})$, except the estimated distance among two points.

\begin{proof}[Proof of Theorem \ref{Theorem:ConvolutionKernel}]
To finish the proof, we study (\ref{Definition:effectiveKernelManifoldCase2}). We focus on $x''=\exp^{(1)}_xv$ so that $\|v\|_{g^{(1)}}\leq 2\epsilon^\gamma$. By the absolute continuity assumption of $\mu_\MM$, we have 
\begin{align}
\tilde{p}_{1,\epsilon}(x'):=\epsilon^{-d/2}\int_{\MM} \tilde{K}^{(1)}_\epsilon(x',x'')p_1(x'')\ud V^{(1)}(x''),
\end{align}
which by the assumption that $p_1\in C^4(\MM)$ and a direct expansion becomes
\begin{align}
\tilde{p}_{1,\epsilon}(x')=p_1(x')+\epsilon Q(x')+O(\epsilon^{2}),
\end{align}
where 
\begin{align}
Q(x'):= \frac{\mu^{(1)}_{2,0}}{2}\Big[s^{(1)}(x')+\frac{\int_{S^{d-1}}\|\Pi^{(1)}(\theta,\theta)\|\ud \theta}{24|S^{d-1}|} p_1(x')-\Delta^{(1)}p_1(x')\Big],
\end{align}
where $s^{(1)}(x')$ is the scalar curvature of $g^{(1)}$ at $x'$.
See, for example, \cite[Lemma B.10]{Singer_Wu:2012} for a proof. Thus, we have
\begin{align}
\frac{p_1(x')}{\tilde{p}_{1,\epsilon}(x')}=1-\epsilon Q(x')+O(\epsilon^{3/2}).\label{Proof:MainTheorem1:Formula1}
\end{align}
Since $g^{(1)}$ and $p_1$ are both $C^4$, we know that $Q\in C^2(\MM)$. Hence, by (\ref{Proof:MainTheorem1:Formula1}) we have
\begin{align}
A_{1,\epsilon}\Big(\frac{p_1}{\tilde{p}_{1,\epsilon}},v\Big)=A_{1,\epsilon}(1,v)+O(\epsilon)=O(\epsilon)
\end{align}
since $A_{1,\epsilon}(1,v)=0$, and
\begin{align}
A_{2,\epsilon}\Big(\frac{p_1}{\tilde{p}_{1,\epsilon}},v\Big)=A_{2,\epsilon}(1,v)+O(\epsilon).
\end{align}
By (\ref{Proof:MainTheorem1:Formula1}), the assumptions and Lemma \ref{Lemma:ConvolutionKernel}, the effective kernel becomes
\begin{align}
&\tilde{K}^{(e)}_\epsilon(x,x'')\\
=\,&\int_{\MM} \frac{\epsilon^{-d/2}\tilde{K}^{(2)}_\epsilon(x,x') \tilde{K}^{(1)}_\epsilon(x',x'')}{\epsilon^{-d/2}\int_{\MM} \tilde{K}^{(1)}_\epsilon(x',x'')p_1(x'')\ud V^{(1)}(x'')}p_1(x')\ud V^{(1)}(x')\nonumber\\
=\,&\epsilon^{-d/2}\int_{\MM} \tilde{K}^{(2)}_\epsilon(x,x') \tilde{K}^{(1)}_\epsilon(x',x'')\frac{p_1(x')}{\tilde{p}_{1,\epsilon}(x')}\ud V^{(1)}(x')\nonumber\\
=\,&A_{0,\epsilon}(v)+\epsilon(A_{2,\epsilon}(1,v)+Q(x)A_{0,\epsilon}(v))+O(\epsilon^{3/2})\nonumber,
\end{align}
\end{proof}

While the proof of Theorem \ref{Theorem:MainTheorem2} becomes standard once the effective kernel is known, we provide the proof for the sake of self-containedness and to show the net diffusion outcome on $(\MM,g^{(1)})$. We need the following Lemma, which describes the asymptotical behavior of a kernel not in the normalization form.
\begin{lemma}\label{KDE}
Suppose Assumption (A1)-(A4) hold. Fix $x\in \MM$ and pick $F \in C^3(\MM)$. Fix normal coordinates around $x$ associated with $g^{(1)}$ and $g^{(2)}$ so that $\{E_i\}_{i=1}^d\subset T_x\MM$ is orthonormal associated with $g^{(1)}$. Set $R_x=[\ud\exp^{(2)}_x|_0]^{-1}[\ud \iota^{(2)}]^{-1}\nabla\Phi[\ud\iota^{(1)}][\ud\exp^{(1)}_x|_0]$ and by the \ac{SVD} $R_x=U_x\Lambda_xV_x^T$, where $\Lambda_x=\text{diag}[\lambda_1,\ldots,\lambda_d]$.  Then, when $\epsilon$ is small enough, we have
\begin{align}
&\int_{\MM}\epsilon^{-d/2}\tilde{K}^{(e)}_\epsilon(x,x'')F(x'')\ud V^{(1)}(x'')=\frac{F(x)}{\det(\Lambda_x)}+\epsilon C_2(F)+O(\epsilon^{3/2})\label{Lemma:KeyLemma:KernelDiffusion}
\end{align}
where $C_2(F)$ is defined in (\ref{Proof:MainLemma:Definite:C2}). 
\end{lemma}
\begin{proof}
We focus on $x''=\exp^{(1)}_xv$ so that $\|v\|_{g^{(1)}}\leq 2\epsilon^\gamma$. 
Note that by Lemma (\ref{Lemma:ConvolutionKernelExtraPart}), we will ignore the integration over $\MM\backslash\tilde{B}_{\epsilon^{\gamma}}(x)$ with the error of order $\epsilon^{d/2+3/2}$. 
By the exponential decay of $A_{0,\epsilon}(v)$ and $A_{2,\epsilon}(v)$ and the finite volume assumption of $\MM$, when $\epsilon$ is small enough, by Theorem \ref{Theorem:ConvolutionKernel}, (\ref{Lemma:KeyLemma:KernelDiffusion}) becomes
\begin{align}
&\int_{\MM}\epsilon^{-d/2}\tilde{K}^{(e)}_\epsilon(x,x'')F(x'')\ud V^{(1)}(x'')\label{Proof:Lemma:Numerator1}\\
=&\,\int_{\tilde{B}_{\epsilon^{\gamma}}(x)} \epsilon^{-d/2}\tilde{K}^{(e)}_\epsilon(x,x'')F(x'')\ud V^{(1)}(x'')+O(\epsilon^{3/2})\nonumber\\
=&\,\epsilon^{-d/2}\int_{B_{\epsilon^{\gamma}}} \left[A_{0,\epsilon}(v)+\epsilon(A_{2,\epsilon}(1,v)+Q(x)A_{0,\epsilon}(v))\right] \nonumber\\
&\,\qquad\times \big[F(x)+\nabla^{(1)}_v F(x)+\frac{1}{2}{\nabla_{v,v}^{(1)}}^2F(x)) \big]\nonumber\\
&\,\qquad\times \big[1+\Ric^{(1)}_{ij}v^iv^j\big]\ud v+O(\epsilon^{3/2}).\nonumber
\end{align}
Expand the multiplication and sort terms in the increasing order, we have to evaluate the following integrations to obtain (\ref{Proof:Lemma:Numerator1}):
\begin{align} 
&\epsilon^{-d/2}F(x)\int_{B_{\epsilon^{\gamma}}}A_{0,\epsilon}(v)\ud v\label{Proof:Lemma:Numerator:expansion1}\\
&\epsilon^{-d/2+1/2}\int_{B_{\epsilon^{\gamma}}} A_{0,\epsilon}(v)\nabla^{(1)}_v F(x) \ud  v\label{Proof:Lemma:Numerator:expansion3}\\
&\epsilon^{-d/2}\int_{B_{\epsilon^{\gamma}}} A_{0,\epsilon}(v)\big(\frac{1}{2}{\nabla_{v,v}^{(1)}}^2F(x)+F(x)\Ric^{(1)}_{ij}v^iv^j\big) \ud  v\label{Proof:Lemma:Numerator:expansion7}\\
&\qquad+\epsilon^{-d/2+1} F(x)\int_{B_{\epsilon^{\gamma}}} (A_{2,\epsilon}(v)+Q(x)A_{0,\epsilon}(v)) \ud  v.\nonumber
\end{align}
Here, (\ref{Proof:Lemma:Numerator:expansion1}) is the $0$-th order term, (\ref{Proof:Lemma:Numerator:expansion3}) is the $1$-st order term and (\ref{Proof:Lemma:Numerator:expansion7}) is the $2$-nd order term. By the same argument as the above, we could replace the integral domain $B_{\epsilon^{\gamma}}$ by $\RR^d$ in (\ref{Proof:Lemma:Numerator:expansion1}), (\ref{Proof:Lemma:Numerator:expansion3}) and (\ref{Proof:Lemma:Numerator:expansion7}), with a higher order error. By plugging $A_0(v)$ and changing the integration order, (\ref{Proof:Lemma:Numerator:expansion1}) becomes
\begin{align}
F(x)\int_{\RR^d}\tilde{K}^{(2)}\left(\|R_xw\|\right) \Big[ \int_{\RR^d}\tilde{K}^{(1)}\left(\|w-v\|\right)\ud v\Big] \ud w+O(\epsilon^{3/2}),
\end{align}
which by a direct calculation becomes 
\begin{equation}
\frac{F(x)}{\det(\Lambda_x)}+O(\epsilon^{3/2}).
\label{Proof:MainLemma:Definite:C0}
\end{equation}
Here, we use the assumption that $\mu^{(1)}_{0,0}=\mu^{(2)}_{0,0}=1$, change of variable of $u=R_xw$ and the \ac{SVD} $R_x=U_x\Lambda_xV_x^T$.
The first order term (\ref{Proof:Lemma:Numerator:expansion3}) become $0$ due to the anti-symmetry of $\nabla^{(1)}_v F(x)$ associated with $v$. 

To evaluate (\ref{Proof:Lemma:Numerator:expansion7}), note that by the change of variable we have
\begin{align}
&\epsilon^{-d/2}\int_{B_{\epsilon^{\gamma}}} A_{0,\epsilon}(v)\big(\frac{1}{2}{\nabla_{v,v}^{(1)}}^2F(x)+F(x)\Ric^{(1)}_{ij}v^iv^j\big) \ud  v\label{Proof:MainLemma:Numerator:expansion7:Part1}
\\
=\,&\epsilon\int_{\RR^d} A_{0,1}(v)\big(\frac{1}{2}{\nabla_{v,v}^{(1)}}^2F(x)+F(x)\Ric^{(1)}_{ij}v^iv^j\big) \ud  v+O(\epsilon^{3/2})\nonumber
\end{align}
and
\begin{align}
&\epsilon^{-d/2+1} F(x)\int_{B_{\epsilon^{\gamma}}}(A_{2,\epsilon}(v)+Q(x)A_{0,\epsilon}(v))\ud  v\label{Proof:MainLemma:Numerator:expansion7:Part2}\\
=&\,\epsilon F(x)\int_{\RR^d}(A_{2,1}(v)+Q(x)A_{0,1}(v))\ud  v+O(\epsilon^{3/2}).\nonumber
\end{align}
By noting that ${\nabla_{v,v}^{(1)}}^2F(x)=v^T{\nabla^{(1)}}^2F(x)v$, to simplify (\ref{Proof:MainLemma:Numerator:expansion7:Part1}), we denote $S_x:=\frac{1}{2}{\nabla^{(1)}}^2F(x)+F(x)\Ric^{(1)}(x)$, which is a symmetric matrix. By a change of variable $u=w-v$, we have
\begin{align}
&\int_{\RR^d} A_{0,1}(v)\big(\frac{1}{2}{\nabla_{v,v}^{(1)}}^2F(x)+F(x)\Ric^{(1)}_{ij}v^iv^j\big) \ud  v\\
=\,&\int_{\RR^d}\tilde{K}^{(2)}\left(\|R_xw\|\right) \Big[\int_{\RR^d} \tilde{K}^{(1)}\left(\|w-v\|\right)v^TS_xv \ud  v\Big] \ud w\nonumber\\
=\,&\int_{\RR^d}\tilde{K}^{(2)}\left(\|R_xw\|\right) \Big[\int_{\RR^d} \tilde{K}^{(1)}\left(\|u\|\right)(w-u)^TS_x(w-u) \ud  u\Big] \ud w\nonumber\\
=\,&\int_{\RR^d}\tilde{K}^{(2)}\left(\|R_xw\|\right) \Big[\int_{\RR^d} \tilde{K}^{(1)}\left(\|u\|\right)w^TS_xw \ud  u\Big] \ud w\nonumber\\
&\,\qquad+\int_{\RR^d}\tilde{K}^{(2)}\left(\|R_xw\|\right) \Big[\int_{\RR^d} \tilde{K}^{(1)}\left(\|u\|\right)u^TS_xu\ud  u\Big] \ud w\nonumber\\
=\,&\frac{1}{d}\int_{\RR^d}\tilde{K}^{(2)}\left(\|R_xw\|\right) w^TS_xw\ud w+\frac{1}{d\det(\Lambda_x)}\int_{\RR^d} \tilde{K}^{(1)}\left(\|u\|\right)u^TS_xu\ud  u\nonumber
\end{align}
where the third equality holds since the crossover term 
\begin{align}
\iint\tilde{K}^{(2)}\left(\|R_xw\|\right) \tilde{K}^{(1)}\left(\|u\|\right)u^TSw \ud  u\ud w=0
\end{align}
due to the symmetry of $S^{d-1}$, and the last equality holds due to $\int_{\RR^d}\tilde{K}^{(2)}\left(\|R_xw\|\right)\ud w =\frac{1}{\det(\Lambda_x)}$. Note that 
\begin{align}
&\int_{\RR^d}\tilde{K}^{(2)}\left(\|R_xw\|\right) w^TSw\ud w\\
=&\,\frac{\mu^{(2)}_{2,0}}{d\det(\Lambda_x)}\text{tr}(\Lambda_xV^T_xS_xV_x\Lambda_x)=\frac{\mu^{(2)}_{2,0}}{d\det(\Lambda_x)}\text{tr}(V_x\Lambda_xV^T_xS_xV_x\Lambda_xV_x^T)\nonumber\\
=&\,\frac{\mu^{(2)}_{2,0}}{d\det(\Lambda_x)}\Big(\frac{1}{2}\sum_{i=1}^d\lambda_i\big[{\nabla^{(1)}}^2_{E_i,E_i}F(x)+\Ric^{(1)}_{ii}(x)F(x)\big]\Big)\nonumber
\end{align}
and
\begin{align}
\int_{\RR^d} \tilde{K}^{(1)}\left(\|u\|\right)u^TSu\ud  u=\frac{\mu^{(1)}_{2,0}}{d}\big(\frac{1}{2}\Delta^{(1)}F(x)+s^{(1)}(x)F(x)\big), 
\end{align}
since $\Delta^{(1)}F(x)=\sum_{i=1}^d{\nabla^{(1)}}^2_{E_i,E_i}F(x)$ and $\sum_{i=1}^d\Ric^{(1)}_{ii}(x)=s^{(1)}(x)$.
Hence, we have obtained a simplification of (\ref{Proof:MainLemma:Numerator:expansion7:Part1}).

As a result, by denoting
\begin{align}
C_2(F):=&\frac{\mu^{(2)}_{2,0}\big(\frac{1}{2}\sum_{i=1}^d\lambda_i\big[{\nabla^{(1)}}^2_{E_i,E_i}F(x)+\Ric^{(1)}_{ii}(x)F(x)\big]\big)}{d^2\det(\Lambda_x)}\label{Proof:MainLemma:Definite:C2}\\
&+\frac{\mu^{(1)}_{2,0}\big(\frac{1}{2}\Delta^{(1)}F(x)+s^{(1)}(x)F(x)\big)}{d^2\det(\Lambda_x)}\nonumber\\
&+ F(x)\int_{\RR^d}(A_{2,1}(v)+Q(x)A_{0,1}(v))\ud  v\nonumber,
\end{align}
we conclude that (\ref{Proof:Lemma:Numerator:expansion7}) becomes $\epsilon C_2(F)+O(\epsilon^{3/2})$. Note that we do not simplify $ F(x)\int_{\RR^d}(A_{2,1}(v)+Q(x)A_{0,1}(v))\ud  v$ since this term will be eliminated eventually.

\end{proof}

With the above, we could finish the proof of Theorem \ref{Theorem:MainTheorem2}. We mention that we could see from the proof that the net outcome of the \ac{AD}, while two metrics are involved, eventually could be viewed as a ``deformed'' diffusion solely with the single metric $g^{(1)}$.

\begin{proof}[Proof of Theorem \ref{Theorem:MainTheorem2}]
The proof is based on Lemma \ref{KDE}. The numerator of $D^{(e)}f(x)$ is exactly the same as that stated in Lemma \ref{KDE}; that is,
\begin{align}
&\int_{\MM}\epsilon^{-d/2}\tilde{K}^{(e)}_\epsilon(x,x'')f(x'')p_1(x'')\ud V^{(1)}(x'')\\
=\,&\frac{p_1(x)f(x)}{\det(\Lambda_x)}+\epsilon C_2(fp_1)+O(\epsilon^{3/2}).\nonumber
\end{align}
The denominator of $D^{(e)}f(x)$ is evaluated by plugging $f=1$ into Lemma \ref{KDE}; that is, we have
\begin{align}
\int_{\MM}\epsilon^{-d/2}\tilde{K}^{(e)}_\epsilon(x,x'')p_1(x'')\ud V^{(1)}(x'')=\frac{p_1(x)}{\det(\Lambda_x)}+\epsilon C_2(p_1)+O(\epsilon^{3/2}).
\end{align}
Finally, by the binomial expansion of the denominator, we obtain the result; that is,
\begin{align}
&\frac{\int_{\MM}\tilde{K}^{(e)}_\epsilon(x,x'')f(x'')p_1(x'')\ud V^{(1)}(x'')}{\int_{\MM}\tilde{K}^{(e)}_\epsilon(x,x'')p_1(x'')\ud V^{(1)}(x'')}\\
=\,& f(x)+\epsilon\frac{\det(\Lambda_x)}{p_1(x)}\Big[C_2(fp_1)-f(x)C_2(p_1)\Big]+O(\epsilon^{3/2})\nonumber.
\end{align}
By a direct expansion, we have
\begin{align}
C_2(fp_1)-f(x)C_2(p_1):=&\frac{\mu^{(2)}_{2,0}\sum_{i=1}^d\lambda_i\big[p_1(x){\nabla^{(1)}}^2_{E_i,E_i}f(x)+2\nabla^{(1)}_{E_i}f(x)\nabla^{(1)}_{E_i}p_1(x)\big]}{2d^2\det(\Lambda_x)}\nonumber\\
&+\frac{\mu^{(1)}_{2,0}[p_1(x)\Delta^{(1)}f(x)+2\nabla^{(1)} f(x)\nabla^{(1)} p_1(x)]}{2d^2\det(\Lambda_x)}
\end{align}
and hence the conclusion
\begin{align}
&\frac{\int_{\MM}\tilde{K}^{(e)}_\epsilon(x,x'')f(x'')p_1(x'')\ud V^{(1)}(x'')}{\int_{\MM}\tilde{K}^{(e)}_\epsilon(x,x'')p_1(x'')\ud V^{(1)}(x'')}\\
=\,& f(x)+\frac{\epsilon \mu^{(2)}_{2,0}}{2d^2}\sum_{i=1}^d\lambda_i\big[{\nabla^{(1)}}^2_{E_i,E_i}f(x)+\frac{2\nabla^{(1)}_{E_i}f(x)\nabla^{(1)}_{E_i}p_1(x)}{p_1(x)}\big] \nonumber\\
&\qquad+\frac{\epsilon \mu^{(1)}_{2,0}}{2d^2}\Big[\Delta^{(1)}f(x)+\frac{2\nabla^{(1)} f(x)\nabla^{(1)} p_1(x)}{p_1(x)}\Big]+O(\epsilon^{3/2})\nonumber.
\end{align}
\end{proof}

Finally, we show the proof of Corollary \ref{Theorem:MainCorollary}.
\begin{proof}[Proof of Corollary \ref{Theorem:MainCorollary}]
Combining (\ref{Definition:AD:ExpressionForAnalysisOnM}), Theorem \ref{Theorem:MainTheorem2}, and the assumption of $\nu_{\MM}$, for every fixed $(y,z)$, we have
\begin{align}
Df(x)=&\int_{\MM}\Big[\int_{\MM} P^{(\mathcal{N}_2,z)}_\epsilon(x,x') P^{(\mathcal{N}_1,y)}_\epsilon(x',x'')\ud\nu_{\MM}(x')\Big] \mathcal{E}f(x'')\ud\nu_{\MM}(x'')\\
=&\,\mathcal{E}f(x)+\frac{\epsilon \mu^{(2)}_{2,0}(z)}{2d^2}\Big[\sum_{i=1}^d\lambda_i\big({\nabla^{(1)}}^2_{E_i,E_i}\mathcal{E}f(x)+\frac{2\nabla^{(1)}_{E_i}\mathcal{E}f(x)\nabla^{(1)}_{E_i}p_1(x)}{p_1(x)}\big)  \Big] \nonumber\\
&+\frac{\epsilon \mu^{(1)}_{2,0}(z)}{2d^2}\Big[\Delta^{(1)}\mathcal{E}f(x)+\frac{2\nabla^{(1)} \mathcal{E}f(x)\nabla^{(1)} p_1(x)}{p_1(x)}\Big]+O(\epsilon^{3/2}),\nonumber
\end{align}
where 
\begin{align}
\mu^{(i)}_{2,0}(z)\,=\int_{\mathbb{R}^d}\|y\|^2\tilde{P}^{(O_i,z)}(\|y\|)\ud y.
\end{align}
By the commutativity of $\mathcal{E}$ and $D$, applying the effective operator $\mathcal{E}$ bilaterally yields
\begin{align}
\mathcal{E}Df(x)=\,&\mathcal{E}f(x)+\frac{\epsilon C_2(x)}{2d^2}\Big[\Delta^{(2)}\mathcal{E}f(x)+\frac{2\nabla^{(2)} \mathcal{E}f(x)\nabla^{(2)} p_1(x)}{p_1(x)}\Big] \\
&+\frac{\epsilon C_1(x)}{2d^2}\Big[\Delta^{(1)}\mathcal{E}f(x)+\frac{2\nabla^{(1)} \mathcal{E}f(x)\nabla^{(1)} p_1(x)}{p_1(x)}\Big]+O(\epsilon^{3/2}),\nonumber
\end{align}
where $C_i(x)=\int_{\mathcal{N}_i} \mu^{(i)}_{2,0}(z)\ud\nu_{\mathcal{N}_i|\MM}(z|x)$. To finish the proof, we show that the dependence of $C_i(x)$ on $x$ appears on the higher order term. Denote $\tilde{B}^{(1)}_{h}(x):=\exp^{(1)}_x(B^{(1)}_{h})$,
where $B^{(1)}_{h}=\{u\in T_x\MM|\,\|u\|_{g^{(1)}}\leq h\}\subset T_x\MM$ is a $d$-dim disk with the center $0$ and the radius $h>0$.
Thus, by the same arguments as above, we have
\begin{align}
&\int_{\mathcal{N}_i} \mu^{(i)}_{2,0}(z)\ud\nu_{\mathcal{N}_i|\MM}(z|x)\\
=\,&\int_{\mathcal{N}_i} \int_{\mathbb{R}^d}\|y\|^2\tilde{P}^{(O_i,z)}(\|y\|)\ud y\ud\nu_{\mathcal{N}_i|\MM}(z|x)\nonumber\\
=\,&\int_{\mathcal{N}_i}\int_{\mathcal{N}_i} \int_{\tilde{B}_x(\epsilon^{1/2})}d^2_{g^{(i)}}(y,x)\tilde{P}^{(i)}\left(\sqrt{d^2_{g^{(i)}}(y,x)+d^2_{\mathcal{N}_i}(z,z')}\right)\ud V(y)\ud\nu_{\mathcal{N}_i|\MM}(z'|x)\ud\nu_{\mathcal{N}_i|\MM}(z|x)+O(\epsilon)\nonumber\\
=\,& \int_{\tilde{B}_x(\epsilon^{1/4})}d^2_{g^{(i)}}(y,x) \Big[\int_{\mathcal{N}_i}\int_{\mathcal{N}_i}\tilde{P}^{(i)}\left(\sqrt{d^2_{g^{(i)}}(y,x)+d^2_{\mathcal{N}_i}(z,z')}\right)\ud\nu_{\mathcal{N}_i|\MM}(z'|x)\ud\nu_{\mathcal{N}_i|\MM}(z|x)\Big]\ud V(y)+O(\epsilon)\nonumber.
\end{align}
By denoting 
\begin{align}
\tilde{\tilde{P}}^{(i)}\left(d_{g^{(i)}}(y,x)\right):=\int_{\mathcal{N}_i}\int_{\mathcal{N}_i}\tilde{P}^{(i)}\left(\sqrt{d^2_{g^{(i)}}(y,x)+d^2_{\mathcal{N}_i}(z,z')}\right)\ud\nu_{\mathcal{N}_i|\MM}(z'|x)\ud\nu_{\mathcal{N}_i|\MM}(z|x)\nonumber,
\end{align}
we have
\begin{align}
&\int_{\mathcal{N}_i} \mu^{(i)}_{2,0}(z)\ud\nu_{\mathcal{N}_i|\MM}(z|x)\\
=\,&\int_{\tilde{B}_x(\epsilon^{1/4})}d^2_{g^{(i)}}(y,x) \tilde{\tilde{P}}^{(i)}\left(d_{g^{(i)}}(y,x)\right)\ud V(y)+O(\epsilon)\nonumber\\
=\,&\int_{B_x(\epsilon^{1/4})}\|y\|^2 \tilde{\tilde{P}}^{(i)}\left(\|y\|\right)\ud y+O(\epsilon)\nonumber\\
=\,&\int_{\RR^d}\|y\|^2 \tilde{\tilde{P}}^{(i)}\left(\|y\|\right)\ud y+O(\epsilon)\nonumber.
\end{align}
By denoting $C_i:=\int_{\RR^d}\|y\|^2 \tilde{\tilde{P}}^{(i)}\left(\|y\|\right)\ud y$, we hence obtain the conclusion.
\end{proof}

\end{document}